\documentclass[aps,twocolumn]{revtex4-2}
\usepackage[utf8]{inputenc}
\usepackage{amsmath}
\usepackage{amsthm}
\usepackage{siunitx}
\usepackage[inline]{enumitem}
\usepackage[usenames,dvipsnames]{xcolor}
\makeatletter
\newtheorem*{rep@theorem}{\rep@title}
\newcommand{\newreptheorem}[2]{%
\newenvironment{rep#1}[1]{%
 \def\rep@title{#2 \ref*{##1}, repeated}%
 \begin{rep@theorem}}%
 {\end{rep@theorem}}}
\makeatother

\usepackage{xcolor}  
\usepackage{graphicx}
\usepackage{amsfonts}
\usepackage{booktabs}
\usepackage{xcolor}
\usepackage{braket}
\usepackage{dsfont}

\usepackage{hyperref}
\hypersetup{colorlinks,allcolors=black}

\newtheorem{theorem}{Theorem}
\newreptheorem{theorem}{Theorem}

\newcommand\numeq[1]%
  {\stackrel{\scriptscriptstyle(\mkern-1.5mu#1\mkern-1.5mu)}{=}}
\newcommand\numleq[1]%
  {\stackrel{\scriptscriptstyle(\mkern-1.5mu#1\mkern-1.5mu)}{\leq}}
\newcommand\numgeq[1]%
  {\stackrel{\scriptscriptstyle(\mkern-1.5mu#1\mkern-1.5mu)}{\geq}}
\newcommand\numpreceq[1]%
  {\stackrel{\scriptscriptstyle(\mkern-1.5mu#1\mkern-1.5mu)}{\preceq}}
\newcommand\numsucceq[1]%
  {\stackrel{\scriptscriptstyle(\mkern-1.5mu#1\mkern-1.5mu)}{\succeq}}

\graphicspath{{Images/}}

\begin{document}
\title{Dynamical complexity of non-Gaussian many-body systems with dissipation}
\author{Guillermo González-García$^{1, 2}$}
\author{Alexey V. Gorshkov$^{3,4}$}
\author{J.~Ignacio Cirac$^{1, 2}$}
\author{Rahul Trivedi$^{1, 2}$}
\email{rahul.trivedi@mpq.mpg.de}
\address{$^1$Max-Planck-Institut für Quantenoptik, Hans-Kopfermann-Str.~1, 85748 Garching, Germany \\
$^2$Munich Center for Quantum Science and Technology (MCQST), Schellingstr. 4, D-80799 Munich, Germany \\
$^3$Joint Quantum Institute, NIST/University of Maryland, College Park, Maryland 20742, USA.\\ $^4$Joint Center for Quantum Information and Computer Science, NIST/University of Maryland, College Park, Maryland 20742, USA.}

\date{\today}

\begin{abstract}
We characterize the dynamical state of many-body bosonic and fermionic many-body models with inter-site Gaussian couplings, on-site non-Gaussian interactions and local dissipation comprising incoherent particle loss, particle gain, and dephasing. We first establish that, for fermionic systems, if the dephasing noise is larger than the non-Gaussian interactions, irrespective of the Gaussian coupling strength, the system state is a convex combination of Gaussian states at all times. Furthermore, for bosonic systems, we show that if the particle loss and particle gain rates are larger than the Gaussian inter-site couplings, the system remains in a separable state at all times. Building on this characterization, we establish that at noise rates above a threshold, there exists a classical algorithm that can efficiently sample from the system state of both the fermionic and bosonic models. Finally, we show that,  unlike fermionic systems, bosonic systems can evolve into states that are not convex-Gaussian even when the dissipation is much higher than the on-site non-Gaussianity. Similarly, unlike bosonic systems, fermionic systems can generate entanglement even with noise rates much larger than the inter-site couplings.
\end{abstract}
\maketitle

\raggedbottom

\emph{Introduction}. Whether many-body quantum systems evolve into classically non-trivial states under decoherence is of fundamental interest to the theory of open quantum systems and also has implications for the quantum advantage in quantum computers and simulators \cite{Preskill2018NISQ, daley2022practical_analogue,ebadi2021analog,scholl2021analog,wei2022_analog,Lukin_2021_probing}. Traditionally, this has been mostly studied for many-body spin models, including  extensive recent activity in both the discrete-time setting (i.e.~quantum circuits interspersed with noise) and in the continuous-time setting (modeled by a many-body Lindblad master equation \cite{shtanko2021_fermion_complexity,Trivedi2022_transitions}). For discrete-time models, early results showed that sufficiently high noise suppresses entanglement thus enabling classical simulation \cite{aharonov2000_percolation}. Recent results have shown classical simulability, even with a small amount of depolarizing noise, for both sampling or computing local observables in both random \cite{Aharonov2023_paulipaths,tindall2023efficient,kechedzhi2024effective,shao2023simulating,fontana2023classical_LOWESA,rudolph2023classical_LOWESA,gao2018efficient_simulation,liao2023simulation,gonzalez2022_errors,schuster2024polynomialtime} and structured models \cite{gonzalez2024_paulipaths,rajakumar2024_IQP,debmishra2024bounds}. These results have partly been generalized to the geometrically-local continuous-time setting, which more accurately models analog quantum simulators, to show that the system remains classically simulable when the noise rate is larger than the interaction terms in the Hamiltonian \cite{Trivedi2022_transitions}.

Quantum simulators based on platforms such as ultracold atoms in optical lattices \cite{Yan2022_fermion_simulation,Spar2022_fermion_simulation,Norcia2018_fermion_simulation,bloch_2017_boson_simulation,yang2020simulator_bose_hubbard}, superconducting circuits \cite{zhang2023_superconducting_simulator,shi2024simulator_superconducting} or nonlinear photonics \cite{saxena2023realizing, chang2014quantum, noh2016quantum}, are often described by a family of Hamiltonians that, only in certain regimes, reduce to spin systems. They are modeled by a fermionic or bosonic lattice with two kinds of terms: (i) Gaussian coupling terms which are linear or quadratic in creation/annihilation operators (e.g. particle hopping or pair production); (ii) Non-Gaussian interaction terms that typically act on particles only on one site. Consequently, there are two relevant frequency scales: the strength of the Gaussian couplings, $J$, and that of the on-site (non-Gaussian) interactions, $U$. If $J=0$ or $U=0$, this model is classically simulable. When $J=0$, the Hamiltonian is a sum of single-site terms which maps product states to product states which can be classically simulated. When $U=0$, the dynamics is Gaussian and thus local observables can be efficiently computed and, for fermions, even sampling is efficient \cite{Koenig_2012_simulation_fermions,bartlett2002_gaussian_cv,Terhal2002_classical_simulation_fermions, valiant2001quantum}. However, when both $J, U \neq 0$ and the system is noiseless, this model is universal for quantum computation \cite{BRAVYI2002_fermionicqc,LLoyd_1999_continuousvariables} and thus worst-case hard to simulate classically. While the presence of dissipation should make the model classically simulable, the amount and type of local dissipation needed remains unclear.

For fermionic systems, the impact of noise on non-Gaussianity was studied in a circuit model with Gaussian gates and non-Gaussian ancillas, which showed that the state remains a convex combination of Gaussian states above a noise-threshold\cite{Terhal_2013_noisy_fermionic,oszmaniec2014classical,Bravyi_2006_majorana}. Studies analyzing continuous-time dynamics have focused on the non-Gaussianity introduced via two-body dissipation \cite{shtanko2021_fermion_complexity}. For bosonic systems, previous studies have either focused on understanding their complexity as a function of evolution time in the absence of noise \cite{maskara2022_bosons_phase_diagram,Deshpande_2018_bosons_complexity,Muraleedharan_2019_bosons}, or for the specific task of boson sampling in the presence of noise \cite{oh2023_noisy_boson_sampling,Qi2020_noisy_boson_sampling,Shchesnovich_2021_noisy_boson, pan_2020_boson_sampling_exp,madsen2022boson_sampling_supremacy}. However, the classical simulability of the noisy continuous-time model motivated above remains unresolved.

In this Letter, we rigorously address this question---we consider fermionic and bosonic systems with $n$ sites, each containing locally $L$ modes (Fig. \ref{fig:diagram}). The annihilation operators corresponding to the $\sigma^\text{th}$ mode at the $i^\text{th}$ site, where $\sigma \in \{1, 2 \dots L\}$ and $i \in \{1, 2 \dots n\}$, is given by  $a_{i,\sigma}$. We will use the Hermitian operators $c_{i,\sigma}^{\alpha}$, with $\alpha \in \{1,2\}$, defined as
$c_{i,\sigma}^1=(a_{i,\sigma}+a_{i,\sigma}^{\dagger})/\sqrt{2},c_{i,\sigma}^2=-i(a_{i,\sigma}-a_{i,\sigma}^{\dagger})/\sqrt{2}$, which represent either Majorana operators (for fermions) or position and momentum quadrature operators (for bosons). The noisy dynamics is described by the Lindblad master equation

\begin{align}\label{eq:lindblad}
 \frac{d \rho(t)}{dt}=-i[H(t),\rho(t)] + \sum_{i, \sigma} \mathcal{L}_{i,\sigma}\rho(t),
\end{align}
were $H(t)$ is a (possibly time-dependent) Hamiltonian describing the system. $\mathcal{L}_{i,\sigma}$ captures the noise, which is assumed to act locally on every mode $(i,\sigma)$, and is modeled by
\begin{align}\label{eq:noise_lindblad_fermions}
\mathcal{L}_{i,\sigma} (\cdot) = \sum_{l=1}^3\kappa_l \bigg(L_{i, \sigma}^{(l)} (\cdot)L^{(l)\dagger}_{i, \sigma} -\frac{1}{2}\{L_{i, \sigma}^{(l)\dagger}L_{i, \sigma}^{(l)},(\cdot)\}\bigg),
\end{align}
with jump operators $L_{i, \sigma}^{(1)}=a_{i,\sigma},L_{i, \sigma}^{(2)}=a_{i,\sigma}^{\dagger},L_{i, \sigma}^{(3)}=a_{i,\sigma}^{\dagger}a_{i,\sigma}$ and decay rates $\kappa_1, \kappa_2, \kappa_3$ respectively. The jump operator $a_{i,\sigma}$ models particle loss, $a_{i,\sigma}^{\dagger}$ models incoherent particle gain, and $a_{i,\sigma}^{\dagger}a_{i,\sigma}$ models dephasing. We remark that our conclusions also hold for other physically relevant dissipators such as $L_{i, \sigma}^{(1)}=c_{i,\sigma}^1$, $L_{i, \sigma}^{(2)}=c_{i,\sigma}^2$ which, in the bosonic case, would correspond to white noise fluctuations in the quadratures. We use the same noise Lindbladian for both bosons and fermions---for the bosonic case, we will additionally assume that the particle loss occurs at a rate strictly higher than (both coherent and incoherent) particle gain (see Supplement \cite{supplemental_material} for the exact assumption) so as to avoid an unbounded growth of the number of particles with $t$ which would be unphysical in an actual experiment.

We will assume that the Hamiltonian can be written as $H(t) = H_\text{g}(t) + H_\text{ng}(t)$, where $H_\text{g}(t)$ contains Gaussian general intersite terms:
\begin{subequations}\label{eq:Hamiltonian}
\begin{align}\label{eq:Gaussian}
 &H_\text{g}(t)=\sum_{i, j}\sum_{\substack{\alpha,\alpha^{\prime} \\ \sigma,\sigma^{\prime}}}
 J_{i,\sigma;j,\sigma^{\prime}}^{\alpha,\alpha^{\prime}}(t) c_{i,\sigma}^{\alpha}c_{j,\sigma^{\prime}}^{\alpha^{\prime}}+\sum_{i,\alpha,\sigma}\Omega_{i;\sigma}^{\alpha}(t)c_{i,\sigma}^{\alpha},
\end{align}

and $H_\text{ng}(t)$ contains on-site non-Gaussian terms which account for particle-particle repulsion and attraction between different modes at the same site:
\begin{align}\label{eq:non_Gaussian}
    H_\text{ng}(t) = \sum_{i,\sigma,\sigma^{\prime}} U_{i,\sigma;i,\sigma^{\prime}}(t)n_{i,\sigma}n_{i,\sigma^{\prime}}.
\end{align}
\end{subequations}
Note that in the fermionic case $\Omega_{i;\sigma}^{\alpha}(t)=0$, since physical Hamiltonians must preserve fermionic parity, while in the bosonic case $\Omega_{i;\sigma}^{\alpha}(t)$ can be a non-zero real scalar. For fermions, we can assume that $J_{i,\sigma;j,\sigma^{\prime}}^{\alpha,\alpha'}(t)$ is purely imaginary and anti-symmetric i.e. 
\begin{align}
J_{i,\sigma;j,\sigma^{\prime}}^{\alpha,\alpha'}(t)= \big(J_{i,\sigma;j,\sigma^{\prime}}^{\alpha,\alpha'}(t)\big)^* = -J_{j,\sigma^{\prime};i,\sigma}^{\alpha',\alpha}(t),
\end{align}
while for bosons it can be assumed to be purely real and symmetric: 
\begin{align}
J_{i,\sigma;j,\sigma^{\prime}}^{\alpha,\alpha^{\prime}}(t)=J_{j,\sigma^{\prime};i,\sigma}^{\alpha^{\prime},\alpha}(t).
\end{align}
The non-Gaussian onsite interactions $U_{i,\sigma; i, \sigma'}(t)$ can be assumed to be real and symmetric for both fermions and bosons.

\begin{figure}
    \includegraphics[scale=0.27]{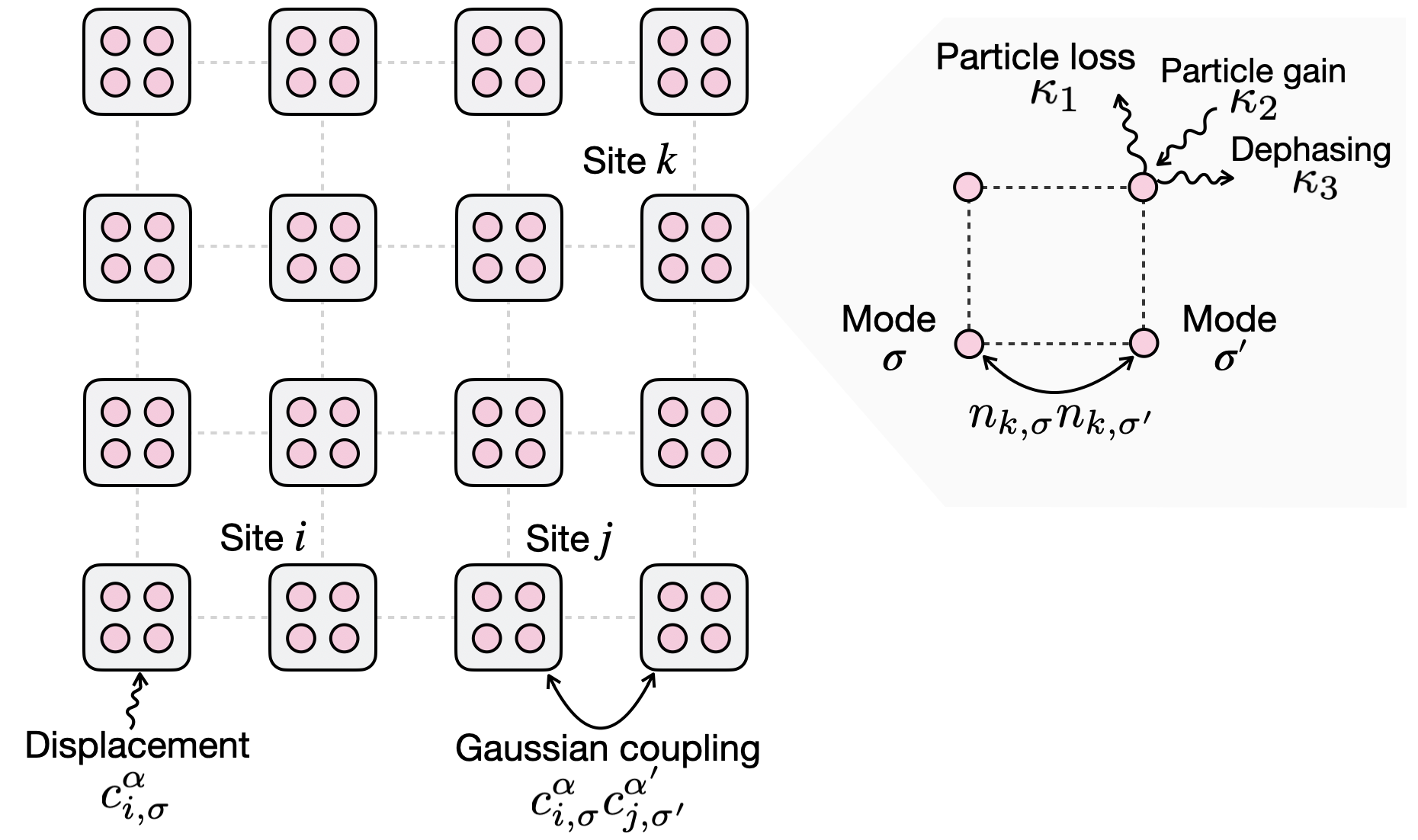}
    \caption{Sketch of the model, with $n$ sites on a lattice, where each site contains $L$ modes. There are Gaussian couplings between the different sites, while the non-Gaussian interactions are only onsite. Nonlocal couplings are allowed. In the bosonic case, interactions of the form $n_{i,\sigma}^2$ are also allowed.
    }\label{fig:diagram}
\end{figure}

We also define the parameters $J, \Omega, U$ as the smallest constants such that for every mode $(i, \sigma)$ and all times $t$
\begin{align}\label{eq:constants}
&\sum_{j\neq i, \sigma'} \sum_{\alpha, \alpha'}|J_{i,\sigma;j,\sigma^{\prime}}^{\alpha,\alpha^{\prime}}(t)| \leq  J, \sum_{\alpha}|\Omega_{i;\sigma}^{\alpha}(t)| \leq \Omega \text{ and } \nonumber \\ 
&\ \ \sum_{\sigma^{\prime}} |U_{i,\sigma;i,\sigma^{\prime}}(t)| \leq U.
\end{align}
The parameter $J$ captures the Gaussian coupling strength between a mode and the modes at all other sites, $U$ captures the on-site non-Gaussian interaction strength, and $\Omega$ captures the coherent drive at each site. We will also assume that $J$, $U$, and $\Omega$ are $O(1)$ constants, which is true in most physical models. Finally, we remark that we do not need to assume geometrical locality of the model---our results will apply to geometrically local and non-local models.

The initial state $\rho(0)$ is either a product state (when analyzing entanglement), a Gaussian state (when analyzing non-Gaussianity), or both  (e.g., the vacuum state). In the bosonic case, additionally, $\rho(0)$ will be assumed to satisfy $\mathrm{Tr}(n_{i,\sigma}^k \rho(0)) \leq C_0^k k^{\alpha_0 k + \beta_0}, \forall (i,\sigma)$, for $k \in \{1,2...\}$, and for some $C_0, \alpha_0, \beta_0 > 0$: this condition guarantees that the probability of finding $\geq k$ particles in a mode decreases super-polynomially with $k$, as would be expected in a physically preparable bosonic state \cite{kuwahara2024bosons_liebrobinson}.

\begin{figure}
    \includegraphics[scale=0.45]{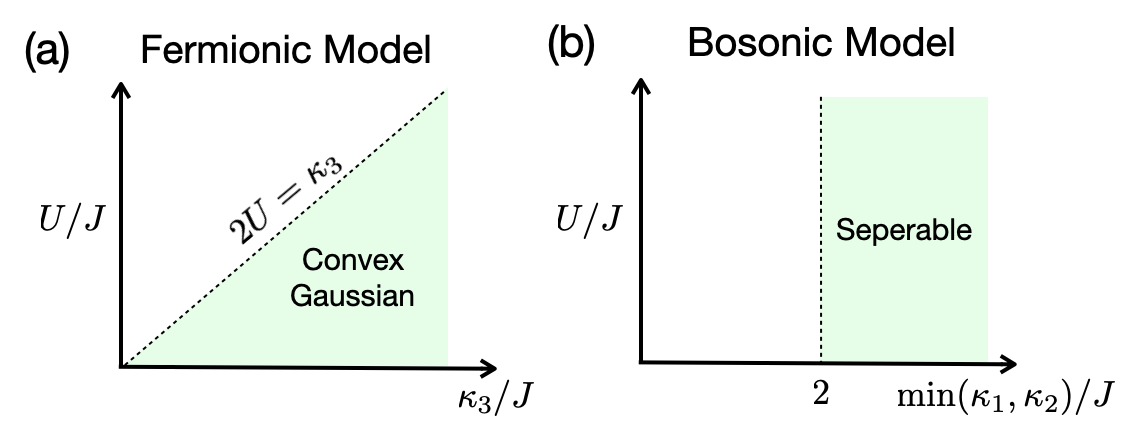}
    \caption{Phase diagram for both bosonic and fermionic systems in the presence of generic noise. (a) For fermionic system the state remains convex-Gaussian at all times for error rates $\kappa_3 \geq 2U$. (b) In bosonic systems the state remains separable at all times for error rates $\kappa_1, \kappa_2 \geq 2J$.}\label{fig:phase_diagram}
\end{figure}

\emph{Results}: Our results, depicted in Fig.~\ref{fig:phase_diagram}, show the simulability of the fermionic and bosonic models as a function of $J, U, \kappa_i$. We first establish that when the noise rate is larger than the on-site non-Gaussian interaction strength $U$, the fermionic model remains convex-Gaussian at all times, and can therefore be classically efficiently sampled from.

\begin{figure*}
    \centering
    \includegraphics[scale=0.35]{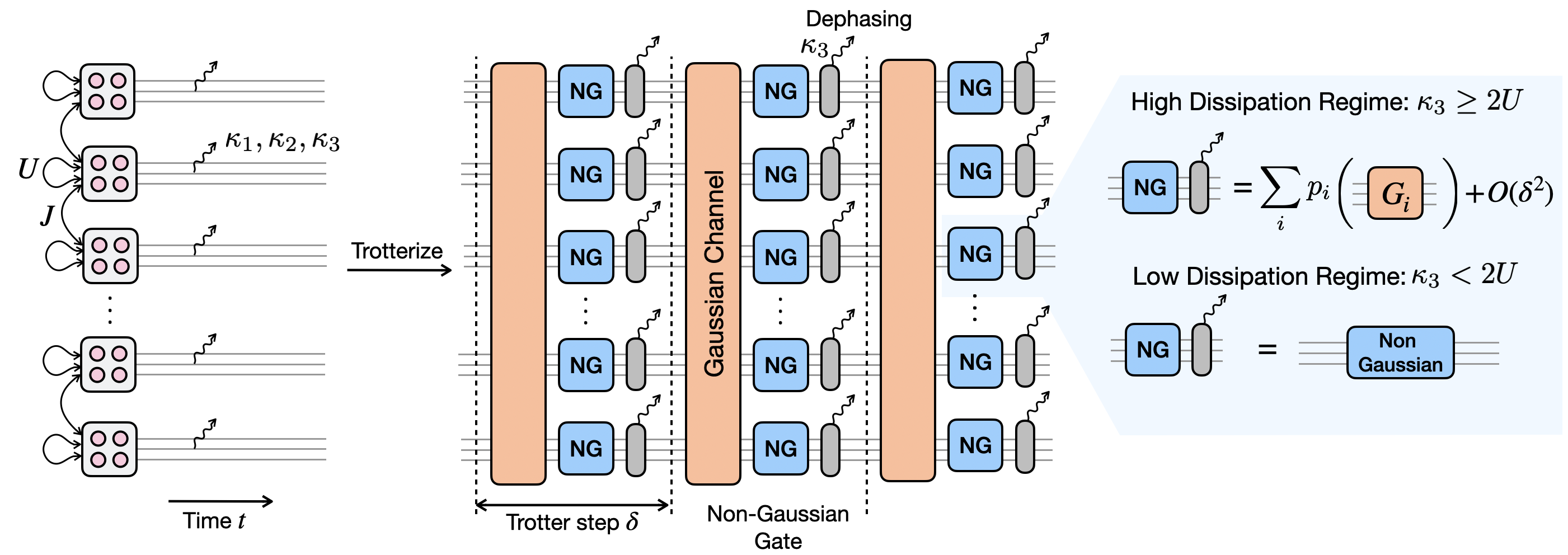}
    \caption{Schematic depiction of the Trotterization schemes in the proof of Theorem \ref{theorem:fermions_high_noise} (fermionic systems with weak non-Gaussianity). For simplicity, we only depict a 1D setting, with each site containing $3$ modes ($L=3$).
     A single Trotter step consists of a Gaussian channel (orange rectangles) that includes the combined effect of $H_g(t)$ and the particle gain and loss dissipators, followed by non-Gaussian gates (blue rectangles) interspersed with the dephasing dissipator (gray curved rectangles). Crucially, a non-Gaussian gate followed by sufficiently strong dephasing can be written as a convex combination of Gaussian channels.}
    \label{fig:trotterization_gaussian}
\end{figure*}

\begin{theorem}\label{theorem:fermions_high_noise}
For an initial Gaussian state, if $\kappa_3 \geq 2U$, then the state of the fermionic model at time $t$, $\rho(t)$, is a convex combination of Gaussian states for all $t \geq 0$. Furthermore, $\rho(t)$ can be classically sampled in the Fock state basis to an $\epsilon$ total variation error in $\textnormal{poly}(n, t, 1/\epsilon)$ time.
\end{theorem}
\noindent Physically Theorem \ref{theorem:fermions_high_noise} suggests that when $\kappa_3 \geq 2U$, dephasing noise destroys non-Gaussianity faster than $H_\text{ng}(t)$ creates it, so that $\rho(t)$ always remains convex-Gaussian. Since $H_{\mathrm{g}}(t)$ preserves convex-Gaussianity, the noise threshold in Theorem~\ref{theorem:fermions_high_noise} is independent of $J$. Notably, it is dephasing that results in this convex-Gaussianity. With only incoherent particle loss/gain, $\rho(t)$ could evolve into a non-Gaussian state at short times even with large dissipation. This arises from the fermionic parity structure---the density matrix of the fermionic model has the form $\rho(t) = \rho_+(t)+\rho_-(t)$, where $\rho_\pm(t)$ is supported only on even/odd parity states. Due to this structure, convex-Gaussianity in $\rho(t)$ requires both $\rho_\pm(t)$ to be convex-Gaussian \cite{oszmaniec2014classical}. $H_\text{ng}(t)$ generates non-Gaussianity individually in both $\rho_\pm(t)$. However, the loss/gain dissipators, to first order, switch the parity of the state and {do not} act within the two parity subspaces. Consequently, they cannot immediately counter the non-convex-Gaussianity created by $H_\text{ng}(t)$.

We provide a complete proof of Theorem \ref{theorem:fermions_high_noise} in the supplement: The starting point is a Trotterization of the Lindbladian in Eq.~\ref{eq:lindblad} --- in each Trotter step, we express the evolution as (a) a Gaussian unitary corresponding to $H_\text{g}(t)$, particle loss and gain followed by (b) the single-site channels generated by $H_\text{ng}(t)$ and dephasing (Fig.~\ref{fig:trotterization_gaussian}). Analyzing the single-site channel, we show that for $\kappa_3 \geq 2U$, this channel maps an input convex-Gaussian state to a convex-Gaussian state. Furthermore, we explicitly construct the output convex-Gaussian state, which allows us to sample from it \cite{Terhal2002_classical_simulation_fermions,knill2001fermioniclinearopticsmatchgates}.

Next, we consider the bosonic model and establish that when the noise rate is larger than the inter-site Gaussian coupling, $\rho(t)$ remains separable at all times, and can therefore be classically efficiently sampled from in the Fock state basis. 

\begin{figure*}
    \centering
    \includegraphics[scale=0.35]{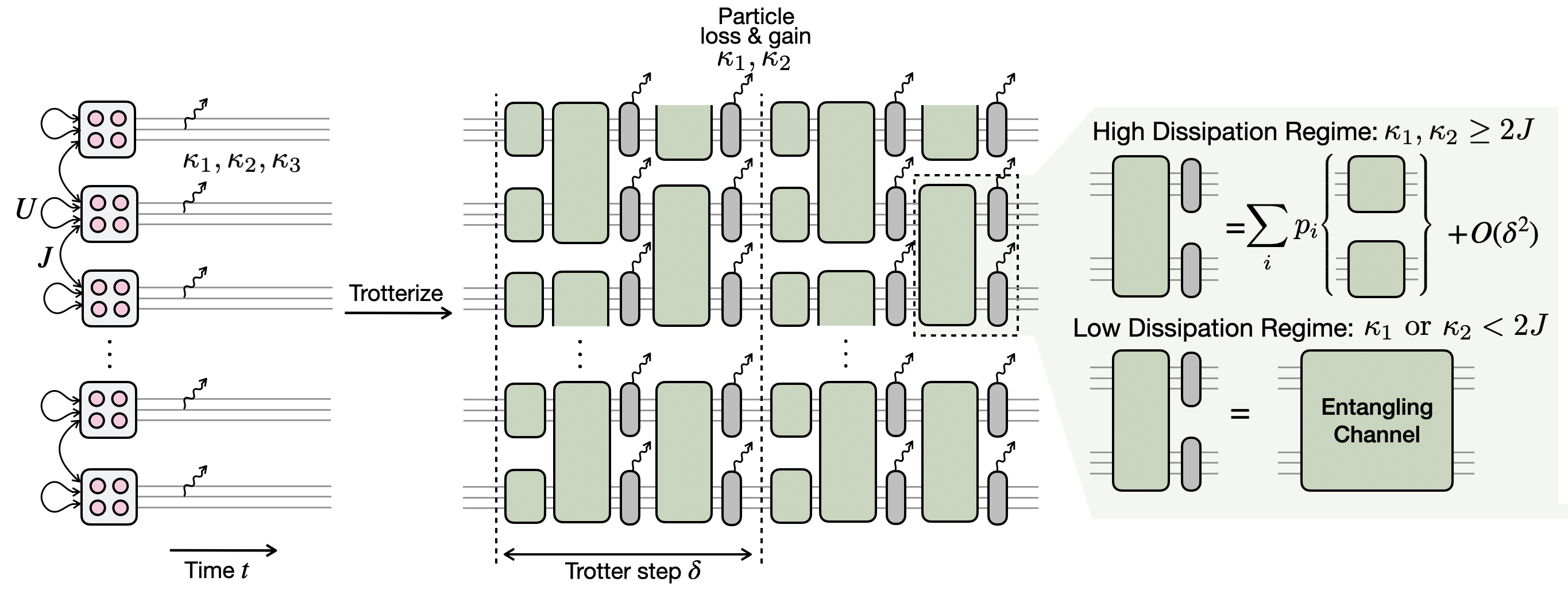}
    \caption{Schematic depiction of the Trotterization schemes in the proof of Theorem \ref{theorem:weak_couplings} (bosonic systems with weak inter-site couplings). For simplicity, we only depict a 1D setting, with each site containing $3$ modes ($L=3$).
    A single Trotter step consists of a layer of single-site channels which include the Hamiltonian terms acting on that site and the dephasing dissipator, followed by 2-site gates interspersed with particle gain and loss dissipators (in gray circles). Crucially, a 2-site gate followed by sufficiently strong noise can be written as a convex combination of single-site channels}
    \label{fig:trotterization_entgl}
\end{figure*}

\begin{theorem}\label{theorem:weak_couplings} Suppose $\rho(t)$ is the state obtained after evolving the bosonic model for time $t$ with an initial product state, then for $\kappa_1, \kappa_2 \geq 2J $ the state $\rho(t)$ is separable for all $t\geq 0$. Furthermore, there is a randomized classical algorithm that can sample $\rho(t)$ in the Fock state basis to $\epsilon$ total variation error in $\textnormal{poly}(n, t, 1/\epsilon)$ time.
\end{theorem}
\noindent Our result formalizes the intuition that, when noise exceeds the inter-site coupling strength, a buildup of entanglement is prohibited and  no classically non-trivial state is generated. Notably, the noise threshold is determined by particle loss and gain noise.  This arises from the dephasing dissipator being diagonal in the Fock basis, while evolution under $H_g(t)$ creates entanglement through off-diagonal elements (i.e., coherences). Thus, to first order, dephasing cannot counter this entanglement generation. Particle gain or loss dissipators, however, are not diagonal in the Fock basis and can prevent it. In the Supplement, we also extend Theorem \ref{theorem:weak_couplings} to quantum spin models with single spin noise and bosonic models with inter-site non-Gaussian couplings \cite{supplemental_material}. Unlike previous percolation-based arguments limiting entanglement to $O(\log n)$ qubit clusters for sufficiently strong noise \cite{Trivedi2022_transitions, aharonov2000_percolation}, we show that the state is entirely separable, and provide an explicit construction that can be efficiently sampled from.

A detailed proof of Theorem 2 is provided in the Supplement \cite{supplemental_material}. Similar to Theorem 1, we begin by a first-order Trotterization of the model but with a different decomposition of the Lindbladian: We express it as a product of (a) single site gates, which contain the unitary generated by $H_\text{ng}(t)$, single-site terms in $H_\text{g}(t)$ and the dephasing dissipator and (b) two-site channels which contain the unitary generated by the inter-site terms in $H_\text{g}(t)$ paired together with the particle gain and loss dissipators [Fig.~\ref{fig:trotterization_entgl}]. We denote the channel acting between modes $(i,\sigma)$ and $(j,\sigma^{\prime})$ at the Trotter-step $\tau$ as $\Phi_{\tau \delta,(\tau-1) \delta}^{i,\sigma;j,\sigma^{\prime}}$, where $\delta$ is the size of the Trotter-step. Importantly, the Trotterization is performed in such a way that the channel $\Phi_{\tau \delta,(\tau-1) \delta}^{i,\sigma;j,\sigma^{\prime}}$ can be understood as a time evolution of the inter-site Gaussian couplings between modes $(i,\sigma)$ and $(j,\sigma^{\prime})$, followed by noise on both modes. This effectively redistributes the single-site noise into ``gate-based" noise on the inter-site gates. Analyzing $\Phi_{\tau \delta,(\tau-1) \delta}^{i,\sigma;j,\sigma^{\prime}}$, we show that for $\kappa_1, \kappa_2 \geq 2J$, it is separability-preserving and thus the state remains separable at all times. Furthermore, we also explicitly construct an $O(\delta^2)$ approximation to the separable state after each time-step and obtain an explicit algorithm to sample from $\rho(t)$.

\emph{Tightness of Theorems 1 and 2}. We can now ask if a version of Theorem \ref{theorem:fermions_high_noise} holds for bosonic systems i.e., is there a noise threshold $\kappa_\text{th}(U)$ dependent only on the non-Gaussian strength $U$ and uniform in $J, \Omega$ that guarantees a convex-Gaussian at all times? For dephasing noise, we provide numerical evidence to the contrary: even for $\kappa_3 \gg U$, single-mode dynamics can yield states with negative Wigner function, and thus not convex-Gaussian states can be generated at time-scales $\sim 1/U$ \cite{Mari_2012_wigner, supplemental_material}. While negativity of the Wigner function suggests classical simulation hardness \cite{Mari_2012_wigner,cormick2006_wigner_classicality}, we do not rule out the existence of an efficient classical algorithm. When $\kappa_3 = 0$, in the Supplement we show that no matter how small $U$ is relative to $\kappa_1, \kappa_2$, computing expected local particle numbers is BQP-hard if $J, \Omega$ can be arbitrarily large but $O(1)$ \cite{supplemental_material}. This builds upon Refs.~\cite{liangjiang2023universalcontrolbosonicsystems,eickbusch2022fast_gates} which perform a universal gate-set on a single bosonic mode with an arbitrarily small effective gate-error rate by engineering the displacement and squeezing. We extend this technique to also implement an entangling gate between two oscillators thus yielding a high-fidelity universal multi-mode gate-set. Together with results from Ref.~\cite{benor2013quantumrefrigerator}, this suggests that when the noise is non-unital (i.e. $\kappa_1 \neq \kappa_{2}$), by using sufficiently large $J$ and $\Omega$, a fault-tolerant quantum computation can be encoded into the model \cite{noh2020_fault_tolerant_bosons,matsuura2024fault_tolerant_bosons,aharonov1999faulttolerantquantum}. Thus, it is unlikely to be able to classically compute even local observables in this setting unless BQP = BPP.

Finally, we consider if a version of Theorem \ref{theorem:weak_couplings} holds for fermions i.e., do noise rates larger than the Gaussian inter-site couplings result in separability at all times. We answer this question in the negative: in the Supplement we show by analyzing few-mode fermionic models that, in contrast with bosons, no matter how high $\kappa_2, \kappa_3$ are, the system does not remain separable at all times and can exhibit entanglement at time-scales $\sim \min(\kappa_2^{-1}, \kappa_3^{-1})$. In fact, this short-time non-separability holds not only if we consider separability with respect to all observables \cite{Moriya_2006_entanglement_fermions}, but also if we consider a weaker notion of separability with respect to only parity-conserving observables \cite{maricarmen2007_entanglement_fermions}. However, this result does not rule out separability at longer times or other routes to classical simulation.

\emph{Conclusion and outlook}. We have characterized the classical complexity of simulating the continuous-time evolution of fermionic and bosonic systems as a function of the noise, Gaussian and non-Gaussian interaction strengths. Future theoretical directions include the study extending our results to non-Markovian models of dissipation. 

The models considered in this paper can be experimentally realised in several platforms. The bosonic model can be implemented in superconducting systems where the Gaussian Hamiltonian can be controlled by designing capacitive couplings between different qubits and the single-site non-Gaussianity by the nonlinear Josephson potential in the qubit \cite{zhang2023_superconducting_simulator,shi2024simulator_superconducting}. We can also use cold bosonic atoms in optical lattices where the strength of both the Gaussian and the non-Gaussian Hamiltonians can be controlled by tuning the optical lattice potential  \cite{bloch_2017_boson_simulation,yang2020simulator_bose_hubbard}. The fermionic model can be implemented either with cold atoms in optical lattices by using a fermionic species of atoms  \cite{Yan2022_fermion_simulation,Spar2022_fermion_simulation,Norcia2018_fermion_simulation}, or in solid-state systems such Moir{\'e} superlattices hosting trions \cite{liu2021signatures, wang2021moire, baek2021optical}. Since all of these systems will have intrinsic particle loss, gain, and dephasing, tuning the parameters in the Gaussian and non-Gaussian Hamiltonians could allow us to access the parameter regimes in Theorems 1 and 2. A major challenge in experimentally verifying the threshold behavior predicted by Theorems 1 and 2 would be verifying the presence (or absence) of entanglement/non-Gaussianity in $\rho(t)$. While this could be hard to do for large systems, we remark that the difference in the threshold behavior in fermionic and bosonic models that we described can be understood even with systems with few ($\leq 4$) fermionic or bosonic  modes, which is well within the regime where a full state tomography can already be performed.

\begin{acknowledgements}
We thank Ashish Clerk and Liang Jiang for useful discussions and Peter McMahon for discussions that inspired this project. R.T acknowledges support from Center for Integration of Modern Optoelectronic Materials on Demand (IMOD) seed grant (DMR-2019444). This research was supported in part by grant NSF PHY-2309135 to the Kavli Institute for Theoretical Physics (KITP). The research is part of the Munich Quantum Valley, which is supported by the Bavarian State Government with funds from the High tech Agenda Bayern Plus. J.I.C, R.T, G.G.G acknowledge funding from the project FermiQP of the Bildungsministerium für Bildung und Forschung (BMBF). A.V.G.~acknowledges support from the U.S.~Department of Energy, Office of Science, Accelerated Research in Quantum Computing, Fundamental Algorithmic Research toward Quantum Utility (FAR-Qu). A.V.G.~was also supported in part by NSF QLCI (award No.~OMA-2120757), DoE ASCR Quantum Testbed Pathfinder program (awards No.~DE-SC0019040 and No.~DE-SC0024220), NSF STAQ program, AFOSR MURI, DARPA SAVaNT ADVENT, and NQVL:QSTD:Pilot:FTL. A.V.G.~also acknowledges support from the U.S.~Department of Energy, Office of Science, National Quantum Information Science Research Centers, Quantum Systems Accelerator. 

\end{acknowledgements}
\nocite{kraus2009thesis,2022_tagliacozzo_fermionic_gs,cormen2022introduction,Fagotti_2010_entanglement_fermions,Terhal2002_classical_simulation_fermions,knill2001fermioniclinearopticsmatchgates,kuwahara2024bosons_liebrobinson,noh2020_fault_tolerant_bosons,aharonov1999faulttolerantquantum,matsuura2024fault_tolerant_bosons,Mari_2012_wigner,1974_hudson_theorem,walschaers2021_nongaussian_states,liangjiang2023universalcontrolbosonicsystems,eickbusch2022fast_gates,lingenfelter2021_fock_state_generation,boykin2002_cooling1,schulman1999_cooling2,Alhambra2019heatbathalgorithmic,benor2013quantumrefrigerator,Trivedi2022_transitions,shtanko2024complexitylocalquantumcircuits,maricarmen2007_entanglement_fermions,Moriya_2006_entanglement_fermions,code_and_data}
\bibliographystyle{apsrev4-1}
\bibliography{references.bib}

\end{document}


\title{Supplemental material to \\ ``Dynamical complexity of non-Gaussian many-body systems with dissipation" }
\author{Guillermo González-García$^{1, 2}$}
\author{Alexey V. Gorshkov$^{3,4}$}
\author{J.~Ignacio Cirac$^{1, 2}$}
\author{Rahul Trivedi$^{1, 2}$}
\email{rahul.trivedi@mpq.mpg.de}
\address{$^1$Max-Planck-Institut für Quantenoptik, Hans-Kopfermann-Str.~1, 85748 Garching, Germany \\
$^2$Munich Center for Quantum Science and Technology (MCQST), Schellingstr. 4, D-80799 Munich, Germany \\
$^3$Joint Quantum Institute, NIST/University of Maryland, College Park, Maryland 20742, USA.\\ $^4$Joint Center for Quantum Information and Computer Science, NIST/University of Maryland, College Park, Maryland 20742, USA.}

\date{\today}
\maketitle
\raggedbottom
\onecolumngrid

This Supplemental Material is organized as follows: First, in section \ref{supplemental:preliminaries}, we provide the necessary notation and background for the rest of the Supplemental Material. In section \ref{supplemental:theorem_1_proof}, we provide the proof of Theorem 1, showing convex-Gaussianity and simulability for the fermionic model for sufficiently high noise rates. Then, in section \ref{supplemental:theorem_2_proof}, we prove Theorem 2 for the bosonic model, which implies separability and simulability for sufficiently high noise rates. In section \ref{supplemental:spins}, we extend this result to a class of spin models: for $2$-local Hamiltonians and sufficiently high noise rates, the system can be shown to be separable at all times. Finally, in section \ref{supplemental:counterexamples}, we show that an analogue of Theorem 1 cannot exist for bosonic systems, and that an analogue of Theorem 2 cannot exist for fermionic systems.

\section{Notation and preliminaries}\label{supplemental:preliminaries}

In this section, we provide the necessary notation and background for the rest of the Supplemental Material. This includes the notation regarding operators and norms (subsection \ref{subsection:operators}), a brief summary on several properties of bosonic and fermionic systems (subsection \ref{subsection:preliminaries_fermions_bosons}), the Trotter formula that will be used throughout the proofs (subsection \ref{subsection:trotter}), the asymptotic notation that we will employ (subsection \ref{subsection:asymptotic}), and a detailed presentation of the bosonic and fermionic models that we will analyze (subsection \ref{subsection:model}).

\subsection{Operators, superoperators and their norms}\label{subsection:operators}
For a quantum state $\ket{\psi}$, $\norm{\ket{\psi}}$ will denote its usual norm $\norm{\ket{\psi}}^2 = \bra{\psi}\psi\rangle$. For an operator $A$, we will use $\norm{A}_p$ to denote its Schatten-$p$ norm:
\begin{align}
\norm{A}_p = \bigg(\sum_{i} \sigma_i^p(A)\bigg)^{1/p} \text{, where }\sigma_1(A) \geq \sigma_2(A) \geq \sigma_3(A) \dots \text{ are the singular values of }A.
\end{align}
We will often use $\norm{A} = \sigma_1(A) = \norm{A}_\infty$ to denote its operator norm and $\norm{A}_F = \norm{A}_2 =  [\text{Tr}(A^\dagger A)]^{1/2}$ to denote its Frobenius norm. We will often use the Holder's inequality, which states that
\begin{align}
\norm{AB}_1 \leq \norm{A}_p \norm{B}_q \text{ where }\frac{1}{p} + \frac{1}{q} = 1.
\end{align}
In particular, $\norm{AB}_1 \leq \norm{A}\norm{B}_1$. It is also convenient to note the Cauchy-Schwarz inequality for operators: Suppose $\omega$ is a positive semi-definite operator, then
\begin{align}
\abs{\text{Tr}(A^\dagger B \omega)}^2 \leq \text{Tr}(A^\dagger A \omega) \text{Tr}(B^\dagger B \omega).
\end{align}
For super-operators $\mathcal{A}$, we will use $\norm{\mathcal{A}}_\diamond$ to denote its diamond norm. In our analysis, we will often encounter super-operators of the form
\begin{align}
\mathcal{A}(\rho) = \sum_i A_i \rho B_i,
\end{align}
where $A_i$ and $B_i$ are some operators. For such super-operators, it is convenient to note that the Holder's inequality implies that
\begin{align}\label{eq:super_op_bound}
\norm{\mathcal{A}}_\diamond \leq \sum_{i}\norm{A_i}\norm{B_i}.
\end{align}
For instance, given an operator $L$, we will often use $\mathcal{D}_L$ to denote the following superoperator:
\begin{align}
\mathcal{D}_L = L \rho L^\dagger - \frac{1}{2}\{L^\dagger L, \rho \},
\end{align}
where $\{ \ \cdot \ , \ \cdot \ \}$ is the anti-commutator between two operators. $\mathcal{D}_L$ will be called the ``dissipator corresponding to $L$". From Eq.~(\ref{eq:super_op_bound}), we then obtain that
\begin{align}
\norm{\mathcal{D}_L}_\diamond \leq \norm{L}^2 + \norm{L^\dagger L} \leq 2\norm{L}^2.
\end{align}
A super-operator $\mathcal{E}$ is completely positive if and only if it can be expressed as 
\begin{align}
\mathcal{E}(\rho) = \sum_i K_i \rho K_i^\dagger
\end{align}
for some operators $K_i$. It will be called a channel if it is additionally trace preserving which requires $\sum_{i} K_i^\dagger K_i = I$. For any completely-positive trace preserving map $\mathcal{E}$, $\norm{\mathcal{E}}_\diamond \leq 1$.
\subsection{Fermions and Bosons}\label{subsection:preliminaries_fermions_bosons}
The Hilbert space of $m$ fermionic modes is described by the  vacuum state $\ket{\text{vac}}$ and the standard creation ($a_{i}^{\dagger}$) and annihilation ($a_{i}$) operators, with $i \in \{1,2, \cdots m\}$ labeling the fermionic mode. These satisfy the canonical anticommutation relations:
\begin{align}
\{a_{i},a_{j}\}=0 \text{ and }\{a_{i},a_{j}^{\dagger}\}=\delta_{i,j}.
\end{align}
The Hilbert space of the fermionic model is the finite-dimensional vector space given by $\text{span}\{\prod_{k = 1}^m  (a_{k}^\dagger)^{\mu_{k}} \ket{\text{vac}} : \mu_{k} \in \{0, 1\}\}$. It will be convenient to work with the $2m$ Majorana fermion operators defined by 
\begin{align}
c_{i}^1=\frac{1}{\sqrt{2}}\big(a_{i}^{\dagger}+a_{i}\big) \text{ and }  c_{i}^2=\frac{i}{\sqrt{2}}\big(a_{i}^{\dagger}-a_{i}\big).
\end{align}
The Majorana operators are each Hermitian, traceless and satisfy $\{c_{i}^\alpha, c_{i'}^{\alpha'}\} = \delta_{i, i'} \delta_{\sigma, \sigma'}$. We define $\mathcal{C}_{2m}$ as the algebra generated by the $2m$ Majorana operators: An operator $X \in \mathcal{C}_{2m}$ can be expressed as a linear combination of monomials of the form $\prod_{i = 1}^m \prod_{\alpha = 1}^2 (c_{i}^\alpha)^{\mu_{i}^{\alpha}}$, where $\mu_{i}^\alpha \in \{0, 1\}$. The operator $X$ will be even if it is a linear combination of only even degree monomials, and odd if it is a linear combination of odd degree monomials. Furthermore, any Hermitian operator defined on the fermionic Hilbert space is also in $\mathcal{C}_{2m}$ and, as usual, fermionic quantum states are positive semi-definite Hermitian operators in $\mathcal{C}_{2m}$.

Given a fermionic state $\rho$, its correlation matrix elements are defined by $\Gamma_{i, i'}^{\alpha, \alpha'}=i\mathrm{tr}(\rho[c_{i}^{\alpha},c_{i'}^{\alpha'}])/2.$ A fermionic state $\rho$ is called Gaussian if it can be expressed as $\exp(-\beta H) / \text{Tr}(\exp(-\beta H))$ for some Hermitian operator $H$ that is quadratic in the Majorana operators and  $\beta \in \mathbb{R} \cup \{-\infty, \infty\}$. Thus, fermionic Gaussian states are either Gibb's states of Hamiltonians that are quadratic in the Majorana operators, or are projectors on their ground-state subspace. Fermionic Gaussian states are fully characterized by their correlation matrix elements $\Gamma_{i,i'}^{\alpha, \alpha'}$ \cite{kraus2009thesis,2022_tagliacozzo_fermionic_gs}. We will refer to a fermionic state as \emph{convex-Gaussian} if it can be expressed as a convex combination of Gaussian states.

Similar to fermions, the Hilbert space of $m$ bosonic modes will be described by a vacuum state $\ket{\text{vac}}$ and the creation ($a_{i}^\dagger$) and annihilation ($a_{i}$) operators, with $i \in \{1, 2 \dots m\}$ labeling the bosonic mode. These satisfy the canonical commutation relations:
\begin{align}
[a_{i}, a_{i'}] = 0 \text{ and }[a_{i}, a^\dagger_{i'}] = \delta_{i, i'}.
\end{align}
The Hilbert space of the bosonic model is the infinite-dimensional vector space given by $\text{span}\{\prod_{i} (a^\dagger_{i})^{\mu_{i}} \ket{\text{vac}} : \mu_{i} \in \{0, 1, 2 \dots \}\}$. It will be convenient to work with the $2m$ quadrature operators
\begin{align}
{c}_{i}^1=\frac{1}{\sqrt{2}}\big({a}_{i}^{\dagger}+{a}_{i}\big) \text{ and }  {c}_{i}^2=\frac{i}{\sqrt{2}}\big({a}_{i}^{\dagger}-{a}_{i}\big).
\end{align}
The quadrature operators are Hermitian and satisfy $[c_{i}^\alpha, c_{i}^{\alpha'}] = -i \delta_{i, i'} \Omega_{\alpha, \alpha'} $, where $\Omega$ is the $2 \times 2$ symplectic matrix. Similar to a fermionic state, a bosonic state $\rho$ is Gaussian if it can be expressed as $\exp(-\beta H) / \text{Tr}(\exp(-\beta H))$ for some Hermitian operator $H$ which is quadratic or linear in the quadrature operators and for some $\beta \in \mathbb{R} \cup \{-\infty, \infty\}$. A state will be called convex-Gaussian if it can be expressed as a convex combination of Gaussians. A useful property of Gaussian states that we will use in our analysis is given in the lemma below.
\begin{lemma}\label{lemma:closure_gaussian_product}
    Suppose $\rho$ is a (fermionic or bosonic) Gaussian state and $A = \sum_{i, i'}A_{i, i'}^{\alpha, \alpha'}c_i^\alpha c_{i'}^{\alpha'}$ is a quadratic operator, then $\rho' = e^{-A} \rho e^{-A^\dagger} / \textnormal{Tr}(e^{-A} \rho e^{-A^\dagger})$ is also a Gaussian state.
\end{lemma}
\begin{proof}
    This follows from the closure of quadratic and linear operators under commutation, i.e.
    \begin{enumerate}
        \item[(1)] For fermions, the commutator of any two operators of the form $\sum_{i, i'} \sum_{\alpha, \alpha'} x_{i, i'}^{\alpha, \alpha'} c_{i}^\alpha c_{i'}^{\alpha'}$, where $x_{i, i'}^{\alpha, \alpha'}\in \mathbb{C}$, is again of the same form. 
        \item[(2)] For bosons, the commutator of any two operators of the form $\sum_{i, i'} \sum_{\alpha, \alpha'} x_{i, i'}^{\alpha, \alpha'} c_{i}^\alpha c_{i'}^{\alpha'} + \sum_{i, \alpha} y_{i}^\alpha c_{i}^\alpha$, where $x_{i, i'}^{\alpha, \alpha'}, y_{i}^\alpha \in \mathbb{C}$, is again of the same form.
    \end{enumerate}
    Since $\rho$ is a Gaussian state, it is expressible as $\exp(-\beta H) / \text{Tr}(\exp(-\beta H))$, where $H$ is a quadratic form in $c_{i}^\alpha$ with a possible linear term in $c_{i}^\alpha$ for bosons. Consequently, using the Baker-Campbell-Hausdorff formula, we obtain that $e^{-A} \rho e^{-A^\dagger}$ can be written as a linear combination of $A, A^\dagger, H$ and their nested commutators. Consequently, from 1 and 2 above, we obtain that $e^{-A} \rho e^{-A^\dagger} \propto \exp(-\beta H')$ for some $\beta$ and $H'$ that is also a quadratic form in $c_{i}^\alpha$ with a possible linear term for bosons. Furthermore, note that since $\rho$ is positive-semidefinite, so is $e^{-A}\rho e^{-A^\dagger}$ and thus is a valid quantum state. 
\end{proof}

\subsection{Trotter formula}\label{subsection:trotter}
In our analysis below, we will often use first-order Trotterization for time-dependent models. Given a time-dependent Lindbladian $\mathcal{L}(t) =\mathcal{L}^{(1)}(t) + \mathcal{L}^{(2)}(t) + \dots \mathcal{L}^{(M)}(t)$, its first-order Trotterization in the time-interval $[0, t]$, with $T$ Trotter steps each of length $\delta = t / T$, will be given by
\begin{align}
\Phi = \prod_{\tau = T}^{1} \Phi^{(1)}_{\tau \delta, (\tau - 1)\delta} \Phi^{(2)}_{\tau \delta, (\tau - 1)\delta} \dots \Phi^{(M)}_{\tau \delta, (\tau - 1)\delta} \text{ where }  \Phi^{(j)}_{\tau \delta, (\tau - 1)\delta} = \mathcal{T}\exp\bigg(\int_{(\tau - 1)\delta}^{\tau \delta}\mathcal{L}^{(j)}(s) ds\bigg). 
\end{align}
In Lemma \ref{lemma:Trotter_bounded_lind} below, we provide an upper bound on the error between the exact evolution $\mathcal{T}\exp(\int_0^t \mathcal{L}(s) ds)$ and the Trotter formula $\Phi$ that we will use repeatedly in the following sections. 

\begin{lemma}[Trotter error for bounded Lindbladians]\label{lemma:Trotter_bounded_lind}
Suppose for any $s \geq 0$ and $j \in \{1, 2 \dots M\}$, $\smallnorm{\mathcal{L}^{(j)}(s)}_\diamond \leq \ell_j$, then for any $T > 0$,
\begin{align}\label{eq:Trotter_bounded}
\norm{\mathcal{T}\exp\bigg(\int_0^t \mathcal{L}(s) ds\bigg) - \Phi}_\diamond \leq   \frac{t^2}{T} \bigg(\sum_{j = 1}^M \ell_j\bigg)^2.
\end{align}
\end{lemma}

\subsection{Asymptotic notation}\label{subsection:asymptotic}
Throughout the paper, we employ the following asymptotic notation commonly used in complexity theory \cite{cormen2022introduction}:

\begin{table}[htpb]
\begin{tabular}{ |c|c|c| } 
\hline
 Notation & Formal definition & Informal description \\ 
 \hline
 $f(n)=\Omega(g(n))$ & $\exists k>0,  n_0 :\forall n>n_0,  |f(n)| \geq kg(n)$ & $f(n)$ grows at least as fast as $g(n)$ \\   \hline
 $f(n)=O(g(n))$ & $\exists k>0, n_0: \forall n>n_0, |f(n)| \leq kg(n)$ & $f(n)$ grows no faster than $g(n)$  \\ \hline
 $f(n)=\Theta(g(n))$ & $\exists k_1>0, k_2>0, n_0: \forall n>n_0, k_1g(n) \leq f(n) \leq k_2g(n)$ & $f(n)$ and $g(n)$ grow equally fast\\ 
 \hline
\end{tabular}
 \caption{Table of asymptotic notation used in this paper.}
\end{table}

\subsection{Model}\label{subsection:model}
Here, we briefly recap the fermionic and bosonic models introduced in the main text and streamline the notation. We will consider a more general setting than the one described in the main text: specifically, we will allow here for inter-site non-Gaussian interactions. 
We recall that we consider systems with $n$ sites, with each site containing $L$ bosonic or fermionic modes. We will use $m = nL$ to denote the total number of modes in the system. With the $\sigma^\text{th}$ mode at the $i^\text{th}$ site, where $\sigma \in \{1, 2 \dots L\}$ and $i \in \{1, 2 \dots n\}$, we will associate an annihilation operator $a_{i, \sigma}$---it will be notationally convenient for us to group $i, \sigma$ into a single index $v = (i, \sigma)$ and denote the corresponding annihilation operator by $a_{v}$. Furthermore, corresponding to a mode index $v$, we will use $i_v$ to denote the site the mode is at and $\sigma_v$ to be the local index of the mode. Associated with the mode at $v$, we will also define the operators $n_v, c_v^1, c_v^2$ via
\begin{align}
    n_v =a_v^\dagger a_v, c_v^1 = \frac{a_v + a_v^\dagger}{\sqrt{2}} \text{ and }c_v^2 = \frac{a_v - a_v^\dagger}{\sqrt{2}i}.
\end{align}
Here, $n_v$ is an operator measuring the number of particles in mode $v$, and $c_v^1, c_v^2$ are the Majorana operators (for fermions) or the quadrature operators (for bosons).

As in the main text, the Hamiltonian for the fermionic or bosonic problem will be decomposed as 
\begin{align}
H(t) = H_\text{g}(t) + H_\text{ng}(t),
\end{align}
where $H_\text{g}(t)$ is Gaussian given by
\begin{align}\label{eq:gaussian_Hamiltonian}
H_\text{g}(t) = \sum_{v, v'} \sum_{\alpha, \alpha'} J_{v, v'}^{\alpha, \alpha'}(t) c_{v}^{\alpha} c_{v'}^{\alpha'} + \sum_{v, \alpha}\Omega_{v}^\alpha(t) c_{v}^\alpha,
\end{align}
with $\Omega_{v}^\alpha(t) = 0$ for fermions, and $H_\text{ng}(t)$ is non-Gaussian given by
\begin{align}\label{eq:non_gaussian_Hamiltonian}
H_\text{ng}(t) = \sum_{v,  v'} U_{v, v'}(t) n_{v} n_{v'}.
\end{align}

Without loss of generality, we can assume that
\begin{align*}
&\text{For fermions, }J_{v, v'}^{\alpha, \alpha'}(t)\text{ is purely imaginary  and }J_{v, v'}^{\alpha, \alpha'}(t) = -J_{v', v}^{\alpha', \alpha}(t), \\
&\text{For bosons, } J_{v, v'}^{\alpha, \alpha'}(t) \text{ is purely real and } J_{v, v'}^{\alpha, \alpha'}(t) = J_{v', v}^{\alpha', \alpha}(t), \\
&\text{For both fermions and bosons, }U_{v, v'}(t) \text{ is purely real and }U_{v, v'}(t) = U_{v', v}(t).
\end{align*}
As in the main text, we will also define constants $J_C, J_\text{os}, U_C, U_\text{os}, \Omega$:
\begin{enumerate}
    \item[(1)] $J_C$ is a measure of the strength of the Gaussian terms coupling modes at different sites: It is the smallest number such that $\forall t$ and $v = (i, \sigma)$
    \begin{align}\label{eq:definition_JC}
    \sum_{i': i' \neq i, \sigma'} \sum_{\alpha, \alpha'} \smallabs{J^{\alpha, \alpha'}_{i, \sigma; i', \sigma'}(t)} \leq J_C.
    \end{align}
    \item[(2)] $J_\text{os}$ is a measure of the strength of the Gaussian terms coupling modes at the same site: It is the smallest number such that $\forall t$ and $v = (i, \sigma)$
    \begin{align}\label{eq:definition_Jos}
    \sum_{\sigma'} \sum_{\alpha, \alpha'} \smallabs{J^{\alpha, \alpha'}_{i, \sigma; i, \sigma'}(t)} \leq J_\text{os}.
    \end{align}
    \item[(3)] $U_C$ is a measure of the strength of the non-Gaussian terms coupling modes at different sites: It is the smallest number such that $\forall t$ and $v = (i, \sigma)$
    \begin{align}\label{eq:definition_UC}
    \sum_{i' \neq i, \sigma'} \smallabs{U_{i, \sigma; i', \sigma'}(t)} \leq U_C.
    \end{align}
    \item[(4)] $U_\text{os}$ is a measure of the strength of the non-Gaussian terms coupling modes at the same site: It is the smallest number such that $\forall t$ and $v = (i, \sigma)$
    \begin{align}\label{eq:definition_Uos}
    \sum_{i\sigma'} \smallabs{U_{i, \sigma; i, \sigma' }(t)} \leq U_\text{os}.
    \end{align}
    \item[(5)] $\Omega$ is a measure of the on-site displacement: it is the smallest number such that $\forall t$ and $v = (i, \sigma)$
    \begin{align}\label{eq:definition_Omega}
    \sum_{\alpha}\smallabs{\Omega_{i, \sigma}^{\alpha}(t)} \leq \Omega.
    \end{align}
    Note that $\Omega \neq 0$ only for the bosonic model---we do not include a displacement term in the fermionic model.
\end{enumerate}
Note that the inter-site non-Gaussian interactions were not included in the main text, and the results quoted in the main text can be obtained by setting $U_C = 0$. Finally, it will also be convenient to define the parameter $\Lambda$ as
\begin{align}\label{eq:Lambda_Def}
\Lambda = J_C + U_C + J_\text{os} + U_\text{os} + \kappa + \Omega.
\end{align}

While analyzing the bosonic model, it will be more convenient to express $H_\text{g}$ as a sum of particle number conserving and non-conserving terms via
\begin{align}\label{eq:H_g_bosonic_model}
H_\text{g}(t) = \underbrace{\sum_{v, v'} \big(\mathcal{J}_{v, v'}(t) a_v^\dagger a_{v'} + \text{h.c.}\big)}_{H_\text{g}^{\text{hop}}(t)} + \underbrace{\sum_{v, v'} \big(\mathcal{G}_{v, v'}(t) a_v a_{v'} + \text{h.c.}\big)}_{H_\text{g}^{\text{sq}}(t)} + \underbrace{\sum_{v}\big(\mathcal{D}_v(t) a_v + \text{h.c.}\big)}_{H_\text{g}^{\text{disp}}(t)},
\end{align}
where, up to a possibly time-dependent energy shift in $H_\text{g}(t)$, $\mathcal{J}_{v, v'}= (J_{v,v^{\prime}}^{1,1}-iJ_{v,v^{\prime}}^{1,2}-iJ_{v,v^{\prime}}^{2,1}+J_{v,v^{\prime}}^{2,2})/2$, $\mathcal{G}_{v, v'}(t) = (J_{v,v^{\prime}}^{1,1}+iJ_{v,v^{\prime}}^{1,2}+iJ_{v,v^{\prime}}^{2,1}-J_{v,v^{\prime}}^{2,2})/2$ and $\mathcal{D}_v = (\Omega_{\nu}^1(t)-i\Omega_{\nu}^2(t))/\sqrt{2}$. Here, $H_\text{g}^{\text{hop}}(t)$ is a particle hopping term between different bosonic modes and conserves the total particle number $N = \sum_{v} n_v$, $H_\text{g}^{\text{sq}}(t)$ can be considered to be a multi-mode squeezing term in the Hamiltonian and $H_\text{g}^{\text{disp}}(t)$ displaces the individual bosonic modes. Both $H_\text{g}^{\text{sq}}(t)$ and $H_\text{g}^{\text{disp}}(t)$ do not conserve the total particle number $N$. It will also be convenient to define the constant $\mathcal{G}$ as the smallest number such that for all $t$ and $v$
\begin{align}\label{eq:squeezing_strength}
\sum_{v'} \abs{\mathcal{G}_{v, v'}(t)} \leq \mathcal{G}.
\end{align}
Furthermore, it can be noted that $\abs{\mathcal{D}_v(t)} \leq \Omega / \sqrt{2}$.

Finally, the noise in the dynamics of the bosonic and fermionic models will be modeled by the Lindbladian $\mathcal{L}_\text{n}$ given by
\begin{align}\label{eq:lindbladian_noise_supplement}
\mathcal{L}_\text{n} = \sum_{l = 1}^3 \sum_{i, \sigma} \kappa_l \mathcal{D}_{L_{i, \sigma}^{(l)}},
\end{align}
where $\mathcal{D}_L \rho = L \rho L^\dagger - \{L^\dagger L, \rho \} / 2$, $L_{i, \sigma}^{(1)} = a_{i, \sigma}$ (incoherent particle loss), $L^{(2)}_{i, \sigma} = a_{i, \sigma}^\dagger$ (incoherent particle gain), and $L_{i, \sigma}^{(3)} = a_{i, \sigma}^\dagger a_{i, \sigma} = n_{i, \sigma}$ (dephasing). Unless otherwise mentioned, we will assume that all three dissipators act on each mode $\kappa_1, \kappa_2, \kappa_3 > 0$ and will denote the total dissipation rate by $\kappa = \kappa_1 + \kappa_2 + \kappa_3$. We summarize all the parameters of the model in Table \ref{tab:par}.
\begin{table}[htpb]
\begin{tabular}{ |m{4cm}|m{5cm}|m{8cm}| } 
\hline
 Parameter & Defined in & Informal description \\ 
 \hline
 $J^{\alpha, \alpha'}_{v, v'}(t)$ or $J^{\alpha, \alpha'}_{i, \sigma; i', \sigma'}(t)$ & Eq.~(\ref{eq:gaussian_Hamiltonian}) & Gaussian coupling between two fermionic or bosonic modes \\   \hline
 $J_C$ & Eq.~(\ref{eq:definition_JC}) & Maximum total strength of Gaussian coupling between a mode and all other modes at different sites  \\  \hline
 $J_\text{os}$ & Eq.~(\ref{eq:definition_Jos}) & Maximum total strength of Gaussian coupling between a mode and all  other modes at the same site \\  \hline
 $U_{v, v'}(t)$ or $U_{i, \sigma; i', \sigma'}(t)$ & Eq.~(\ref{eq:non_gaussian_Hamiltonian}) & Non-Gaussian interaction between two fermionic or bosonic modes\\  \hline
 $U_C$ & Eq.~(\ref{eq:definition_UC}) & Maximum total strength of non-Gaussian interaction between a mode and all other modes at different sites\\  \hline
 $U_\text{os}$ & Eq.~(\ref{eq:definition_Uos}) & Maximum total strength of non-Gaussian interaction between a mode and all other modes at the same site \\  \hline
 $\Lambda$ & Eq.~(\ref{eq:Lambda_Def}) & Total coupling strength \\  \hline
 $\Omega_v(t)$ or $\Omega_{i, \sigma}(t)$ & Eq.~(\ref{eq:gaussian_Hamiltonian}) & On-site displacement acting on bosonic modes\\ \hline
 $\mathcal{J}_{v, u}(t)$ & Eq.~(\ref{eq:H_g_bosonic_model}) & Gaussian hopping between two bosonic modes\\ \hline
 $\mathcal{G}_{v, u}(t)$ & Eq.~(\ref{eq:H_g_bosonic_model})& Multi-mode squeezing term between two bosonic modes\\ \hline
 $\mathcal{G}$ & Eq.~(\ref{eq:squeezing_strength})& Maximum strength of multi-mode squeezing between one bosonic mode with all other modes \\ \hline
 $\mathcal{D}_v(t)$ & Eq.~(\ref{eq:H_g_bosonic_model})& Single-mode displacement acting on bosonic modes \\ \hline
 $\kappa$ & Eq.~(\ref{eq:lindbladian_noise_supplement}) & Total dissipation rate \\ \hline
 $\kappa_1$ & Eq.~(\ref{eq:lindbladian_noise_supplement})& Dissipation rate for incoherent particle loss\\ \hline
 $\kappa_2$ & Eq.~(\ref{eq:lindbladian_noise_supplement})& Dissipation rate for incoherent particle gain\\ \hline
 $\kappa_3$ & Eq.~(\ref{eq:lindbladian_noise_supplement})& Dissipation rate for dephasing\\ \hline
 $\gamma$ & Assumption \ref{assump:gain_loss} & Defined as $\gamma = \kappa_1 - \kappa_2 - 2\mathcal{G}$ \\ \hline
 $n$ & --- & Number of sites\\ \hline
 $L$ & ---  & Number of modes per site\\ \hline
 $m$ & ---  & Total number of modes $m = nL$ \\
 \hline
\end{tabular}
 \caption{Table of all the coefficients and parameters relevant to the bosonic and fermionic models. \label{tab:par}}
\end{table}
\section{High noise simulability of the fermionic  model (Theorem 1)}\label{supplemental:theorem_1_proof}
In this section, we will present the proof of Theorem 1, which establishes the high-noise simulability of the fermionic model. We will first analyze the Trotterization of the continuous-time model, followed by analyzing each Trotter time-step to establish its convex Gaussianity for high noise and to obtain an explicit algorithm for classically simulating either sampling from or computing local observables in the fermionic state. We begin with a first-order Trotter approximation to $\rho(t)$ with the following splitting of the Lindbladian $\mathcal{L}(t)$ into a Gaussian and non-Gaussian Lindbladian:
\begin{align}
    \mathcal{L}(t) = \underbrace{-i[H_\text{g}(t), \ \cdot \ ] +  \sum_{l \in 
    \{1, 2\}}\sum_{i, \sigma} \kappa_l \mathcal{D}_{L_{i, \sigma}^{(l)}}}_{\mathcal{L}_\text{g}(t)} \underbrace{-i[H_\text{ng}(t), \ \cdot \ ] +  \sum_{i, \sigma}\kappa_3\mathcal{D}_{L_{i, \sigma}^{(3)}}}_{\mathcal{L}_\text{ng}(t)},
\end{align}
We next Trotterize the state $\rho(t)$ into $T$ Trotter steps: The Trotterized state $\sigma_T$ will be given by
\begin{subequations}\label{eq:Trotter_theorem_2_ferm}
\begin{align}
\sigma_T = \bigg(\prod_{\tau = T}^1 \Phi^{\text{ng}}_{\tau \delta, (\tau - 1)\delta} \Phi^{\text{g}}_{\tau \delta, (\tau - 1)\delta}\bigg) \rho(0),
\end{align}
where $\delta = t/ T$ and
\begin{align}
\Phi^{\text{ng}}_{t, t'} = \mathcal{T}\exp\bigg(\int_{t'}^t \mathcal{L}_\text{ng}(s) ds\bigg) \text{ and }\Phi^{\text{g}}_{t, t'} = \mathcal{T}\exp\bigg(\int_{t'}^t \mathcal{L}_\text{g}(s) ds\bigg).
\end{align}
\end{subequations}
We first provide a bound on $\norm{\sigma_T - \rho(t)}_1$.
\begin{lemma}[Trotterization: Fermionic model]\label{lemma:Trotter_fermion_theorem_2}
For all $T > 0$,
\begin{align*}
\norm{\sigma_T - \rho(t)}_1 \leq \frac{4t^2 m^2}{T} \Lambda^2.
\end{align*}
\end{lemma}
\begin{proof}
Noting that $\norm{\mathcal{L}_\text{ng}(s)}_\diamond \leq 2m (U_C + U_\text{os} + \kappa_3)$ and $\norm{\mathcal{L}_\text{g}(s)}_\diamond \leq 2m (J_C + J_\text{os} + \kappa_1 + \kappa_2)$, we obtain that the parameter $\ell$ in Lemma \ref{lemma:Trotter_bounded_lind} can be chosen to be $2m \Lambda$. The lemma statement then follows directly from Lemma \ref{lemma:Trotter_bounded_lind}.
\end{proof}
\begin{lemma}[Convex-gaussianity condition for 2-fermionic modes]\label{lemma:convex_gaussian_fermions}Consider a Lindbladian on two fermionic modes given by
\begin{align*}
\mathcal{L}(t) = -i[h(t), \cdot] + \sum_{i \in \{1, 2\}} \kappa_i(t) \mathcal{D}_{n_i},
\end{align*}
where $h(t) = u(t) n_1 n_2$ and $n_i$ is the number operator for the $i^\text{th}$ mode. If $\kappa_i(t) \geq \abs{u(t)}$, then the channel $\mathcal{T}\exp(\int_t^{t + \tau}\mathcal{L}(s)ds)$ generated by the Lindbladian in the time interval $(t, t + \tau)$ maps a convex Gaussian state to another convex Gaussian state.
\end{lemma}
\begin{proof}
    It will be convenient to define the scalars
    \begin{align}
    U = \int_{t}^{t + \tau} u(s) ds, \quad K_i = \int_{t}^{t + \tau} \kappa_i(s) ds, \quad \text{and }K = K_1 + K_2.
    \end{align}
    We also note that, since both the Hamiltonian and the jump operators are expressible as polynomials of the fermionic number operators $n_1, n_2$, they commute with each other. Therefore, \begin{align}\label{eq:non_gaussian_channel}
    \mathcal{T}\exp\bigg(\int_{t}^{t + \tau}\mathcal{L}(s) ds\bigg) = \exp\big(-iU[n_1 n_2, \ \cdot \ ]\big)\exp\big(K_1 \mathcal{D}_{n_1}\big) \exp\big(K_2 \mathcal{D}_{n_2}\big).
    \end{align}
    We define the channel $\mathcal{R}_{t + \tau, t}$ via
    \begin{align}\label{eq:convex_gaussian_channel}
        \mathcal{R}_{t + \tau, t}(\rho) = \mathbb{E}_z\big(R(z)\rho R^\dagger(z)  \big) \text{ where }R(z) = \exp\big(\sqrt{U}e^{-i\pi/4} (z n_1 + z^* n_2)\big),
    \end{align}
    where $z = (a + ib) / \sqrt{2}$ with $a, b$ being independent standard normal random variables. Note that, due to Lemma \ref{lemma:closure_gaussian_product}, $\mathcal{R}_{t + \tau, t}$ maps an input Gaussian state to a (possibly unnormalized) convex-Gaussian state.
    
        We now explicitly compute $\mathcal{R}_{t + \tau, t}(\rho)$. define $\mathcal{N}_{i, l}$ as the superoperator which left multiplies by $n_i$ (i.e.~$\mathcal{N}_{i, l}(\rho) = n_i \rho$) and $\mathcal{N}_{i, r}$ as the superoperator which right multiplies by $n_i$ (i.e.~$\mathcal{N}_{i, r}(\rho) = \rho n_i$). Then
    \begin{align}
        \mathcal{R}_{t + \tau, t} &= \mathbb{E}_z\big(\exp\big(\sqrt{U}e^{-i\pi / 4}(z \mathcal{N}_{1, l} + z^* \mathcal{N}_{2, l}) + \sqrt{U}^* e^{i\pi / 4}(z^* \mathcal{N}_{1, r} + z \mathcal{N}_{2, r})\big)\big) \nonumber \\
        &=\mathbb{E}_a\bigg(\exp\bigg(\frac{a}{\sqrt{2}}\big(\sqrt{U}e^{-i\pi/4}(\mathcal{N}_{1, l} + \mathcal{N}_{2, l}\big) + \sqrt{U}^* e^{i\pi / 4}(\mathcal{N}_{1, r} + \mathcal{N}_{2, r}\big)\big)\bigg)\bigg) \times \nonumber \\
        &\qquad \qquad \mathbb{E}_b\bigg(\exp\bigg(\frac{ib}{\sqrt{2}}\big(\sqrt{U}e^{-i\pi/4}(\mathcal{N}_{1, l} - \mathcal{N}_{2, l}\big) - \sqrt{U}^* e^{i\pi / 4}(\mathcal{N}_{1, r} - \mathcal{N}_{2, r}\big)\big)\bigg)\bigg) \nonumber \\
        &\numeq{1}\exp\bigg(-i\frac{U}{4}(\mathcal{N}_{1, l} + \mathcal{N}_{2, l})^2 + i\frac{U}{4}(\mathcal{N}_{1, r} + \mathcal{N}_{2, r})^2 + \frac{\abs{U}}{2}(\mathcal{N}_{1, l} + \mathcal{N}_{2, l})(\mathcal{N}_{1, r} + \mathcal{N}_{2, r})\bigg) \times \nonumber \\
        &\qquad \qquad \exp\bigg(i\frac{U}{4}(\mathcal{N}_{1, l} - \mathcal{N}_{2, l})^2 - i\frac{U}{4}\big(\mathcal{N}_{1, l} - \mathcal{N}_{2, r}\big)^2 + \frac{\abs{U}}{2}(\mathcal{N}_{1, l} - \mathcal{N}_{2, l})(\mathcal{N}_{1, r} - \mathcal{N}_{2, r}\big)\bigg) \nonumber \\
        &=\exp\big(-iU(\mathcal{N}_{1, l}\mathcal{N}_{2, l} -\mathcal{N}_{1, r}\mathcal{N}_{2, r} ) + \abs{U}(\mathcal{N}_{1, l}\mathcal{N}_{1, r} + \mathcal{N}_{2, l}\mathcal{N}_{2, r})\big),
    \end{align}
where, in (1), we have used the fact that, for any operator $O$, $\mathbb{E}_{x \in \mathcal{N}(0, 1)}(e^{xO}) = e^{O^2/2}$. Identifying $\mathcal{N}_{1, l}\mathcal{N}_{2, l} - \mathcal{N}_{1, r}\mathcal{N}_{2, r} = [n_1 n_2, \ \cdot \ ]$, we obtain that
\begin{align}\label{eq:hamiltonian_in_terms_of_R}
    \exp(-iU[n_1 n_2, \ \cdot \ ]) = \exp(-\abs{U}(\mathcal{N}_{1, l}\mathcal{N}_{1, r}+ \mathcal{N}_{2, l}\mathcal{N}_{2, r}))\mathcal{R}_{t + \tau, t}.
\end{align}
Using Eq.~(\ref{eq:hamiltonian_in_terms_of_R}) and Eq.~(\ref{eq:non_gaussian_channel}) together with the fact that $\mathcal{D}_{n_i} = \mathcal{N}_{i, l} \mathcal{N}_{i, r} - (\mathcal{N}_{i, l}^2 + \mathcal{N}_{i, r}^2)/2$, we obtain that
\begin{align}\label{eq:lind_cp_decomp}
    \mathcal{T}\exp\bigg(\int_{t}^{t + \tau}\mathcal{L}(s) ds\bigg) = \bigg(\prod_{i \in \{1, 2\}}\underbrace{\exp((K_i - \abs{U})\mathcal{N}_{i, l}\mathcal{N}_{i, r})}_{\mathcal{E}_i}\underbrace{\exp(-(\mathcal{N}_{i, l}^2 + \mathcal{N}_{i, r}^2)/2)}_{\mathcal{F}_i}\bigg)\mathcal{R}_{t + \tau, t}.
\end{align}
We note that $\mathcal{R}_{t + \tau, t}$ and $\mathcal{F}_i$ are completely positive maps that map convex Gaussian states to possibly unnormalized Gaussian states. Furthermore, if $K_i \geq \abs{U}$, which is implied by $\kappa_i(t) \geq \abs{u(t)}$ quoted in the lemma statement, then $\mathcal{E}_i$ also have this property. Consequently, since $\mathcal{T}\exp(\int_{t}^{t + \tau}\mathcal{L}(s) ds)$ is a channel, as long as $K_i \geq \abs{U}$, it maps convex Gaussian states to (normalized) convex Gaussian states.
\end{proof}
\begin{theorem}[High-noise convex Gaussianity and classical simulation of the fermionic model, reproduced from the main text]
For an initial Gaussian state, if $\kappa_3 \geq 2U$, then the state of the fermionic model at time $t$, $\rho(t)$, is convex Gaussian for all $t \geq 0$. Furthermore, $\rho(t)$ can be classically sampled in the Fock state basis to an $\epsilon$ total variation error in $O(m^7 \Lambda^2 t^2 / \epsilon)$ time.
\end{theorem}
\begin{proof} Consider the Trotterized state $\sigma_T$ [Eq.~(\ref{eq:Trotter_theorem_2_ferm})]---note that $\Phi^\text{g}_{\tau \delta, (\tau - 1) \delta}$ is a Gaussian channel and hence trivially preserves convex Gaussianity. We now obtain the condition under which $\Phi^\text{ng}_{\tau \delta, (\tau - 1)\delta}$ also preserves convex Gaussianity using Lemma \ref{lemma:convex_gaussian_fermions}. We first perform the decomposition
\begin{align}
    \Phi^\text{ng}_{\tau \delta, (\tau - 1) \delta} = \prod_{v, u} \mathcal{T}\exp\bigg(\int_{(\tau - 1)\delta}^{\tau \delta}\mathcal{L}_{v, u}(s) ds\bigg) \text{ where }\mathcal{L}_{v, u} = -i[U_{v, u}(t) n_v n_u, \ \cdot \ ] + \kappa_3 \big(p_{v, u}(t) \mathcal{D}_{n_v} + q_{v, u}(t) \mathcal{D}_{n_u}\big),
\end{align}
where we choose
\begin{align}
    p_{v, u}(t) = \frac{1}{2}\frac{\abs{U_{v, u}(t)}}{\sum_{u} \abs{U_{v, u}(t)}} \text{ and }q_{v, u}(t) = \frac{1}{2}\frac{\abs{U_{v, u}(t)}}{\sum_{v} \abs{U_{v, u}(t)}}.
\end{align}
Next, we apply Lemma \ref{lemma:convex_gaussian_fermions}: For $\mathcal{L}_{v, u}(t)$ to generate a channel that is convex-Gaussianity preserving, a sufficient condition is that
\begin{align}
\kappa_3 p_{v, u}(t), \kappa_3 q_{v, u}(t) \geq \abs{U_{v, u}(t)} \text{ or equivalently } \kappa_3 \geq 2\sum_{k'} \abs{U_{k, k'}(t)} \ \text{ for } k \in \{v, u\}.
\end{align}
Since $\sum_{k'}\abs{U_{k, k'}(t)} \leq U_C + U_\text{os}$, this condition is satisfied if $\kappa_3 \geq 2 (U_C + U_\text{os})$. Assuming this to be true, it then follows from Lemma \ref{lemma:convex_gaussian_fermions} that $\Phi^\text{ng}_{\tau \delta, (\tau - 1) \delta}$ maps an input Gaussian state to a convex-Gaussian state---consequently, the Trotterized state $\sigma_T$ is a convex-Gaussian state such that $\norm{\rho(t) - \sigma_T}_1 \leq \epsilon$ when $T =\Theta(t^2 m^2 \Lambda^2 / \epsilon)$.

\emph{Time-complexity of sampling in the Fock state basis}. Since $\sigma_T$ is convex-Gaussian by construction, it can be expressed as $\sigma_T = \int \rho_{\alpha} d\mu(\alpha)$,
where $\rho_\alpha$ is a Gaussian state and $\mu$ is a probability measure. To sample from $\sigma_T$, we can then first sample from $\mu$ to obtain a Gaussian state and then use the standard algorithm for sampling from fermionic Gaussian states. Consider sampling from $\mu(\alpha)$: Suppose the initial state $\rho(0)$ is a Gaussian state. Lemma \ref{lemma:convex_gaussian_fermions} provides an explicit characterization of the convex combination of Gaussian states that result when applying $\mathcal{T}\exp(\int_{(\tau -1)\delta}^{\tau \delta}\mathcal{L}_{v, u}(s) ds)$ on an input Gaussian state. Furthermore, since the covariance matrix of the Gaussian state is a $2m \times 2m$ matrix, the probabilities of each Gaussian state in the convex combination being computable from the result for covariance matrices of products of Gaussian states \cite{2022_tagliacozzo_fermionic_gs,Fagotti_2010_entanglement_fermions} in $O(m^3)$ time---sampling from this convex combination thus requires $O(m^3)$ time. At every time-step, this has to be done for every pair of fermionic modes to apply $\Phi^{\text{ng}}_{\tau \delta, (\tau - 1)\delta}$, thus yielding a total time-complexity of $O(m^5)$. The application of the Gaussian evolution in each time-step can also be done at the level of covariance matrices in $O(m^3)$ time. Thus, the total time of sampling from $\mu$ is given by $O(m^5 \times T) = O(m^7 \Lambda^2 t^2 / \epsilon)$. Finally, having sampled a Gaussian state $\rho_\alpha$ from $\sigma_T$, we can draw a sample in the Fock state basis in $O(m^3)$ time \cite{Terhal2002_classical_simulation_fermions,knill2001fermioniclinearopticsmatchgates}---the total time complexity of the sampling algorithm thus is dominated by the cost of sampling from $\mu$ and is given by $O(m^7 \Lambda^2 t^2 / \epsilon)$.
\end{proof}

\section{High-noise separability of the bosonic model (Theorem 2)}\label{supplemental:theorem_2_proof}
In this section, we will present proof of Theorem 2, which considers the high-noise regime of the bosonic model. Since the bosonic model is infinite-dimensional with unbounded terms in the Hamiltonian, its analysis first requires an analysis of the particle number (as well as its moments) in the model. We do so in the first subsection---then, in the proof of Theorem 2, we first approximate the infinite-dimensional bosonic modes with finite-dimensional qudits and quantify the approximation error. Finally, we analyze the resulting finite-dimensional model and establish high-noise separability in the model.

\subsection{Analyzing particle number moments}
We begin by introducing a physically motivated assumption on the initial state of the model---the initial state will be assumed to be a product state with a ``uniform particle moment density" assumption, similar to that used in Ref.~\cite{kuwahara2024bosons_liebrobinson}.
\begin{assumption} [Uniform particle moment density]\label{assump:uniform_particle_num}
    The initial state $\rho(0)$ is a product state and $\exists C_0, \alpha_0, \beta_0 > 0$ such that $\forall v$ and $k \in \{1, 2, 3 \dots\}$
    \begin{align*}
    \textnormal{Tr}(n_v^k \rho(0))\leq C_0^k k^{\alpha_0 k + \beta_0}. 
    \end{align*}
\end{assumption}
\noindent As shown in Ref.~\cite{kuwahara2024bosons_liebrobinson}, this assumption is satisfied for a wide variety of physically relevant initial states of the bosonic model, notably for the vacuum state, thermal states, as well as coherent states. Furthermore, it implies a bound on the moments of the total particle number $N = \sum_v n_v$ since
\begin{align}\label{eq:particle_number_moments_init_state}
    \text{Tr}(N^k \rho(0)) = \sum_{v_1, v_2 \dots v_k }\text{Tr}(n_{v_1} n_{v_2} \dots n_{v_k} \rho(0))\leq \sum_{v_1, v_2 \dots v_k} \prod_{i = 1}^k \text{Tr}(n_{v_i}^k \rho(0))^{1/k} \leq (C_0 m)^k k^{\alpha_0 k + \beta_0},
\end{align}
where we remind the reader that $m = nL$ is the total number of bosonic modes in the model. This particle number moment bound, in turn, implies that the probability of high-particle-number states being occupied is exponentially suppressed, which we make precise in the following lemma.
\begin{lemma}[Probability of high-particle-number states (Ref.~\cite{kuwahara2024bosons_liebrobinson})]\label{lemma:chernoff}
Suppose $\rho$ is a state which satisfies $\textnormal{Tr}(N^k \rho) \leq (Cm)^k k^{\alpha k + \beta}$ and $\Pi_{\geq d}$ is a projector on the subspace with $\geq d$ particles, then
\begin{align*}
    \textnormal{Tr}(\Pi_{\geq d} \rho) \leq \bigg(\frac{d e^{\alpha/\beta}}{Cme}\bigg)^{\beta/\alpha}e^{-(d/Cm e)^{1/\alpha}}.
\end{align*}
\end{lemma}
\begin{proof}
    Note that, for any $k > 0$, 
    \begin{align}
    d^k \text{Tr}(\Pi_{\geq d} \rho) \leq \text{Tr}(N^k \Pi_{\geq d} \rho)\leq \text{Tr}(N^k \rho) \leq (Cm)^k k^{\alpha k + \beta} \implies \text{Tr}(\Pi_{\geq d}\rho) \leq k^\beta\bigg(\frac{Cmk^\alpha}{d}\bigg)^k.
    \end{align}
    We can now pick $k$ to be the greatest integer smaller than $(d/ Cm e)^{1/\alpha}$---we then have that $(d/Cme)^{1/\alpha}- 1 \leq k \leq (d/Cme)^{1/\alpha}$ and therefore
    \begin{align}
    \text{Tr}(\Pi_{\geq d} \rho) \leq \bigg(\frac{d}{Cme}\bigg)^{\beta/\alpha} e^{-k} \leq \bigg(\frac{d e^{\alpha/\beta}}{Cme}\bigg)^{\beta/\alpha}e^{-(d/Cm e)^{1/\alpha}},
    \end{align}
    which proves the lemma statement.
\end{proof}

While we will assume that the uniform particle moment density condition holds for the initial state, the subsequent dynamics of the bosonic model could possibly violate this condition. In the remainder of this section, we show that under the condition that the total rate of particle loss is higher than the total rate of particle gain (which we make precise below in assumption \ref{assump:gain_loss}), the moments of the total particle number $\textnormal{Tr}(N^k \rho(t))$ satisfy an inequality similar to Eq.~(\ref{eq:particle_number_moments_init_state}), which by Lemma \ref{lemma:chernoff} implies that the probability of higher particle number states being occupied is super-polynomially small in the particle number.

\begin{assumption}\label{assump:gain_loss}
    The parameters $\kappa_1, \kappa_2$ and $\mathcal{G}$ are such that $2\gamma = \kappa_1 - \kappa_2 - 2\mathcal{G} > 0$.
\end{assumption}
Physically, this assumption restricts the rate of 3 processes in the bosonic model that can change its particle number: In the noise terms, incoherent particle loss  can decrease the particle number at a rate $\sim \kappa_1$ and incoherent particle gain  can increase the particle number at a rate $\sim \kappa_2$. Furthermore, in the Hamiltonian, the squeezing term ($H_\text{g}^\text{sq}(t)$ in Eq.~(\ref{eq:H_g_bosonic_model})) can also increase the number of particles in the system at a rate $\sim \mathcal{G}$. Assumption \ref{assump:gain_loss} constrains the model to have  particle loss higher than  particle gain, without which the number of particles can increase arbitrarily with time. We remark that we do not need any assumption on the strength of displacement term ($H_\text{g}^\text{disp}(t)$ in Eq.~(\ref{eq:H_g_bosonic_model}))---we will show in Lemma \ref{lemma:particle_number_bound} that, as long as assumption \ref{assump:gain_loss} is satisfied, no matter how large the displacement term is, the particle number (and its moments) do not grow arbitrarily with time.

  We begin with a two technical lemmas that will be useful in our analysis.
\begin{lemma}\label{lemma:particle_num_inequality}
Suppose $x_k(t)$, for $k \in \{0, 1, 2 \dots \}$, are non-negative functions of time which satisfy the differential inequalities
\begin{align*}
\frac{d}{dt}x_k(t) \leq -\gamma k x_k(t) + \lambda m k \sum_{q = 1}^k 2^q {k \choose q} x_{k - q}(t),
\end{align*}
where $\gamma, \lambda, m > 0$. Furthermore, suppose $x_0(t) = 1$  $\forall t \geq 0$ and $\exists  C_0, \alpha_0, \beta_0 > 0 : x_k(0) \leq (C_0 m)^k k^{\alpha_0 k + \beta_0}$ for all $k \in \{1, 2, 3\dots \}$. Then
\begin{align*}
x_k(t) \leq (C m) e^{\alpha k + \beta}, \text{ where } C = e^{\lambda / \gamma}(C_0 + 2), \alpha = \textnormal{max}(\alpha_0, 1) \text{ and }\beta = \beta_0.
\end{align*}
\end{lemma}
\begin{proof}
The differential inequality can be written as an integral inequality:
\begin{align}\label{eq:integral_inequation}
x_k(t) \leq x_k(0)e^{-\gamma k t} + \lambda m k\sum_{q = 1}^k  2^q {k \choose q} \int_0^t x_{k - q}(s) e^{-\gamma k (t - s)}ds.
\end{align}
Recursing Eq.~(\ref{eq:integral_inequation}), we obtain
\begin{align}\label{eq:rec_int_ineq}
&x_k(t) \leq x_k(0) e^{-\gamma k t} + \sum_{p = 1}^k \sum_{q_1 = 1}^k \sum_{q_2 = 1}^{k - q_1}  \dots \sum_{q_p = 1}^{k -\sum_{i = 1}^{p-1}q_i} \frac{2^{\sum_{i=1}^p q_i} \lambda^p k!}{q_1!q_2!\dots q_p!(k - \sum_{i = 1}^{p}q_i)!}\bigg(\prod_{i = 1}^p \bigg(k - \sum_{j=1}^{i - 1}q_j\bigg)\bigg) x_{k - \sum_{i = 1}^p q_i}(0) I_{q_1, q_2 \dots q_p}^{(k)}(t),
\end{align}
where
\begin{align}
 I_{q_1, q_2 \dots q_p}^{(k)}(t) &= \int_0^t \int_0^{s_1} \dots \int_0^{s_{p - 1}} e^{-\gamma k(t - s_1)} e^{-\gamma (k - q_1)(s_1 - s_2)}e^{-\gamma (k - q_1 - q_2)(s_2 - s_3)} \dots e^{-\gamma (k - q_1-q_2 \dots -q_p)s_p} ds_1 ds_2 \dots ds_p\nonumber\\
 &=e^{-k\gamma t} \int_0^t \int_0^{s_1} \int_0^{s_2} \dots \int_0^{s_{p - 1}} e^{\gamma(q_1 s_1 + q_2 s_2 + \dots q_p s_p)} ds_1 ds_2 \dots ds_p.
\end{align}
We note that $I_{q_1, q_2 \dots q_p}^{(k)}(t)$ can be upper bounded:
\begin{align}\label{eq:upper_bound_I}
 I_{q_1, q_2 \dots q_p}^{(k)}(t) &\leq e^{-k\gamma t} \int_{-\infty}^t \int_{-\infty}^t \dots \int_{-\infty}^t  e^{\gamma(q_1 s_1 + q_2 s_2 + \dots q_p s_p)} ds_1 ds_2 \dots ds_p \nonumber \\
 &\leq \frac{1}{\gamma^p q_1 q_2 \dots q_p}e^{-(k - q_1 - q_2 - \dots q_p)\gamma t}\nonumber\\
 & \leq \frac{1}{\gamma^p} e^{-\gamma (t - q_1 - q_2 - \dots q_p)}.
\end{align}
In the calculation done below, it will be useful to note that, given any $f(n)$ where $n \in \{0, 1, 2\dots\}$,
\begin{align}\label{eq:perm_sum}
\sum_{q_1 = 1}^k \sum_{q_2 = 1}^{k - q_1}  \dots \sum_{q_p = 1}^{k - \sum_{i = 1}^{p - 1}q_i} \frac{1}{q_1 !q_2 ! \dots q_p !}f\bigg({\sum_{i = 1}^p q_i}\bigg) &\numeq{1}\sum_{q=p}^k \sum_{q_1 = 1}^q \sum_{q_2 = 1}^{q - q_1 }\dots \sum_{q_{p} = 1}^{q - \sum_{i = 1}^{p - q}q_i}\frac{1}{q_1 !q_2 ! \dots q_p!}f(q) \nonumber\\
&\numeq{2}\sum_{q=p}^k f(q) \sum_{\substack{q_1, q_2 \dots q_p \geq 1 \\ q_1 + q_2 + \dots q_p = q}} \frac{1}{q_1 ! q_2 ! \dots q_p !} \nonumber \\
&\numleq{3}\sum_{q=p}^k \frac{p^q f(q)}{q!},
\end{align}
where in (1) we have introduced the index $q = q_1 +q_2 +\dots q_p$, which ranges from $p$ to $k$, and re-expressed the summation over $q_1, q_2 \dots q_p$ as first a sum over $q$, and then a sum over $q_1, q_2 \dots q_p$ subject to the contraint $q_1 + q_2 +\dots q_p = q$. In (2), we have simply noted the fact that the summation over $q_1 \in \{1, 2 \dots q\}, q_2\in\{1, 2 \dots q-q_1\}\dots q_{p -1}\in\{1, 2 \dots q-(q_1+q_2+q_{p-1})\}$ is identical to summation over $q_1, q_2 \dots q_p \in\{1, 2\dots q\}$ with the additional constraint that $q_1 + q_2 + \dots q_p\leq q$. Finally, (3) is obtained by identifying the summation as a multinomial sum. 

Returning to Eq.~(\ref{eq:rec_int_ineq}), we obtain that
\begin{align}
x_k(t) &\numleq{1} x_k(0) e^{-\gamma k t} +  \sum_{p = 1}^k \sum_{q_1 = 1}^k \sum_{q_2 = 1}^{k}  \dots \sum_{q_p = 1}^{k} \frac{k! 2^{q_1 + q_2 \dots q_p}}{q_1 ! q_2 ! \dots q_p ! (k - \sum_{i = 1}^p q_i)! p!} \bigg(\frac{\lambda m k}{\gamma}\bigg)^p x_{k - \sum_{i = 1}^p q_i}(0)  e^{-\gamma t (k - \sum_{i = 1}^p q_i)}\nonumber \\
&\numleq{2} x_k(0) e^{-\gamma k t} + \sum_{p = 1}^k \sum_{q = p}^k (2p)^q {k \choose q}\bigg(\frac{\lambda m k}{\gamma}\bigg)^p x_{k - q}(0)  e^{-\gamma t (k - q)} \nonumber \\
&\leq (C_0 m)^k k^{\alpha_0 k + \beta_0} e^{-\gamma k t} + \sum_{q = 1}^k \sum_{p = 1}^q (2k)^q {k \choose q}\bigg(\frac{\lambda m k}{\gamma}\bigg)^p (C_0 m)^{k - q} (k - q)^{\alpha_0 (k - q) + \beta_0} e^{-\gamma t (k - q)} \nonumber \\
&\leq (C_0 m)^k k^{\alpha_0 k + \beta_0} e^{-\gamma k t} + m^k e^{k\lambda / \gamma} \sum_{q = 1}^k {k \choose q} (2k)^q C_0^{k - q} k^{\alpha_0(k - q) + \beta_0} e^{-\gamma t (k - q)} \nonumber\\
&\leq (e^{\lambda / \gamma} m)^k k^{\beta_0} \sum_{q = 1}^k {k \choose q} (C_0 k^{\alpha_0}e^{-\gamma t})^{k - q} (2k)^q = (e^{\lambda / \gamma} m)^k k^{\beta_0} (C_0 k^{\alpha_0}e^{-\gamma t} + 2k)^k,
\end{align}
where, in (1), we have used Eq.~(\ref{eq:upper_bound_I}) and in (2) we have used Eq.~(\ref{eq:perm_sum}). Finally, using $C_0 k^{\alpha_0}e^{-\gamma t} + 2k \leq k^{\max(\alpha_0, 1)}(C_0 + 2)$, the lemma statement follows.
\end{proof}
\begin{lemma}\label{lemma:a_moment_comm}
For any $k > 0, v$,
\begin{align*}
[a_v, N^k] = a_v (N^k - (N - I)^k) \text{ and } [a_v, N^k] = ((N + I)^k - N^k) a_v.
\end{align*}
Furthermore, for any $k>0, v$,
\begin{align*}
a_v^\dagger N^k a_v \preceq n_v N^k.
\end{align*}
\end{lemma}
\begin{proof}
We begin by noting that, for any $z$, it follows from $e^{zN} a_v e^{-zN}=e^{-z} a_v$ that
\begin{align}
[a_v, e^{z N}] = a_v e^{z N} - e^{z N} a_v = a_v (e^{zN} - e^{z(N - I)}) = (e^{z(N + I)}-e^{zN}) a_v.
\end{align}
We thus obtain that
\begin{align}
[a_v, N^k] = \frac{d^k}{dz^k} [a_v, e^{zN}] \bigg |_{z = 0} = a_v \big(N^k - (N- I)^k \big) = \big((N+ I)^k - N^k ) a_v.
\end{align}
Furthermore, for any state $\ket{\psi} = \sum_{\vec{n}} \psi_{\vec{n}} \ket{\vec{n}}$, where $\psi_{\vec{n}}$ is the amplitude of $\ket{\psi}$ on the basis state $\ket{\vec{n}} = \ket{n_1, n_2 \dots n_m}$,
\begin{align}
\bra{\psi} a_v^\dagger N^k a_v\ket{\psi} = \sum_{\vec{n}} \abs{\psi_{\vec{n}}}^2 n_v (\norm{\vec{n}}_1 - 1)^k \leq  \sum_{\vec{n}} \abs{\psi_{\vec{n}}}^2 n_v \norm{\vec{n}}_1^k = \bra{\psi} n_v N^k \ket{\psi},
\end{align}
from which it follows that $a_v^\dagger N^k a_v \preceq n_v N^k $.
\end{proof}
\noindent In the next lemma, we derive an upper bound on $\textnormal{Tr}(N^k \rho(t))$, which will be central to analyzing the Hilbert space truncation and Trotter bounds in the subsequent subsections.
\begin{lemma}[Upper bounding particle number moments]\label{lemma:particle_number_bound}
    Consider a bosonic model satisfying assumption \ref{assump:gain_loss} with the bosonic modes in an initial state $\rho(0)$ satisfying assumption \ref{assump:uniform_particle_num}, then, $\forall t\geq 0$,
    \begin{align*}
    \textnormal{Tr}(N^k \rho(t)) \leq (Cm)^k k^{\alpha k + \beta},
    \end{align*}
    where $C = e^{1 + 4\Omega^2/\gamma^2 + 2\mathcal{G}/\gamma + 4(\kappa_1 + \kappa_2)/\gamma}(C_0 + 2), \alpha = \max(\alpha_0, 1)$ and $\beta = \beta_0$ with $C_0, \alpha_0, \beta_0$ defined in assumption \ref{assump:gain_loss}.
\end{lemma}
\begin{proof}
    We will use the Heisenberg equations of motion for the operator $N^k$. Note that $[N^k, H_\text{ng}(t)] = 0, [N^k, H_\text{g}^\text{hop}(t)] = 0$ and $\mathcal{D}_{n_v}^\dagger(N^k) = 0$ (where $H_\text{g}^\text{hop}(t)$ is defined in Eq.~(\ref{eq:H_g_bosonic_model})). Using notation $\expect{O}_t = \text{Tr}(O \rho(t))$, we then have that
    \begin{align}\label{eq:particle_number_decomp}
    \frac{d}{dt}\expect{N^k}_t = \underbrace{  \sum_v \big(\kappa_1 \expect{\mathcal{D}_{a_{v}}^\dagger(N^k)}_t + \kappa_2\expect{\mathcal{D}_{a_{v}^\dagger}^\dagger(N^k)}_t \big)}_{\expect{\mathcal{L}_\text{n}^\dagger(N^k)}_t} - i\expect{[N^k, H_\text{g}^\text{sq}(t)]}_t  - i\expect{[N^k, H_\text{g}^\text{disp}(t)]}_t.
    \end{align}
    Consider first $\mathcal{D}_{a_v}^\dagger(N^k), \mathcal{D}_{a_v^\dagger}^\dagger(N^k)$---using Lemma~\ref{lemma:a_moment_comm}, we obtain that
\begin{align}
&\sum_v \mathcal{D}_{a_v}^\dagger(N^k) = - \sum_v a_v^\dagger [a_v, N^k] = -N \big(N^k - (N - I)^k\big) = -kN^k + \sum_{l \geq 1} (-1)^{l} {k \choose l + 1} N^{k - l}, \nonumber\\
&\sum_v \mathcal{D}_{a_v^\dagger}^\dagger(N^k) = \sum_v [a_v, N^k] a_v^\dagger = \big((N + I)^k - N^k \big) (N + I) = kN^k + \sum_{l \geq 1} {k + 1 \choose l + 1} N^{k - l}.
\end{align}
Therefore,
\begin{align}\label{eq:derivative_noise}
\expect{\mathcal{L}_\text{n}^\dagger(N^k)}_t &= -(\kappa_{1} - \kappa_2) k \expect{N^k}_t + \sum_{l = 1}^{k} \bigg(\kappa_1 (-1)^l {k \choose l + 1} + \kappa_2 {k + 1\choose l + 1} \bigg) \expect{N^{k - l}}_t\nonumber\\
&\numleq{1} -(\kappa_1 - \kappa_2) k \expect{N^k}_t  + \sum_{l = 1}^k \big(\kappa_1 k + \kappa_2 (k + 1)\big) {k \choose l} \expect{N^{k - l}}_t \nonumber \\
&\leq  - (\kappa_1 - \kappa_2) k \expect{N^k}_t + 2 (\kappa_1 + \kappa_2) km\sum_{l = 1}^k 2^l {k \choose l} \expect{N^{k - l}}_t, 
\end{align}
where we implicitly set ${k \choose l} = 0$ if $l < 0$ or $l > k$ and in (1) we have used the fact that ${k \choose l + 1} \leq k {k \choose l}, {k + 1 \choose l + 1} \leq (k + 1) {k \choose l}$. 

Next, consider $[N^k, H_\text{g}^{\text{sq}}(t)] = \sum_{v, u} \mathcal{G}_{v, u}(t) [N^k, a_v a_u]  - \text{h.c.}$ --- we begin by noting that from Lemma \ref{lemma:a_moment_comm}
\begin{align}\label{eq:comm_squeezing}
[N^k, a_v a_u] = -[a_v, N^k] a_u - a_v [a_u, N^k] = -2\sum_{l\geq 0}  {k \choose 2l + 1} a_v N^{k - 2l - 1} a_u ,
\end{align}
and therefore
\begin{align}
\abs{\expect{[N^k, H_\text{g}^{\text{sq}}(t)]}_t} &\leq 2\sum_{v, u} \smallabs{\mathcal{G}_{v, u}(t)} \smallabs{\expect{[N^k, a_v a_u]}_t} \nonumber\\
&\numleq{1} 2\sum_{v, u} \sum_{l\geq 0} \abs{\mathcal{G}_{v, u}(t)} {k \choose 2l + 1}\smallabs{\expect{a_v N^{k - 2l - 1}a_u}_t} \nonumber \\
&\numleq{2} \sum_{v, u} \sum_{l\geq 0} \abs{\mathcal{G}_{v, u}(t)} {k \choose 2l + 1}\big({\smallabs{\expect{a_v N^{k - 2l - 1}a_v^\dagger}_t} +  \smallabs{\expect{\text{Tr}( a_u^\dagger N^{k - 2l - 1}a_u}_t}}\big) \nonumber \\
&\leq \mathcal{G}\sum_u \sum_{l \geq 0}  {k \choose 2l + 1} \big(\expect{a_u N^{k - 2l - 1}a_u^\dagger}_t +  \expect{a_u^\dagger N^{k - 2l - 1}a_u }_t\big)\nonumber \\
&\numleq{3} \mathcal{G} \sum_u \sum_{l \geq 0} {k \choose 2l + 1}\big(\expect{(N + I)^{k - 2l - 1} a_u a_u^\dagger}_t +  \expect{(N - I)^{k - 2l - 1} a_u^\dagger a_u}_t\big) \nonumber \\
&\leq 2\mathcal{G}  \sum_{l, p \geq 0} {k \choose 2l + 1} {k - 2l - 1 \choose 2p} \expect{N^{k - 2l - 2p}}_t + 2\mathcal{G} m \sum_{l, p \geq 0} {k \choose 2l + 1} {k- 2l - 1 \choose p} \expect{N^{k - 2l - p - 1}}_t.
\end{align}
where, in (1), we have used Eq.~(\ref{eq:comm_squeezing}), in (2) we have used the fact that, for any two operators $A, B$, $\expect{AB} \leq \sqrt{\expect{AA^\dagger}\expect{B^\dagger B}}\leq (\expect{AA^\dagger} + \expect{B^\dagger B})/2$ and in (3) we have used Lemma \ref{lemma:a_moment_comm}. We can thus conclude that
\begin{align}
\abs{\expect{[N^k, H_\text{g}^\text{sq}(t)]}_t} \leq 2\mathcal{G} k \expect{N^k}_t + 2\mathcal{G} \sum_{q \geq 1} f^{(k)}_q \expect{N^{k - q}}_t, 
\end{align}
where
\begin{align}
f_q^{(k)} = \begin{cases}
m \sum_{l\geq 0} {k \choose 2l + 1} {k - (2l + 1) \choose q - (2l + 1)} & \text{ if } q \in \{1, 3, 5 \dots \}, \\
m \sum_{l\geq 0} {k \choose 2l + 1} {k - (2l + 1) \choose q - (2l + 1)} + \sum_{l \geq 0} {k \choose 2l + 1} { k - (2l + 1) \choose q - 2l } & \text{ if } q\in \{2, 4, 6 \dots \}.
\end{cases}
\end{align}
The expression for $f_q^{(k)}$ can be further simplified by noting that
\begin{align}
\sum_{l\geq 0} {k \choose 2l + 1} {k - (2l + 1) \choose q - (2l + 1)} &= \sum_{l \geq 0} \frac{k!}{(2l + 1)! (k - q)! (q - (2l + 1))! } ={k \choose q} \sum_{l \geq 0} {q \choose 2l + 1} = 2^{q - 1}{k \choose q},
\end{align}
and
\begin{align}
\sum_{l\geq 0} {k \choose 2l + 1} {k - (2l + 1) \choose q - 2l} &= \sum_{l \geq 0} \frac{k!}{(2l + 1)! (k - q - 1)! (q - 2l)! } ={k \choose q + 1} \sum_{l \geq 0} {q + 1 \choose 2l + 1} = 2^{q}{k \choose q + 1}.
\end{align}
We then obtain that
\begin{align}
f_q^{(k)} = \begin{cases}
2^{q - 1}m {k \choose q} & \text{ if } q \in \{1, 3, 5 \dots \}, \\
2^{q - 1}m {k \choose q}  + 2^q {k \choose q + 1} & \text{ if } q \in \{2, 4, 6 \dots\}.
\end{cases}
\end{align}
Again, we note that, since ${k \choose q + 1} \leq k {k \choose q}$, it follows that $f_q^{(k)} \leq 2^{q - 1}(m + 2k) {k \choose q} \leq mk 2^q {k \choose q}$, and thus we obtain that
\begin{align}\label{eq:derivative_squeezing}
    \abs{\expect{[N^k, H_\text{g}^\text{sq}(t)]}_t} \leq 2\mathcal{G}k \expect{N^k}_t + mk\mathcal{G} \sum_{q \geq 1} {k \choose q}2^q \expect{N^{k - q}}_t.
\end{align}

Finally, we consider $[N^k, H_\text{g}^\text{disp}(t)] = \sum_v \mathcal{D}_v(t) [N^k, a_v] - \text{h.c.}$---we begin by noting that, from Lemma \ref{lemma:a_moment_comm},
\begin{align}
[N^k, a_v] = ((N + I)^k - N^k) = \sum_{q \geq 1}  {k \choose q} N^{k - q} a_v,
\end{align} 
and therefore
\begin{align}
\abs{\expect{[N^k, H_2(t)]}_t} &\leq 2\Omega \sum_v \sum_{q \geq 1} {k \choose q} \abs{\expect{ N^{k - q} a_v}_t} \nonumber\\
&\numleq{1}  \sum_v \sum_{q\geq 1} {k \choose q} \sqrt{\gamma \expect{a_v^\dagger N^{k - q} a_v}_t\times  \frac{4\Omega^2}{\gamma}\expect{N^{k - q}}_t} \nonumber \\
&\leq  \sum_v \sum_{q \geq 1} {k \choose q} \bigg(\frac{\gamma}{2} \expect{a_v^\dagger N^{k - q} a_v}_t + \frac{2\Omega^2}{\gamma}\expect{N^{k - q} }_t\bigg) \nonumber \\
&\numleq{2} \sum_v \sum_{q \geq 1} {k \choose q} \bigg(\frac{\gamma}{2} \expect{N^{k - q} n_v}_t + \frac{2\Omega^2}{\gamma}\expect{N^{k - q} }_t\bigg) \nonumber \\
&\leq \frac{\gamma}{2} k  \expect{N^k}_t + \sum_{q \geq 1} \bigg( \frac{\gamma}{2} {k \choose q + 1} + \frac{2m\Omega^2}{\gamma}{k \choose q} \bigg) \expect{N^{k - q}}_t \nonumber\\
&\numleq{3}\frac{\gamma}{2} k \expect{N^k}_t + \sum_{q\geq 1} \bigg( \frac{\gamma}{2}k + \frac{2m \Omega^2}{\gamma}\bigg){k \choose q} \expect{N^{k - q}}_t,
\end{align}
where, in (1), we have again used that $\expect{AB}\leq \sqrt{\expect{AA^\dagger}\expect{B^\dagger B}}$ and introduced the parameter $\gamma = \kappa_1 - \kappa_2 -2\mathcal{G}$ from assumption \ref{assump:gain_loss}, in (2) we have used Lemma \ref{lemma:a_moment_comm} to obtain that $\expect{a_v^\dagger N^{k - q} a_v} \leq \expect{N^{k - q} n_v}$, and in (3) we have used the fact that ${k \choose q + 1} \leq k {k \choose q}$. Setting $k, m \leq km$, we obtain that
\begin{align}\label{eq:derivative_displacement}
    \abs{\expect{[N^k, H_\text{g}^\text{disp}(t)]}_t} \leq \frac{\gamma}{2} k\expect{N^k}_t + mk\bigg(\frac{\gamma}{2} + \frac{2\Omega^2}{\gamma}\bigg)\sum_{q\geq 1}{k \choose q} \expect{N^{k - q}}_t.
\end{align}
Combining Eq.~(\ref{eq:particle_number_decomp}) with Eqs.~(\ref{eq:derivative_noise}, \ref{eq:derivative_squeezing}, \ref{eq:derivative_displacement}), we obtain that
\begin{align}
    \frac{d}{dt}\expect{N^k}_t \leq -\frac{\gamma }{2}k \expect{N^k}_t + \lambda k m \sum_{q\geq 1}\bigg(\frac{k}{q}\bigg) \expect{N^{k - q}}_t,
\end{align}
where $\lambda = \gamma/2 +2\Omega^2 / \gamma + \mathcal{G}+ 2(\kappa_1 + \kappa_2)$. Then, solving this inequality using Lemma \ref{lemma:particle_num_inequality}, we obtain the lemma statement.
\end{proof}
\noindent Combining this lemma with Lemma \ref{lemma:chernoff}, we straightforwardly obtain the following lemma upper bounding the probability of large number of excitations at any time in the bosonic model.
\begin{lemma}\label{lemma:chernoff_model}
Suppose $\Pi_{\geq d}$ is a projector on the subspace with $\geq d$ particles and the bosonic model satisfies assumptions \ref{assump:uniform_particle_num} and \ref{assump:gain_loss}, then for any $t \geq 0$,
\begin{align*}
    \textnormal{Tr}(\Pi_{\geq d}\rho(t)) \leq e \bigg(\frac{d}{d_0 m}\bigg)^{k_0} \exp\bigg(-\bigg(\frac{d}{d_0 m}\bigg)^{1/\alpha}\bigg),
\end{align*}
where $d_0 = eC$, $k_0 = \beta / \alpha$ with $C, \alpha, \beta$ being defined in Lemma \ref{lemma:particle_number_bound}.
\end{lemma}

\subsection{Proof of Theorem 2 (bosons)}
\noindent The proof of Theorem 2 has three main parts:
\begin{enumerate}
    \item[(1)] Truncation of the Hilbert space of the bosonic model to a finite-dimensional space and an analysis of the truncation error (Lemma \ref{lemma:truncation_error_bosonic}).
    \item[(2)] First-order Trotterization of the truncated finite-dimensional model (Lemma \ref{lemma:Trotter_bound_bosons}).
    \item[(3)] Analysis of each Trotter step to establish high-noise separability (Lemma \ref{lemma:separability_bosons}).
\end{enumerate}

\emph{Truncation of the bosonic model}. Suppose we want to truncate the local Hilbert space of each bosonic mode to $d$ levels---we will denote by $\mathcal{H}_{\leq d}$ the Hilbert space of the bosonic model with each bosonic mode truncated to at most $d$ particles. For the $v^\text{th}$ bosonic mode, we will define the projectors $\Pi_{v, d}, \Pi_{v, \leq d}$, and $\Pi_{v, > d}$ via
\begin{align}
\Pi_{v, d} = \ket{d}\!\bra{d}, \Pi_{v, \leq d} = \sum_{j = 0}^{d}\Pi_{v, j} \text{, and }\Pi_{v, >d} = \sum_{j = d + 1}^\infty \Pi_{v, j}.
\end{align}
We will define the projector $\Pi_{\leq d} = \otimes_{v} \Pi_{v, \leq d}$, which will be the projector onto $\mathcal{H}_{\leq d}$. The truncated model will be described by a Lindbladian $\mathcal{L}_{\leq d}(t)$ while
\begin{subequations}
    \begin{align}
        \mathcal{L}_{\leq d} = -i[H_{\leq d},\ \cdot \ ] +  \sum_{l = 1}^3 \sum_v \kappa_l \mathcal{D}_{L_{v, \leq d}^{(l)}},
    \end{align}
where
\begin{align}
    &H_{\leq d}(t) = \Pi_{\leq d} H(t) \Pi_{\leq d}, \nonumber\\
    &L_{v, \leq d}^{(1)} = a_{v, \leq d} = \Pi_{v, \leq d} a_v \Pi_{v, \leq d},\nonumber\\
    &L_{v, \leq d}^{(2)} = a_{v, \leq d}^\dagger=  \Pi_{v, \leq d} a_v^\dagger \Pi_{v, \leq d}, \nonumber\\
    &L_{v, \leq d}^{(3)} = n_{v, \leq d} =  \Pi_{v, \leq d} n_v \Pi_{v, \leq d}.
\end{align}
\end{subequations}
It will be convenient to define super-operators $\mathcal{P}_{\leq d}$ and $\mathcal{Q}_{\leq d}$ via
\begin{align}
    \mathcal{P}_{\leq d}(\rho) = \Pi_{\leq d} \rho \Pi_{\leq d} \text{ and }\mathcal{Q}_{\leq d} = \text{id} - \mathcal{P}_{\leq d}.
\end{align}
The super-operator $\mathcal{P}_{\leq d}$ projects an input density matrix onto  $\mathcal{H}_{\leq d}$. We first present a lemma that quantifies the error between the state $\rho(t)$ at time $t$ and the state obtained from the truncated evolution: $\rho_{\leq d}(t) = \mathcal{T}\exp(\int_0^t \mathcal{L}_{\leq d}(\tau)d\tau)(\mathcal{P}_{\leq d} \rho(0))$.
\begin{lemma}\label{lemma:adiabatic_elim}
    For any $d > 0$, it follows that
    \begin{align}
    &\norm{\rho(t) - \mathcal{T}\exp\bigg(\int_0^t \mathcal{L}_{\leq d}(\tau)d\tau\bigg)(\mathcal{P}_{\leq d} \rho(0))}_1 \nonumber\\
    &\qquad \qquad \leq \norm{\mathcal{Q}_{\leq d}\rho(t)}_1 + (d + 1)\sum_v \int_0^t \norm{\Pi_{v, d}\rho(s)}_1 ds + \int_0^t \norm{\mathcal{P}_{\leq d}\mathcal{L}(s) \mathcal{Q}_{\leq d}}_{\diamond} \norm{\mathcal{Q}_{\leq d}\rho(s)}_1 ds.\nonumber
    \end{align}
\end{lemma}
\begin{proof}
Using $\mathcal{P}_{\leq d} + \mathcal{Q}_{\leq d} = \text{id}$ together with the master equation ($d\rho(t) / dt = \mathcal{L}(t) \rho(t)$), we obtain that
\begin{subequations}
\begin{align}
&\frac{d}{dt}\mathcal{P}_{\leq d} \rho(t) = \mathcal{P}_{\leq d} \mathcal{L}(t) \mathcal{P}_{\leq d} \rho(t) + \mathcal{P}_{\leq d} \mathcal{L}(t) \mathcal{Q}_{\leq d} \rho(t),\label{eq:adiabatic_eq_P}\\
&\frac{d}{dt}\mathcal{Q}_{\leq d} \rho(t) = \mathcal{Q}_{\leq d} \mathcal{L}(t)  \mathcal{P}_{\leq d} \rho(t) + \mathcal{Q}_{\leq d} \mathcal{L}(t) \mathcal{Q}_{\leq d} \rho(t).
\end{align}
\end{subequations}
Furthermore, we note that, for any operator $X$ that is supported on the truncated subspace $\mathcal{H}_d$ (i.e. $X = \Pi_{\leq d} X \Pi_{\leq d}$), and defining $H_{\leq d}(t) = \Pi_{\leq d} H(t) \Pi_{\leq d}$, we have
\begin{align}
\mathcal{P}_{\leq d} \mathcal{L} \mathcal{P}_{\leq d} (X) &= -i [H_{\leq d}(t), X] +\Pi_{\leq d} \mathcal{L}_\text{n}  \Pi_{\leq d} (X),
\end{align}
\begin{subequations}
\begin{align}
\mathcal{P}_{\leq d} \mathcal{D}_{a_v} \mathcal{P}_{\leq d}(X) &= \Pi_{v, \leq d}a_v \Pi_{v, \leq d} X \Pi_{v, \leq d} a_v^\dagger \Pi_{v, \leq d} - \frac{1}{2} \big(\Pi_{v, \leq d} n_v \Pi_{v, \leq d} X + X \Pi_{v, \leq d} n_v \Pi_{v, \leq d}\big) \nonumber\\
&= a_{v, \leq d} X a_{v, \leq d}^\dagger - \frac{1}{2}\big(a_{v, \leq d}^\dagger a_{v, \leq d} X + X a_{v, \leq d}^\dagger a_{v, \leq d}\big) \nonumber \\
&=\mathcal{D}_{a_{v, \leq d}}(X), \\
\mathcal{P}_{\leq d} \mathcal{D}_{a_v^\dagger} \mathcal{P}_{\leq d}(X) &= \Pi_{v, \leq d}a_v^\dagger \Pi_{v, \leq d} X \Pi_{v, \leq d} a_v \Pi_{v, \leq d} - \frac{1}{2} \big(\Pi_{v, \leq d} a_v a_v^\dagger \Pi_{v, \leq d} X + X \Pi_{v, \leq d} a_v a_v^\dagger \Pi_{v, \leq d}\big) \nonumber\\
&= a_{v, \leq d}^\dagger X a_{v, \leq d}^{ \ } - \frac{1}{2}\big(a_{v, \leq d}^{ \ } a_{v, \leq d}^\dagger X + X a_{v, \leq d}^{ \ } a_{v, \leq d}^\dagger\big) - \frac{d + 1}{2}\big(\Pi_{v, d} X + X \Pi_{v, d}\big) \nonumber \\
&=\mathcal{D}_{a_{v, \leq d}^\dagger}(X) - \frac{d + 1}{2}\big(\Pi_{v, d} X + X \Pi_{v, d}\big), \\
\mathcal{P}_{\leq d} \mathcal{D}_{n_v} \mathcal{P}_{\leq d}(X) &= \Pi_{v, \leq d}n_v \Pi_{v, \leq d} X \Pi_{v, \leq d} n_v^2 \Pi_{v, \leq d} - \frac{1}{2} \big(\Pi_{v, \leq d} n_v^2 \Pi_{v, \leq d} X + X \Pi_{v, \leq d} n_v^2 \Pi_{v, \leq d}\big) \nonumber\\
&= n_{v, \leq d} X n_{v, \leq d} - \frac{1}{2}\big(n_{v, \leq d}^2 X + X n_{v, \leq d}^2 \big) \nonumber \\
&=\mathcal{D}_{n_{v, \leq d}}(X).
\end{align}
\end{subequations}
Defining $\mathcal{L}_{\text{n}, \leq d} = \sum_v \big(\kappa_1 \mathcal{D}_{a_{v, \leq {d}}} + \kappa_2 \mathcal{D}_{a_{v, \leq d}^\dagger} + \kappa_3 \mathcal{D}_{n_{v, \leq d}}\big)$, we then obtain that, $\forall X \in \mathcal{H}_{\leq d}$,
\begin{align}
\mathcal{P}_{\leq d} \mathcal{L}(t) \mathcal{P}_{\leq d} (X) = \mathcal{L}_{\leq d}(t)(X) - \frac{d + 1}{2}\sum_v \big( \Pi_{v, d} X + X\Pi_{v, d}\big).
\end{align}
Consequently, from Eq.~(\ref{eq:adiabatic_eq_P}), we obtain that
\begin{align}
    \frac{d}{dt}\mathcal{P}_{\leq d}\rho(t) = \mathcal{L}_{\leq d}(t) \mathcal{P}_{\leq d} \rho(t)+ \mathcal{P}_{\leq d}\mathcal{L}(t) \mathcal{Q}_{\leq d}\rho(t)- \frac{d+ 1}{2}\sum_v \bigg(\Pi_{v, d} \mathcal{P}_{\leq d}\rho(t) + (\mathcal{P}_{\leq d} \rho(t))\Pi_{v, d}\bigg),
\end{align}
which can be integrated to obtain
\begin{align}
\mathcal{P}_{\leq d} \rho(t) = \mathcal{E}_{\leq d}(t, 0) \mathcal{P}_{\leq d}\rho(0)+\int_0^t \mathcal{E}_{\leq d}(t, s)\bigg( \mathcal{P}_{\leq d} \mathcal{L}(s) \mathcal{Q}_{\leq d} \rho(s)  - \frac{d + 1}{2} \sum_v \big( \Pi_{v, d} \mathcal{P}_{\leq d}\rho(s) + (\mathcal{P}_{\leq d}\rho(s)) \Pi_{v, d}\big)\bigg) ds,
\end{align}
where $\mathcal{E}_{\leq d}(t, s) = \mathcal{T}\exp(\int_s^t \mathcal{L}_{\leq d}(\tau) d\tau)$. From here, it immediately follows that
\begin{align}
&\norm{\mathcal{E}_{\leq d}(t, 0)\mathcal{P}_{\leq d}\rho(0) - \rho(t)}_1 \\
&\qquad \leq \norm{\rho_{\leq d}(t) - \mathcal{P}_{\leq d}\rho(t)}_1 + \norm{\mathcal{Q}_{\leq d}\rho(t)}_1 \nonumber \\
&\qquad \leq (d + 1) \sum_v \int_0^t \norm{\Pi_{v, d}\mathcal{P}_{\leq d}\rho(s)}_1 ds + \int_0^t \norm{\mathcal{P}_{\leq d}\mathcal{L}(s) \mathcal{Q}_{\leq d} \rho(s)}_1 ds + \norm{\mathcal{Q}_{\leq d}\rho(t)}_1 \nonumber\\
&\qquad \numleq{1} (d + 1) \sum_v \int_0^t \norm{\Pi_{v, d}\rho(s)}_1 ds + \int_0^t \norm{\mathcal{P}_{\leq d}\mathcal{L}(s) \mathcal{Q}_{\leq d} \rho(s)}_1 ds + \norm{\mathcal{Q}_{\leq d}\rho(t)}_1.
\end{align}
where, in (1), we have used the fact that $\Pi_{v, d}\Pi_{\leq d} = \Pi_{\leq d}\Pi_{v, d}$ to set $\norm{\Pi_{v, d}\mathcal{P}_{\leq d}\rho(s)}_1 = \norm{\Pi_{\leq d}\Pi_{v, d}\rho(s)\Pi_{\leq d}}_1 \leq \norm{\Pi_{\leq d}} \norm{\Pi_{v, d}\rho(s)}_1 \norm{\Pi_{\leq d}} = \norm{\Pi_{v, d}\rho(s)}_1$,
\end{proof}
\noindent Finally, combining Lemma \ref{lemma:adiabatic_elim} with Lemmas \ref{lemma:chernoff} and \ref{lemma:particle_number_bound}, we obtain the next lemma quantifying the truncation error as a function of $d$.
\begin{lemma}\label{lemma:truncation_error_bosonic}
    For any $d \geq 1$, it follows that 
    \begin{align*}
    \norm{\rho(t) - \mathcal{T}\exp\bigg(\int_0^t \mathcal{L}_{\leq d}(s) ds\bigg) \mathcal{P}_{\leq d}\rho(0)}_1 \leq O\big(m^{1-k_0/2}d^{2 + k_0/2} t (J_C + J_\textnormal{os} + U_C + U_\textnormal{os} + \kappa) e^{-\frac{1}{2}({d}/{d_0 m})^{1/\alpha}}\big),
    \end{align*}
    where $d_0, k_0, \alpha$ are the constants in Lemma \ref{lemma:chernoff_model}.
\end{lemma}
\begin{proof}
We bound each term in Lemma \ref{lemma:adiabatic_elim}. We first note that
\begin{align}\label{eq:Q_bound}
    \norm{\mathcal{Q}_{\leq d}\rho(t)}_1 &= \norm{\Pi_{> d} \rho(t) + \Pi_{\leq d}\rho(t)\Pi_{>d}}_1 \nonumber \\
    &\leq \norm{\Pi_{>d}\rho(t)}_1 + \norm{\Pi_{\leq d}\rho(t) \Pi_{>d}}_1 \nonumber \\
    &\numleq{1} \sqrt{\text{Tr}(\Pi_{>d}\rho(t))} + \sqrt{\text{Tr}(\Pi_{\leq d}\rho(t)) \text{Tr}(\Pi_{>d}\rho(t))} \nonumber \\
    &\leq 2\sqrt{\text{Tr}(\Pi_{\geq d}\rho(t))} \leq 2\sqrt{e}\bigg(\frac{d}{d_0 m}\bigg)^{k_0/2} \exp\bigg(-\frac{1}{2}\bigg(\frac{d}{d_0 m}\bigg)^{1/\alpha}\bigg),
\end{align}
where, in (1), we have used the Holder's inequality to conclude that $\norm{A\rho(t)B}_1 \leq \sqrt{\text{Tr}({A^\dagger A \rho(t)}) \text{Tr}({B^\dagger B \rho(t)})} $. Furthermore, 
\begin{align}\label{eq:Pi_bound}
    \norm{\Pi_{v, d}\rho(t)} \leq \sqrt{\text{Tr}(\Pi_{v, d}\rho(t))} \leq \sqrt{\text{Tr}(\Pi_{\geq d}\rho(t))} \leq \sqrt{e}\bigg(\frac{d}{d_0 m}\bigg)^{k_0/2} \exp\bigg(-\frac{1}{2}\bigg(\frac{d}{d_0 m}\bigg)^{1/\alpha}\bigg).
\end{align}
Finally, we consider upper-bounding $\smallnorm{\mathcal{P}_{\leq d}\mathcal{L}(s) \mathcal{Q}_{\leq d}}_1 \leq \smallnorm{\mathcal{P}_{\leq d} \mathcal{L}(s) \mathcal{P}_{\leq d}}_\diamond + \smallnorm{\mathcal{P}_{\leq d}\mathcal{L}(s)}_\diamond \leq 2 \norm{\mathcal{P}_{\leq d}\mathcal{L}(s)}_\diamond$, where we have used the fact that $\norm{\mathcal{P}_{\leq d}}_\diamond \leq 1$. Next, we note that, for a Hamiltonian $H$ and jump operator $L$,
\begin{align}
    \norm{\mathcal{P}_{\leq d}[H, \cdot]}_\diamond \leq 2\smallnorm{\Pi_{\leq d}H}  \text{ and }\norm{\mathcal{P}_{\leq d}\mathcal{D}_L}_\diamond \leq \smallnorm{\Pi_{\leq d}L}^2 + \smallnorm{\Pi_{\leq d}L^\dagger L}.
\end{align}
Furthermore, since
\begin{align}
    \smallnorm{\Pi_{\leq d} a_v}, \smallnorm{\Pi_{\leq d}a_v^\dagger} \leq \sqrt{d + 1}, \smallnorm{\Pi_{\leq d}a_v^\dagger a_u} \leq d \text{ and } \smallnorm{\Pi_{\leq d}a_v a_u}, \smallnorm{\Pi_{\leq d}a_v^\dagger a_u^\dagger} \leq d + 2,
\end{align}
we obtain
\begin{align}\label{eq:upper_bound_L_H}
    \smallnorm{\mathcal{P}_{\leq d}[\ \cdot \ , H(t)]}_\diamond &\leq 2\smallnorm{\Pi_{\leq d} H_\text{g}^\text{hop}(t)}_\diamond + 2\smallnorm{\Pi_{\leq d} H_\text{g}^\text{sq}(t)}_\diamond +2\smallnorm{\Pi_{\leq d} H_\text{g}^\text{disp}(t)}_\diamond  + 2\smallnorm{\Pi_{\leq d}H_\text{g}^\text{ng}(t)} \nonumber \\
    &\leq 4d\sum_{v, u}\smallabs{\mathcal{J}_{v, u}} + 4(d + 2) \sum_{v, u}\abs{\mathcal{G}_{v, u}} + 4 \sqrt{d + 1}\sum_v \abs{\mathcal{D}_{v, u}(t)} + 4d^2 \sum_{v, u}\abs{U_{v, u}}\nonumber\\
    &\leq 4(d + 1) m (J_\text{os} + J_C) + 2\sqrt{2(d + 1)} \Omega +  4d^2 (U_C + U_\text{os})\nonumber \\
    &\leq 8m\big(d (J_\text{os} + J_C) + \sqrt{d}\abs{\Omega} + d^2(U_C + U_\text{os}\big).
\end{align}
where we have used the decomposition of $H_\text{g}(t)$ in Eq.~(\ref{eq:H_g_bosonic_model}). Furthermore, 
\begin{align}\label{eq:upper_bound_L_N}
    \norm{\mathcal{P}_{\leq d}\mathcal{L}_\text{n}} &\leq \sum_v \bigg(\kappa_1 \smallnorm{\mathcal{P}_{\leq d}\mathcal{D}_{a_v}}_\diamond  + \kappa_2 \smallnorm{\mathcal{P}_{\leq d}\mathcal{D}_{a_v^\dagger}}_\diamond  + \kappa_3\smallnorm{\mathcal{P}_{\leq d}\mathcal{D}_{a_v^\dagger a_v}}_\diamond \bigg)\nonumber\\
    & \leq  m \big(\kappa_1 (2d + 1) + \kappa_2 (2d + 2) + \kappa_3d^2\big) \nonumber \\
    &\leq 8 m \big((\kappa_1 + \kappa_2)d + \kappa_3 d^2 \big).
\end{align}
Combining Eqs.~(\ref{eq:upper_bound_L_H}) and (\ref{eq:upper_bound_L_N}), we obtain that
\begin{align}\label{eq:PLQ_bound}
    \smallnorm{\mathcal{P}_{\leq d} \mathcal{L}(s) \mathcal{Q}_{\leq d}}_\diamond \leq 2 \smallnorm{\mathcal{P}_{\leq d} \mathcal{L}(s)}_\diamond &\leq 16m \big((J_\text{os} + J_C +  (\kappa_1 + \kappa_2))d + \Omega \sqrt{d} + (U_C + U_\text{os} + \kappa_3)d^2\big) \nonumber \\
    &\leq 16 m d^2 \big(J_\text{os} + J_C + \Omega + U_\text{os} + U_C + \kappa\big).
\end{align}
Finally, combining Eqs.~(\ref{eq:Q_bound}, \ref{eq:Pi_bound},  \ref{eq:PLQ_bound}) together with Lemmas \ref{lemma:adiabatic_elim} and \ref{lemma:chernoff_model}, we obtain the lemma.
\end{proof}

\emph{Trotterization of the truncated model}. We will perform a first-order Trotterization of the state $\rho_{\leq d}(t) = \mathcal{T}\exp(\int_0^t \mathcal{L}_{\leq d}(s) ds) \mathcal{P}_{\leq d}(\rho(0))$. We will split the Hamiltonian $H_{\leq d}(t)$ into a sum of inter-site terms $H_{\leq d}^\text{C}(t)$ and a sum of on-site terms $H_{\leq d}^\text{os}(t)$:
\begin{align}
H_{\leq d}(t) = \underbrace{\sum_{i< j} \sum_{\sigma, \sigma'} h^\text{C}_{i, \sigma; j, \sigma'}(t)}_{H^\text{C}_{\leq d}(t)} + \underbrace{\sum_{i} \sum_{\sigma, \sigma'} h^\text{os}_{i; \sigma, \sigma'}(t)}_{H^\text{os}_{\leq d}(t)}, 
\end{align}
where, in $h_{i, \sigma; j, \sigma'}^\text{C}(t)$ we include all the terms that mediate an interaction between $(i, \sigma)$ and $(j, \sigma')$:
\begin{align}
h_{i, \sigma; j, \sigma'}^\text{C}(t) &= \Pi_{\leq d}\bigg(U_{i, \sigma; j, \sigma'}(t)  n_{i, \sigma} n_{j, \sigma'} + \sum_{\alpha, \alpha'}J_{i, \sigma; j, \sigma'}^{\alpha, \alpha'}(t) c_{i, \sigma}^\alpha c_{j, \sigma'}^{\alpha'}\bigg) \Pi_{\leq d} + (i, \sigma) \leftrightarrow (j, \sigma') \nonumber \\
&= 2 U_{i, \sigma; j, \sigma'}(t) n_{i, \sigma; \leq d} n_{j, \sigma'; \leq d} + 2 \sum_{\alpha, \alpha'} J_{i, \sigma; j', \sigma'}(t) c_{i, \sigma; \leq d}^{\alpha} c_{j, \sigma'; \leq d}^{\alpha'},
\end{align}
where $c_{i, \sigma; \leq d}^{1} = (a_{i, \sigma; \leq d} + a^\dagger_{i, \sigma; \leq d}) / \sqrt{2}$, $c_{i, \sigma; \leq d}^{2} = (a_{i, \sigma; \leq d} - a^\dagger_{i, \sigma; \leq d}) / \sqrt{2}i$. In $h_{i;\sigma, \sigma'}^\text{os}(t)$, we include all the terms (Gaussian or non-Gaussian) that act between modes $(i, \sigma)$ and $(i, \sigma')$:
\begin{align}
    h_{i; \sigma, \sigma'}^\text{os}(t) &=  \Pi_{\leq d}\bigg(U_{i, \sigma; i, \sigma'}(t) n_{i, \sigma} n_{i, \sigma'} + \sum_{\alpha, \alpha'}J^{\alpha, \alpha'}_{i, \sigma; i, \sigma'}(t) c_{i, \sigma}^\alpha c_{i, \sigma'}^{\alpha'}\bigg)\Pi_{\leq d} \nonumber \\
    &= U_{i, \sigma; i, \sigma'}(t) n_{i, \sigma; \leq d}n_{i, \sigma'; \leq d} + \sum_{\alpha, \alpha'}J^{\alpha, \alpha'}_{i, \sigma; i', \sigma'}(t) \Pi_{i, \sigma; \leq d} \Pi_{i, \sigma'; \leq d} c_{i, \sigma}^\alpha c_{i', \sigma'}^{\alpha'}\Pi_{i, \sigma; \leq d} \Pi_{i, \sigma'; \leq d}.
\end{align}
Furthermore, we will also decompose the dissipation $\mathcal{L}_{\text{n}, \leq d}$:
\begin{align}\label{eq:noise_lindbladian_decomp}
\mathcal{L}_{\text{n}, \leq d} = \sum_{i < j}\sum_{\sigma, \sigma'} \mathcal{L}_{i,\sigma; j, \sigma'}^\text{n}(t) \text{ where }\mathcal{L}_{i, \sigma; j, \sigma'}^\text{n}(t) =  \sum_{l = 1}^3  \kappa_l \big(p_{i, \sigma;j, \sigma'}^{(l)}(t)\mathcal{D}_{L_{i, \sigma, \leq d}^{(l)}} + q_{i, \sigma; j, \sigma'}^{(l)}(t)\mathcal{D}_{L_{j, \sigma', \leq d}^{(l)}}\big),
\end{align}
where we will choose $p_{i, \sigma;i', \sigma'}^{(l)}(t), q_{i, \sigma;i', \sigma'}^{(l)}(t) \geq 0$ later. For this decomposition of $\mathcal{L}_{n, \leq d}$ to be consistent, we must also have
\begin{align}\label{eq:pq_norm}
\forall k, \sigma: \ \sum_{k' > k} \sum_{\sigma'} p_{k, \sigma; k', \sigma'}^{(l)}(t) + \sum_{k' < k} \sum_{\sigma'} q_{k', \sigma'; k, \sigma}^{(l)}(t) = 1 .
\end{align}
Now, the state of the truncated model at time $t$, $\rho_{\leq d}(t)$, will be approximated by the state $\sigma_{T, \leq d}$, where $T$ is the number of Trotter steps and
\begin{subequations}\label{eq:Trotter_theorem_1_ferm}
\begin{align}
    \sigma_{T, \leq d} = \prod_{\tau = T}^1 \bigg(\prod_{i < j} \prod_{\sigma, \sigma'} \Phi^{i, \sigma; j, \sigma'}_{\tau\delta, (\tau - 1) \delta }\bigg) \mathcal{U}^{\text{os}}_{\tau \delta, (\tau - 1)\delta} \rho_{\leq d}(0),
\end{align}
where $\delta = t/T$,
\begin{align}
\Phi_{t, t'}^{i, \sigma; j, \sigma'} = \mathcal{T}\exp\bigg(\int_{t'}^{t}\mathcal{L}^\text{C}_{i, \sigma; j, \sigma'}(s) ds\bigg) \text{ where }\mathcal{L}^\text{C}_{i, \sigma; j, \sigma'} (s)= -i[h^\text{C}_{i, \sigma; j, \sigma'}(s), \ \cdot\ ] + \mathcal{L}^\text{n}_{i, \sigma; j, \sigma'}(s),
\end{align}
and
\begin{align}
\mathcal{U}^\text{os}_{\tau \delta, (\tau - 1)\delta} = U^\text{os}_{\tau \delta, (\tau - 1)\delta} (\cdot) U^{\text{os}\dagger}_{\tau \delta, (\tau - 1)\delta} \text{ where }U^\text{os}_{\tau \delta, (\tau - 1)\delta} = \mathcal{T}\exp\bigg(-i\int_{(\tau - 1) \delta }^{\tau \delta } H^\text{os}(s) ds\bigg).
\end{align}
\end{subequations}
The next lemma provides an upper bound on the Trotter error $\norm{\rho_{\leq d}(t) -  \sigma_{T, \leq d}}_1$.
\begin{lemma}\label{lemma:Trotter_bound_bosons}
For any $T > 0$ and $d \geq 1$:
\begin{align*}
\norm{\sigma_{T, \leq d} - \rho_{\leq d}(t)}_1 \leq \frac{16t^2 m^2d^4 \Lambda^2}{T}.
\end{align*}
\end{lemma}
\begin{proof}
    This lemma follows from an application of Lemma \ref{lemma:Trotter_bounded_lind}: We note that
    \begin{align}
    &\smallnorm{\mathcal{L}^{\text{C}}_{i, \sigma; j, \sigma'}(s)}_{\diamond} \leq 2\smallnorm{h^\text{C}_{i, \sigma; j, \sigma'}(s)} +  2 \sum_{l = 1}^3 \kappa_l \big(p_{i, \sigma;j, \sigma'}^{(l)}(s) \smallnorm{L_{i, \sigma, \leq d}^{(l)}}^2 + q_{i, \sigma;j, \sigma'}^{(l)}(s)\smallnorm{L_{j, \sigma', \leq d}^{(l)}}^2 \big) \nonumber\\
    &\qquad \qquad \qquad \numleq{1} \bigg(4\smallabs{U_{i, \sigma; j, \sigma'}(s)} + 8\sum_{\alpha, \alpha'}\smallabs{J^{\alpha, \alpha'}_{i, \sigma; j, \sigma'}(s)} + 2 \sum_{l = 1}^3 \kappa_l (p_{i, \sigma;j, \sigma'}^{(l)}(s) + q_{i, \sigma;j, \sigma'}^{(l)}(s))\bigg)d^2,\\
    &\smallnorm{[\ \cdot\ , H^\text{os}(t)]}_\diamond \leq 2 \sum_{i, \sigma, \sigma'}\smallnorm{h^{\text{os}}_{i; \sigma, \sigma'}(t)} \nonumber \\
    &\qquad \qquad \qquad \ \numleq{2} \bigg(2\sum_{i, \sigma,  \sigma'}\abs{U_{i, \sigma; i, \sigma'}(s)} + 4\sum_{i, \sigma, \sigma'}\sum_{\alpha, \alpha'} \smallabs{J^{\alpha, \alpha'}_{i, \sigma; i, \sigma'}(s)}\bigg)d^2,
\end{align}
where, in (1) and (2), we have used the fact that, for the truncated bosonic model, $\smallnorm{c_{i, \sigma, \leq d}^\alpha} \leq \sqrt{2d} \leq \sqrt{2}d, \smallnorm{n_{i, \sigma, \leq d}} \leq d, \smallnorm{L_{i, \sigma, \leq d}^{(1)}}, \smallnorm{L_{i, \sigma, \leq d}^{(2)}} \leq \sqrt{d} \leq d$ and $\smallnorm{L_{i, \sigma, \leq d}^{(3)}} \leq d$. We can now estimate the parameter $\ell$ from Lemma \ref{lemma:Trotter_bounded_lind}: $\ell$ would be an upper bound on
\begin{align}
    &\smallnorm{[\ \cdot \ , H^\text{os}(s)]}_\diamond + \sum_{j, i < j} \sum_{\sigma,\sigma'}\smallnorm{\mathcal{L}^\text{C}_{i, \sigma; j, \sigma'}(s)}_\diamond \nonumber\\
    & \qquad\leq \bigg(4 \sum_{i, j} \sum_{\sigma, \sigma'} \bigg(\abs{U_{i, \sigma; j, \sigma'}(s)} + \sum_{\alpha, \alpha'} \smallabs{J^{\alpha, \alpha'}_{i, \sigma; j, \sigma'}}\bigg) + 2 \sum_{i, j: i < j} \sum_{\sigma, \sigma'}\sum_{l = 1}^3 \kappa_l (p_{i, \sigma;j, \sigma'}^{(l)}(s) + q_{i, \sigma;j, \sigma'}^{(l)}(s))\bigg)d^2 \nonumber \\
    & \qquad \leq \underbrace{4 (U_C + U_\text{os} + J_C + J_\text{os} + \kappa) m d^4}_{\ell}.
\end{align}
Thus, from Lemma \ref{lemma:Trotter_bounded_lind}, we obtain that $\norm{\sigma_{T, \leq d} - \rho_{\leq d}(t)}_1 \leq 16t^2 m^2d^2 (U_C + U_\text{os} + J_C + J_\text{os} + \kappa)^2 / T$.
\end{proof}

\begin{lemma}[Separability condition for the bosonic model]\label{lemma:separability_bosons} 
Consider the following Lindbladian $\mathcal{L}_{\leq d}(t)$ on two bosonic modes truncated to $d \geq 1$ particles each:
\begin{align*}
\mathcal{L}_{\leq d}(t) = -i[h_{\textnormal{g}, \leq d}(t) + h_{\textnormal{ng}, \leq d}(t), \cdot] + \sum_{i\in \{1, 2\}} \sum_{l = 1}^3\kappa_{i}^{(l)}(t)\mathcal{D}_{L_{i, \leq d}^{(l)}}, 
\end{align*}
where
\begin{align*}
 h_{\textnormal{g}, \leq d}(t) = \sum_{\alpha, \beta \in \{1, 2\}} g_{\alpha, \beta}(t) c_{1, \leq d}^\alpha c_{2, \leq d}^{\beta}, h_{\textnormal{ng}, \leq d}(t) =  u(t) n_{1, \leq d}n_{2, \leq d}, \ L_{i}^{(1)} = a_{i, \leq d}, L_{i}^{(2)} = a_{i, \leq d}^\dagger \text{ and } L_{i}^{(3)} = n_{i, \leq d},
\end{align*}
where $c_{i, \leq d}^1 = (a_{\leq d} + a_{\leq d}^\dagger)/\sqrt{2}$, $c_{i, \leq d}^2 = (a_{\leq d} -a_{\leq d}^\dagger)/\sqrt{2} i$ and $g_{\alpha, \beta}(t), u(t)$ are real and $\kappa_i^{(l)}(t) \geq 0$. If
\begin{enumerate}
    \item[(C1)] $\kappa_i^{(1)}(t), \kappa_i^{(2)}(t) \geq \sum_{\alpha, \beta}\abs{g_{\alpha, \beta}(t)} $ and
    \item[(C2)] $\kappa_i^{(3)}(t) \geq \abs{u(t)}$,
\end{enumerate}
then there is a completely-positive map $\mathcal{M}_{t + \tau, t}$ which maps separable states to separable states and
\begin{align*}
\norm{\mathcal{M}_{t + \tau, t} - \mathcal{T}\exp\bigg(\int_{t}^{t+\tau}\mathcal{L}_{\leq d}(s) ds\bigg)}_\diamond \leq 8d^4\bigg(\sum_{\alpha, \beta}\int_{t}^{t + \tau} \abs{g_{\alpha, \beta}(s)}ds + \int_{t}^{t + \tau}\abs{u(s)}ds + \sum_{l = 1}^3 \sum_{i \in\{1, 2\}} \int_{t}^{t + \tau} \kappa_i^{(l)}(s) ds\bigg)^2.
\end{align*}
\end{lemma}
\begin{proof}
It will be notationally convenient to define the scalars
\begin{align}
    &G_{\alpha, \beta} = \int_{t}^{t + \tau} g_{\alpha, \beta}(s)ds,   K_i^{(l)} = \int_{t}^{t + \tau}\kappa_i^{(l)}(s) ds, U = \int_{t}^{t + \tau} u(s) ds  \text{ and} \nonumber\\
    &G_0 = \sum_{\alpha, \beta}\int_{t}^{t + \tau}\abs{g_{\alpha, \beta}(s)}ds, U_0 = \int_{t}^{t + \tau}\abs{u(s)}ds,  K_i = \sum_{l = 1}^3 K_i^{(l)}, K^{(l)} = \sum_{i \in \{1, 2\}}K_i^{(l)}, K = \sum_{l = 1}^3 K^{(l)}.
\end{align}
We will define the completely positive map $\mathcal{R}_{t + \tau, t}$ via 
\begin{subequations}\label{eq:def_R}
\begin{align}
\mathcal{R}_{t + \tau, t}(\rho) = \mathbb{E}_z\big[R_1(z) R_2(z) \rho R_2^\dagger(z) R_1^\dagger(z)\big],
\end{align}
with
\begin{align}
R_1(z) = Q_1 + e^{-i\pi / 4}\sum_{\alpha, \beta}z_{\alpha, \beta}\sqrt{G_{\alpha, \beta}} c_{1, \leq d}^\alpha, R_2(z) = Q_2 + e^{-i\pi/4} \sum_{\alpha, \beta} z_{\alpha, \beta}^*\sqrt{G_{\alpha, \beta}} c_{2, \leq d}^\beta,
\end{align}
\end{subequations}
where $z_{\alpha, \beta}$ are drawn independently and uniformly at random from the set $\{\pm 1, \pm i\}$ and $Q_i = \exp(-({K^{(1)}_i} a_{i, \leq d}^\dagger a_{i, \leq d} + {K^{(2)}_i} a_{i, \leq d} a_{i, \leq d}^\dagger)/2) $. It can be noted that, by construction, $\mathcal{R}_{t + \tau, t}$ maps a separable input state to a 
separable (but possibly unnormalized) output state. Explicitly evaluating the expectation value in Eq.~(\ref{eq:def_R}), we obtain that
\begin{align}\label{eq:expansion_R}
\mathcal{R}_{t + \tau, t}(\rho) &= Q_1 Q_2 \rho Q_2^\dagger Q_1^\dagger -i \sum_{\alpha, \beta} G_{\alpha, \beta} \big(c_{1, \leq d}^{\alpha} c_{2, \leq d}^\beta \rho Q_2^\dagger Q_1^\dagger -  Q_1 Q_2 \rho  c_{1, \leq d}^\alpha c_{2, \leq d}^\beta \big) + \nonumber\\
&\qquad\sum_{\alpha, \beta}\abs{G_{\alpha, \beta}} \big(c_{1, \leq d}^\alpha Q_2 \rho Q_2^\dagger c_{1, \leq d}^\alpha  +  Q_1 c_{2, \leq d}^\beta \rho c_{2, \leq d}^\beta Q_1^\dagger\big) + \nonumber\\
&\qquad \sum_{\alpha, \alpha', \beta, \beta'} \bigg(G_{\alpha, \beta}G_{\alpha', \beta'} c_{1, \leq d}^\alpha c_{2, \leq d}^\beta \rho c_{2, \leq d}^{\beta'}c_{1, \leq d}^{\alpha'} + \abs{G_{\alpha, \beta}}\abs{G_{\alpha', \beta'} } c_{1, \leq d}^\alpha c_{2, \leq d}^{\beta'} \rho c_{2, \leq d}^{\beta} c_{1, \leq d}^{\alpha'}\bigg) \nonumber \\
&= \rho - i\bigg[\int_t^{t + \delta} h_{\text{g}, \leq d}(s)ds, \rho\bigg] + \sum_{\alpha, \beta}\abs{G_{\alpha, \beta}} \big(c_{1, \leq d}^\alpha  \rho  c_{1, \leq d}^\alpha  +   c_{2, \leq d}^\beta \rho c_{2, \leq d}^\beta \big) -\nonumber\\
&\qquad \frac{1}{2} \sum_{i \in \{1, 2\}}\{K_i^{(1)} a_{i, \leq d}^\dagger a_{i, \leq d} + K_i^{(2)} a_{i, \leq d} a_{i, \leq d}^\dagger, \rho \} + \Delta_{t + \tau, t}(\rho),
\end{align}
where, using the fact that 
\begin{align}
&\norm{Q_i} \leq 1, \nonumber \\
    &\smallnorm{Q_i - I} \leq \frac{1}{2} (K_i^{(1)}\norm{a_{1, \leq d}}^2 + K_i^{(2)} \norm{a_{2, \leq d}}^2) \leq \frac{d}{2}(K_i^{(1)} + K_i^{(2)}),\nonumber\\
    &\norm{Q_i - \bigg(I - \frac{1}{2}({K_i^{(1)}}a_{i, \leq d}^\dagger a_{i, \leq d} + {K_i^{(2)}}a_{i, \leq d} a_{i, \leq d}^\dagger)\bigg)} \leq \frac{1}{8} (K_i^{(1)} \norm{a_{i, \leq d}}^2 + K_i^{(2)} \norm{a_{i, \leq d}}^2)^2 \leq \frac{d^2}{8}(K_i^{(1)} + K_i^{(2)})^2, \nonumber\\
    &\norm{c_{i, \leq d}^\alpha} \leq \sqrt{2}\norm{a_{i, \leq d}} \leq \sqrt{2d},
\end{align}
it follows that
\begin{align}
\norm{\Delta_{t+\tau, t}}_\diamond \leq \frac{d^2}{2}(K^{(1)} + K^{(2)})^2 + 4d^2 G (K^{(1)} + K^{(2)}) + 8d^2 G^2.
\end{align}
Similarly, we also define the completely positive map $\tilde{\mathcal{R}}_{t + \tau, t}$:
\begin{subequations}\label{eq:def_R_tilde}
\begin{align}
\tilde{\mathcal{R}}_{t + \tau, t}(\rho) = \mathbb{E}_y\big[\tilde{R}_1(y) \tilde{R}_2(y) \rho \tilde{R}_2^\dagger(y) \tilde{R}_1^\dagger (y)\big],
\end{align}
with
\begin{align}
\tilde{R}_1(y) = \tilde{Q}_1 + y e^{-i\pi /4}\sqrt{U} n_{1, \leq d} \text{ and } \tilde{R}_2(y) = \tilde{Q}_2 + y^* e^{-i\pi / 4}\sqrt{U} n_{2, \leq d},
\end{align}   
\end{subequations}
where $y$ is drawn randomly from $\{-1, 1, i, -i\}$ and $\tilde{Q}_i = \exp(-\frac{K_i^{(3)}}{2}n_{i, \leq d}^2)$. Similar to $\mathcal{R}_{t + \tau, t}$, $\tilde{\mathcal{R}}_{t + \tau, t}$ also maps a separable input state to a separable but possibly unnormalized output state. By explicitly evaluating the expectation in Eq.~(\ref{eq:def_R_tilde}), we find that
\begin{align}\label{eq:expansion_R_tilde}
\tilde{\mathcal{R}}_{t + \tau, t}(\rho) &= \tilde{Q}_1 \tilde{Q}_2 \rho \tilde{Q}_2^\dagger \tilde{Q}_1^\dagger -i U(n_{1, \leq d} n_{2, \leq d} \rho \tilde{Q}_1^\dagger \tilde{Q}_2^\dagger - \tilde{Q}_1 \tilde{Q}_2 \rho n_{1, \leq d} n_{2, \leq d}) + \nonumber\\
&\qquad \abs{U}\big(n_{1, \leq d} \tilde{Q}_2 \rho \tilde{Q}_2^\dagger n_{1, \leq d} + \tilde{Q}_1 n_{2, \leq d} \rho n_{2, \leq d} \tilde{Q}_1^\dagger\big) + \abs{U}^2 n_{1, \leq d} n_{2, \leq d} \rho n_{2, \leq d} n_{1, \leq d} \nonumber\\
&=\rho - i\bigg[\int_{t}^{t + \delta} h_{\text{ng}, \leq d}(s) ds, \rho\bigg] + \abs{U}\big(n_{1, \leq d} \rho n_{1, \leq d} + n_{2, \leq d} \rho n_{2, \leq d}\big) - \frac{1}{2}K_3 \{n_{1, \leq d} + n_{2, \leq d}, \rho\} + \tilde{\Delta}_{t+\tau, t}(\rho),
\end{align}
where, using the fact that $\smallnorm{\tilde{Q}_i} \leq 1, \smallnorm{\tilde{Q}_i - I} \leq K_i^{(3)} \norm{n_{i, \leq d}}^2/ 2\leq K_i^{(3)}d^2/2 $, $\smallnorm{\tilde{Q}_i - (I - \frac{K_i^{(3)}}{2}n_i^2)} \leq (K_i^{(3)}\norm{n_{i, \leq d}}^2)^2/8 \leq (K_i^{(3)})^2 d^4 / 8$ and $\norm{n_{i, \leq d}} \leq d$, it follows that
\begin{align}
\smallnorm{\tilde{\Delta}_{t + \tau, t}}_\diamond \leq \frac{1}{2} (K^{(3)})^2 d^4 + {U}_0^2 d^4 + 2 {U}_0 K^{(3)} d^4.
\end{align}
Finally, we consider the channel generated by the fermionic Lindbladian in the time interval $(t, t + \tau)$: Performing a first-order Taylor expansion, we obtain that
\begin{align}\label{eq:expansion_channel}
    \mathcal{T}\exp\bigg(\int_{t}^{t + \tau} \mathcal{L}(s) ds\bigg) \rho = \rho + \int_{t}^{t + \tau}\mathcal{L}(s) \rho ds + \Delta^\mathcal{E}_{t + \tau, t}(\rho),
\end{align}
where $\smallnorm{\Delta^{\mathcal{E}}_{t + \tau, t}}_\diamond \leq 8d^4(G + U + K)^2$. From Eqs.~(\ref{eq:expansion_R}, \ref{eq:expansion_R_tilde}, \ref{eq:expansion_channel}), we then obtain that
\begin{align}
    \mathcal{T}\exp\bigg(\int_{t}^{t + \tau} \mathcal{L}(s) ds\bigg) = \mathcal{M}_{t + \tau, t} + E_{t + \tau, t},
\end{align}
where
\begin{align}\label{eq:ferm_thm_1_sep_channel_decomp}
    \mathcal{M}_{t + \tau, t} &= \mathcal{R}_{t + \tau, t} + \tilde{\mathcal{R}}_{t + \tau, t} +\nonumber\\
    &\qquad  \underbrace{\sum_{i \in \{1, 2\}} \big(K_i^{(1)} a_{i, \leq d} \cdot a_{i, \leq d}^\dagger + K_i^{(2)} a_{i, \leq d}^\dagger \cdot a_{i, \leq d}\big) - \sum_{\alpha, \beta}\abs{G_{\alpha, \beta}}\big(c_{1, \leq d}^\alpha \cdot  c_{1, \leq d}^\alpha + c_{2, \leq d}^\beta \cdot c_{2, \leq d}^\beta\big)}_{\mathcal{V}_{t + \tau, t}} + \nonumber\\
    &\qquad   \underbrace{\sum_{i \in \{1, 2\}}(K_i^{(3)} - \abs{U}) n_{i, \leq d} \rho n_{i, \leq d}}_{\tilde{\mathcal{V}}_{t + \tau, t}}, \\
    E_{t + \tau, t} &=  {\Delta^\mathcal{E}_{t + \tau, t} - \Delta_{t + \tau, t} - \tilde{\Delta}_{t + \tau, t}}.
\end{align}
We note that $\mathcal{M}_{t+\tau, t}$ is a channel that preserves separability as long as $\mathcal{V}_{t + \tau, t}$ and $\tilde{\mathcal{V}}_{t + \tau, t}$ are completely positive. The complete positivity of $\tilde{\mathcal{V}}_{t + \tau, t}$ is ensured by requiring $K_i^{(3)} \geq \abs{U}$ which is implied by the condition C2 quoted in the lemma statement. To ensure that $\mathcal{V}_{t + \tau, t}$ is completely positive, we note that it can be re-written as
\begin{align}
\mathcal{V}_{t + \tau, t} = \sum_{i \in \{1, 2\}}\bigg( F^{(i)}_{0, 0} a_{i, \leq d} \cdot a_{i, \leq d}^\dagger + F^{(i)}_{0, 1}a_{i, \leq d} \cdot a_{i, \leq d} + F^{(i)}_{1, 0}a_{i, \leq d}^\dagger \cdot a_{i, \leq d} + F^{(i)}_{1, 1} a_{i, \leq d}^\dagger \cdot a_{i, \leq d}^\dagger\bigg),
\end{align}
where
\begin{align}
F^{(1)} = \begin{bmatrix}K_1^{(1)} - \frac{1}{2}G & \frac{1}{2}\sum_{\beta}\big(\abs{G_{2, \beta}} - \abs{G_{1, \beta}}\big) \\
\frac{1}{2}\sum_{\beta} \big(\abs{G_{2, \beta}} - \abs{G_{1, \beta}}\big) & K_1^{(2)} - \frac{1}{2}G
\end{bmatrix}, F^{(2)} =\begin{bmatrix}K_2^{(1)} - \frac{1}{2}G & \frac{1}{2}\sum_{\alpha}\big(\abs{G_{\alpha, 2}} - \abs{G_{\alpha, 1}}\big) \\
\frac{1}{2}\sum_{\beta} \big(\abs{G_{\alpha, 2}} - \abs{G_{\alpha, 1}}\big) & K_2^{(2)} - \frac{1}{2}G
\end{bmatrix}.
\end{align}
As long as $F^{(1)}, F^{(2)}$ are positive-semidefinite, it would follow that $\mathcal{V}_{t + \tau, t}$ is completely positive. Now, it is easy to see that a sufficient condition for $F^{(i)} \succeq 0$ is that $K_i^{(1)}, K_i^{(2)} \geq G$, which is implied by the condition C1 quoted in the lemma statement. Finally, the error term $E_{t + \tau, t}$ can be bounded by
\begin{align}
\smallnorm{E_{t + \tau, t}}_\diamond &\leq \smallnorm{\Delta^\mathcal{E}_{t + \tau, t}}_\diamond + \smallnorm{\Delta_{t + \tau, t}}_\diamond + \smallnorm{\tilde{\Delta}_{t + \tau, t}}_\diamond \leq 8d^4(G + U_0 + K)^2,
\end{align}
which establishes the error bound in the lemma statement.
\end{proof}
\begin{theorem}[High-noise seperability and classical simulation of bosonic model; reproduced from the main text]
Suppose $\rho(t)$ is the state obtained after evolving the bosonic system for time $t$ with an initial product state, then for $\min(\kappa_1, \kappa_2) \geq 2J$ the state $\rho(t)$ is separable. Furthermore, there is a randomized classical algorithm that can sample within $\epsilon$ total variation error of $\rho(t)$ in $O({\Lambda^2 t^2 m^{4L + 8}}{\epsilon^{-1}} \textnormal{polylog}\big({m\Lambda t}/{\epsilon}))$ time.
\end{theorem}
\begin{proof}
To prove Theorem 2, we will start with truncated first-order Trotter approximation of $\rho(t)$, i.e.\ with  $\sigma_{T, \leq d}$ given in Eq.~(\ref{eq:Trotter_theorem_1_ferm}). From Lemmas \ref{lemma:truncation_error_bosonic} and \ref{lemma:Trotter_bound_bosons}, we obtain that
\begin{align}\label{eq:truncated_Trotter_bound}
\norm{\rho(t) - \sigma_{T, \leq d}}_1 \leq O\bigg(\frac{\Lambda^2 t^2 m^2d^2 }{T}\bigg) + O\big( \Lambda t m^{1-k_0/2}d^{2 + k_0/2}   e^{-\frac{1}{2}({d}/{d_0 m})^{1/\alpha}}\big),
\end{align}
where $\Lambda = U_C + U_\textnormal{os} + J_C + J_\textnormal{os} + \kappa$. Next, we use Lemma \ref{lemma:separability_bosons} to further approximate $\sigma_{T, \leq d}$ with a separable state $\phi_T$. However, the noise rates at each step in the Trotterization need to be sufficiently high to meet the necessary conditions for separability [(C1) and (C2) provided in Lemma \ref{lemma:separability_bosons}]---to ensure this, we make a choice of the parameters $p^{(l)}_{i, \sigma; j, \sigma'}(t), q^{(l)}_{i, \sigma; j, \sigma'}(t)$ in Eq.~(\ref{eq:noise_lindbladian_decomp}) that we so far left unspecified: 
\begin{subequations}
\begin{align}
    &p^{(1)}_{i, \sigma; j, \sigma'}(t) = p^{(2)}_{i, \sigma; j, \sigma'}(t) = \frac{\sum_{\alpha, \alpha'}\smallabs{J^{\alpha, \alpha'}_{i, \sigma; j, \sigma'}(t)}}{\sum_{k \neq i} \sum_{\nu}\sum_{\alpha, \alpha'}\smallabs{J^{\alpha, \alpha'}_{i, \sigma; k, \nu}(t)}}, p^{(3)}_{i, \sigma; j, \sigma'}(t) = \frac{\smallabs{U_{i, \sigma; j, \sigma'}(t)}}{\sum_{k \neq i} \sum_\nu \smallabs{U_{i, \sigma; k, \nu}(t)}}, \nonumber \\
    &q^{(1)}_{i, \sigma; j, \sigma'}(t) = q^{(2)}_{i, \sigma; j, \sigma'}(t) = \frac{\sum_{\alpha, \alpha'}\smallabs{J^{\alpha, \alpha'}_{i, \sigma; j, \sigma'}(t)}}{\sum_{k \neq i} \sum_{\nu}\sum_{\alpha, \alpha'}\smallabs{J^{\alpha, \alpha'}_{j, \sigma'; k, \nu}(t)}}, q^{(3)}_{i, \sigma; j, \sigma'}(t) = \frac{\smallabs{U_{i, \sigma; j, \sigma'}(t)}}{\sum_{k \neq i} \sum_\nu \smallabs{U_{j, \sigma'; k, \nu}(t)}},
\end{align}
and it can be checked that they satisfy the normalization condition in Eq.~(\ref{eq:pq_norm}). Considering now the channels $\Phi^{i, \sigma; j, \sigma'}_{\tau \delta, (\tau - 1)\delta}$ from Eq.~(\ref{eq:Trotter_theorem_1_ferm}) in Trotterized state $\sigma_{T, \leq d}$---to apply Lemma \ref{lemma:separability_bosons} to these channels, we need
\begin{enumerate}
    \item[(1)] Imposing condition C1: For $l \in \{1, 2\}$
    \begin{align}
    &\kappa_l p^{(l)}_{i, \sigma; j, \sigma'}(t), \kappa_l q^{(l)}_{i, \sigma; j, \sigma'}(t) \geq 2\sum_{\alpha, \alpha'} \smallabs{J^{\alpha, \alpha'}_{i, \sigma; j, \sigma'}(t)} \text{ or equivalently} \\
    &\kappa_l \geq 2\sum_{k' \neq k} \sum_{\nu'} \sum_{\alpha, \alpha'}\smallabs{J^{\alpha, \alpha'}_{k, \nu; k', \nu'}(t)} \text{ for }(k, \nu) \in \{(i, \sigma), (j, \sigma')\}.
\end{align}
This condition can clearly be satisfied if $\kappa_1, \kappa_2 \geq 2J_C$ since $\sum_{k' \neq k}\sum_{\nu'} \sum_{\alpha, \alpha'}\smallabs{J^{\alpha, \alpha'}_{k, \nu; k', \nu'}(t)} \leq J_C$.
\item[(2)] Imposing condition C2: 
\begin{align}
    &\kappa_3 p^{(3)}_{i, \sigma; j, \sigma'}(t), \kappa_3 q^{(3)}_{i, \sigma; j, \sigma'}(t) \geq  2\smallabs{U_{i, \sigma; j, \sigma'}(t)} \text{ or equivalently} \\
    &\kappa_3 \geq 2\sum_{k' \neq k} \sum_{\nu'} \smallabs{U_{k, \nu; k', \nu'}(t)} \text{ for }(k, \nu) \in \{(i, \sigma), (j, \sigma')\}.
\end{align}
This condition can clearly be satisfied if $\kappa_3 \geq 2U_C$ since $\sum_{k' \neq k}\sum_{\nu'} \smallabs{U_{k, \nu; k', \nu'}(t)} \leq U_C$.
\end{enumerate}
Now, assuming $\kappa_1, \kappa_2 \geq 2J_C$ and $\kappa_3 \geq 2U_C$, we can then approximate $\sigma_{T, \leq d}$ by ${\phi}_T$ given by
\begin{align}
    {\phi}_T = \prod_{\tau = N}^1 \bigg(\prod_{i < j} \prod_{\sigma, \sigma'} \mathcal{M}^{i, \sigma; j, \sigma'}_{\tau \delta, (\tau - 1) \delta} \bigg) \mathcal{U}^\text{os}_{\tau \delta, (\tau - 1) \delta}\rho(0),
\end{align}
where $\mathcal{M}^{i, \sigma; j, \sigma'}_{\tau \delta, (\tau - 1) \delta}$ is the separability preserving completely-positive map corresponding to $\Phi^{i, \sigma; j,\sigma'}_{\tau \delta, (\tau - 1) \delta}$ from Lemma \ref{lemma:separability_bosons}, which also satisfies
\begin{align}
    \smallnorm{\Phi^{i, \sigma; j,\sigma'}_{\tau \delta, (\tau - 1) \delta} -\mathcal{M}^{i, \sigma; j, \sigma'}_{\tau \delta, (\tau - 1) \delta} }_\diamond \leq \varepsilon^{i, \sigma; j, \sigma'}_\tau, 
\end{align}
where
\begin{align}
    \varepsilon^{i, \sigma; j, \sigma'}_\tau = 8d^4\bigg(\int_{(\tau - 1) \delta}^{\tau \delta}\bigg(\sum_{\alpha, \alpha'}\smallabs{J^{\alpha, \alpha'}_{i, \sigma; j, \sigma'}(s)}ds +  \smallabs{U_{i, \sigma; j, \sigma'}(s)} ds + \sum_{l = 1}^3  \kappa_l \big(p_{i, \sigma; j, \sigma'}^{(l)}(s) + q_{i, \sigma; j, \sigma'}^{(l)}(s)\big)\bigg) ds \bigg)^2.
\end{align}
\end{subequations}
We note that $\varepsilon$ is defined by
\begin{align}
    \varepsilon &= \sum_{\tau = 1}^T\sum_{i, j: i < j} \sum_{\sigma, \sigma'}\varepsilon^{i, \sigma; j, \sigma'}_\tau \nonumber\\
    &\leq 8d^4\sum_{\tau = 1}^T \bigg(\int_{(\tau - 1)\delta}^{\tau \delta}\sum_{i, j; i<j}\bigg(\sum_{\alpha, \alpha'}\smallabs{J^{\alpha, \alpha'}_{i, \sigma; j, \sigma'}(s)}ds +  \smallabs{U_{i, \sigma; j, \sigma'}(s)} ds + \sum_{l = 1}^3  \kappa_l \big(p_{i, \sigma; j, \sigma'}^{(l)}(s) + q_{i, \sigma; j, \sigma'}^{(l)}(s)\big)\bigg) ds\bigg)^2 \nonumber \\
    &\leq  8d^4 \big(J_C + U_C + 2\kappa\big)^2 \frac{t^2 m^2}{T} \leq 32 d^4 \frac{\Lambda^2  t^2 m^2}{T}.
\end{align}
Furthermore, we also note from the triangle inequality that
\begin{align}
\smallnorm{\mathcal{M}^{i, \sigma; j, \sigma'}_{\tau \delta, (\tau - 1)\delta}}_\diamond \leq \smallnorm{\Phi_{\tau \delta, (\tau - 1) \delta}^{i, \sigma; j, \sigma'}}_\diamond + \smallnorm{ \mathcal{M}^{i, \sigma; j, \sigma'}_{\tau \delta, (\tau - 1) \delta} - \Phi^{i, \sigma; j,\sigma'}_{\tau \delta, (\tau - 1) \delta}}_\diamond \leq  1 + \varepsilon^{i, \sigma; j, \sigma'}_\tau \leq \exp(\varepsilon^{i, \sigma; j, \sigma'}_\tau ).
\end{align}
We can now bound the error between $\phi_T$ and $\sigma_{T, \leq d}$ using telescoping to obtain
\begin{align}\label{eq:separable_appx_fermion}
    \norm{\phi_T - \sigma_{T, \leq d}}_1 \leq e^{\varepsilon} \varepsilon \leq e^{O(\Lambda^2 t^2 m^2 d^4  / T^2)}O\bigg(\frac{\Lambda^2 t^2 m^2 d^4 }{T}\bigg).
\end{align}
Finally, using Eq.~(\ref{eq:truncated_Trotter_bound}), we obtain that
\begin{align}
\norm{\rho(t) - \phi_T}_1 \leq e^{O( \Lambda^2 t^2 m^2 d^4 / T)}O\bigg(\frac{ \Lambda^2 t^2 m^2 d^4}{T}\bigg) + O\bigg(\frac{\Lambda^2 t^2 m^2d^4 }{T}\bigg) + O\big( \Lambda t m^{1-k_0/2}d^{2 + k_0/2}   e^{-\frac{1}{2}({d}/{d_0 m})^{1/\alpha}}\big).
\end{align}
Thus, choosing 
\begin{align}
d = \Theta\bigg(m\ \textnormal{polylog}\bigg(\frac{m\Lambda t}{\epsilon}\bigg)\bigg), T = \Theta\bigg(\frac{\Lambda^2 t^2 m^6}{\epsilon} \textnormal{polylog}\bigg(\frac{m\Lambda t}{\epsilon}\bigg)\bigg)
\end{align}
ensures that $\norm{\rho(t) - \phi_T}_1 \leq \epsilon$. Finally, we note that $\phi_N$ by itself is guaranteed to be positive semi-definite but not normalized. We will instead consider $\tilde{\phi}_T = \phi_T / \text{Tr}(\phi_T)$---note that, if $\norm{\phi_T - \rho(t)}_1 \leq \epsilon < 1$, 
\begin{align}
\smallnorm{\tilde{\phi}_T - \rho(t)}_1 \leq \frac{1}{\text{Tr}(\phi_N)}\smallnorm{\phi_T - \rho(t)}_1 + \abs{\frac{\text{Tr}(\phi_T) - 1}{\text{Tr}(\phi_T)}}\norm{\rho(t)}_1 \numleq{1} \frac{2\epsilon}{1 - \epsilon} \leq O(\epsilon),
\end{align}
where in (1) we have used that $\abs{\text{Tr}(\phi_T) - 1} = \abs{\text{Tr}(\phi_T) - \text{Tr}(\rho(t))} \leq \norm{\phi_T - \rho(t)}_1\leq \epsilon$.

\emph{Time-complexity of sampling in the Fock state basis}. We now consider the cost of sampling from the state $\tilde{\phi}_T$. By construction, $\tilde{\phi}_T$ is a separable state and hence can be expressed as 
\begin{align}
\tilde{\phi}_T = \sum_{\alpha} p_\alpha \bigg(\bigotimes_i \rho_i^{(\alpha)}\bigg),
\end{align}
where $p_\alpha$ is a probability distribution over $\alpha$ and $\rho_i^{(\alpha)}$ is a state supported on the modes at the $i^\text{th}$ site. To either sample from or compute a local observable in $\tilde{\phi}_T$, we first sample from $p_\alpha$ and obtain a product state $\otimes_{i}{\rho_i^{(\alpha)}}$ from the mixed state ensemble $\tilde{\phi}_T$. Given the initial state $\rho_{\leq d}(0)$ as a product state, we sequentially apply $\mathcal{M}^{i, \sigma; j, \sigma'}_{\tau \delta, (\tau - 1) \delta}$, normalize the result and sample from the resulting separable state to obtain another product state---since the input state is a product state, each application of $\mathcal{M}^{i, \sigma; j, \sigma'}_{\tau \delta, (\tau - 1) \delta}$, normalization and the subsequent sampling involves only the $2L$ truncated bosonic modes at sites $i$ and $j$ and can be classically done in $O(d^{4L}) \leq O(m^{4L}\text{polylog}(m\Lambda t/\varepsilon))$ time. Additionally, the application of the on-site unitaries ($\mathcal{U}^\text{os}_{\tau \delta, (\tau - 1)\delta}$) will map a product state between the different sites to another product state, and it can be applied classically in $O(m d^{3L}) \leq O(m^{3L + 1}\text{polylog}(m\Lambda t/\varepsilon))$ time. Counting the time needed to apply, in this manner, all $\mathcal{M}^{i, \sigma; j, \sigma'}_{\tau \delta, (\tau - 1) \delta}$ and $\mathcal{U}^\text{os}_{\tau \delta, (\tau - 1)\delta}$, the total classical run-time for drawing one product state from $\tilde{\phi}_T$ is thus $O(T m^2 \times m^{4L}  \text{polylog}(m\Lambda t/\varepsilon)) \leq O(\Lambda^2 t^2 m^{4L + 8}\epsilon^{-1} \text{polylog}(m\Lambda t/\varepsilon))$. Having drawn a product state $\otimes_i \rho_i^{(\alpha)}$ from the separable state $\tilde{\phi}_T$, we can now consider the task of drawing a sample in the Fock state basis: Given each $\rho_i^{(\alpha)}$ as a $d^L \times d^L$ matrix, drawing a sample from $\otimes_i \rho_i^{(\alpha)}$ on the Fock state basis requires computational time $O(n d^{L}) \leq O(m^{L + 1} \textnormal{polylog}(m\Lambda t/\epsilon))$. Thus, the total time complexity of drawing a single sample from $\tilde{\phi}_T$ is dominated by the cost of sampling from $p_\alpha$ and is $O({\Lambda^2 t^2 m^{4L + 8}}{\epsilon^{-1}} \textnormal{polylog}\big({m\Lambda t}/{\epsilon}))$.
\end{proof}

\section{High-noise separability in spin models} \label{supplemental:spins}
In this section, we analyze a spin model evolving under a 2-local Hamiltonian in the presence of noise, which closely follows the analysis of the bosonic model in the previous section. We only provide a derivation of the counterpart of Lemma \ref{lemma:separability_bosons} for the spin model, which outlines a sufficient condition for separability preservation for two qudits. Combining this lemma with standard first-order Trotterization can allow us to show that even in the many-body regime, a sufficiently high noise maps a separable state to another separable state.
\begin{lemma}[Separability condition for spin models]\label{lemma:separability_spins}
Consider a Lindbladian on two $d-$level qudits  given by
\begin{align*}
\mathcal{L}(t) = -i[h(t), \cdot] + \kappa(t)\sum_{i \in \{1, 2\}}\sum_{k} \mathcal{D}_{L_{i, k}},
\end{align*}
where $\kappa(t) \geq 0$. Here $h(t)$ is a two-qudit Hamiltonian which we express as
\begin{align*}
h(t) = \sum_{\alpha}s_{\alpha}(t) O_{1, \alpha}\otimes O_{2, \alpha},
\end{align*}
where we can assume $s_\alpha(t) \geq 0$, $O_{i, \alpha}$ are Hermitian operators on the $i^\text{th}$ qudit with $\norm{O_{i, \alpha}}_F \leq 1$. Furthermore, for each $i \in \{1, 2\}$, the jump operators $L_{i, k}$ satisfy $\sum_{k}\norm{L_{i, k}}^2 \leq 1$ and have a full Kraus rank and $\exists \lambda_0 > 0$ such that for any single qudit operator $A$
\begin{align*}
\sum_{k}\smallabs{\textnormal{Tr}(L_{i, k}^\dagger A)}^2 \geq \lambda_0 \norm{A}_F^2.
\end{align*}
Then if $\kappa(t) \geq \sum_{\alpha}s_\alpha(t) / \lambda_0$, there is a completely positive map $\mathcal{M}_{t + \tau, t}$ which maps separable states to separable states and
\begin{align*}
\norm{\mathcal{M}_{t + \tau, t} - \mathcal{T}\exp\bigg(\int_t^{t + \tau}\mathcal{L}(s) ds\bigg)}_\diamond \leq 4 \bigg(\int_{t}^{t + \tau} \kappa(t') dt' + \sum_\alpha \int_{t}^{t + \tau}s_\alpha(t') dt'\bigg)^2.
\end{align*}
\end{lemma}
\begin{proof}
    It will be convenient to introduce the scalars
    \begin{align}
    S_\alpha = \int_{t}^{t + \tau} s_\alpha(t') dt', S = \sum_\alpha S_\alpha \text{ and }K = \int_t^{t + \tau}\kappa(t') dt'.
    \end{align}
    We will also define
    \begin{align}
    q_i^\text{eff}(t) = \frac{\kappa(t)}{2} \sum_{k} L_{i, k}^\dagger L_{i, k}, q^\text{eff} = q_1^\text{eff}(t) \otimes I + I \otimes q_2^\text{eff}(t) \text{ and }Q_i = \exp\bigg(-\int_{t}^{t + \tau}{q_i^\text{eff}}(t') dt'\bigg).
    \end{align}
    Consider the following completely positive map
    \begin{subequations}\label{eq:sep_spins}
    \begin{align}
    \mathcal{R}_{t + \tau, t}\rho = \mathbb{E}_z\big((R_1(z)\otimes R_2(z)) \rho (R_1^\dagger(z) \otimes R_2^\dagger(z))\big), 
    \end{align}
    where
    \begin{align}
    R_1(z) = Q_1 + e^{-i\pi/4}\sum_{\alpha} z_\alpha \sqrt{S_\alpha} O_{1, \alpha} \text{ and }R_2(z) = Q_2 + e^{-i\pi/4}\sum_{\alpha}z_{\alpha}^* \sqrt{S_\alpha}O_{2, \alpha},
    \end{align}
    \end{subequations}
    where $z_\alpha$ are drawn uniformly and independently from the set $\{\pm 1, \pm i\}$. We note that $\mathcal{R}_{t + \tau, t}$ is separability preserving i.e.~maps a separable state to another separable state. Explicitly evaluating the expectation value in Eq.~(\ref{eq:sep_spins}), we obtain
    \begin{align}
        \mathcal{R}_{t + \tau, t}(\rho) &= (Q_1 \otimes Q_2) \rho (Q_1^\dagger \otimes Q_2^\dagger) -i \sum_{\alpha} S_\alpha \big((O_{1,\alpha}\otimes O_{2, \alpha}) \rho (Q_1^\dagger \otimes Q_2^\dagger) - (Q_1 \otimes Q_2) \rho (O_{1, \alpha}\otimes O_{2, \alpha})\big) + \nonumber \\
        &\qquad \qquad \sum_{\alpha}S_\alpha \big((O_{1, \alpha} \otimes Q_2) \rho (O_{1, \alpha}\otimes Q_2^\dagger) + (Q_1\otimes O_{2, \alpha}) \rho (Q_1^\dagger \otimes O_{2, \alpha})\big) + \nonumber \\
        &\qquad \qquad \sum_{\alpha, \alpha'}S_\alpha S_{\alpha'}\big((O_{1, \alpha} \otimes O_{2, \alpha}) \rho (O_{1,\alpha'}\otimes O_{2, \alpha'}) + (O_{1, \alpha}\otimes O_{2, \alpha'}) \rho (O_{1, \alpha}\otimes O_{2, \alpha'})\big),
    \end{align}
    where using $\smallnorm{Q_i} \leq 1, \smallnorm{Q_i - I} \leq K/2, \smallnorm{Q_i - (I - q_i^{\text{eff}})} \leq K^2/8$, we obtain that
    \begin{align}
    \mathcal{R}_{t + \tau,t}(\rho)&=\rho - i\int_t^{t + \tau}[h(t'), \rho]dt' - \int_{t}^{t + \tau}\{q^\text{eff}(t'), \rho\} dt' + \nonumber\\
    &\qquad \qquad \sum_{\alpha}S_\alpha (O_{1, \alpha}\otimes I) \rho (O_{1, \alpha}\otimes I) + (I\otimes O_{2, \alpha}) \rho (I\otimes O_{2, \alpha})) +  \Delta_{t + \tau, t}(\rho) \nonumber\\
    &=\rho + \int_{t}^{t + \tau}\mathcal{L}(t') \rho dt' - \mathcal{G}_{t + \tau, t}\rho  + \Delta^{(R)}_{t + \tau, t}(\rho), 
    \end{align}
    where $\mathcal{G}_{t + \tau, t}$ is a superoperator given by
    \begin{align}\label{eq:superoperator_G}
        &\mathcal{G}_{t+ \tau, t} = \mathcal{G}^{(1)}_{t + \tau, t} \otimes \textnormal{id} + \textnormal{id}\otimes \mathcal{G}^{(2)}_{t + \tau, t}, \text{ where } \nonumber
        \\ &\mathcal{G}^{(i)}_{t + \tau, t}(\rho) = K \sum_{k}L_{i, k} \rho  L_{i, k}^\dagger  - \sum_{\alpha}S_\alpha O_{i, \alpha} \rho O_{i, \alpha} ,
    \end{align}
    and $\Delta_{t+\tau, t}$ is a super-operator with
    \begin{align}
        \smallnorm{\Delta^{(R)}_{t + \tau, t}}_\diamond \leq 2K^2 + 4SK + 2S^2 = 2(S + K)^2.
    \end{align}
    Furthermore, the channel generated by the Lindbladian can be expanded to the first order to obtain
    \begin{align}
        \mathcal{T}\exp\bigg(\int_t^{t + \tau}\mathcal{L}(s) ds\bigg) = \rho + \int_{t}^{t + \tau}\mathcal{L}(t') \rho dt' + \Delta^{(L)}_{t + \tau, t},
    \end{align}
    where, since $\smallnorm{\mathcal{L}(t)}_\diamond \leq 2(\sum_{\alpha} s_\alpha(t) + \kappa(t))$, $\Delta^{(L)}_{t + \tau, t}$ is a superoperator with
    \begin{align}
     \smallnorm{\Delta^{(L)}_{t + \tau, t}}_\diamond \leq 2\big(S + K\big)^2.
    \end{align}
    Consequently, we have that
    \begin{align}
    \norm{\mathcal{T}\exp\bigg(\int_0^t \mathcal{L}(s) ds \bigg) - \big(\mathcal{R}_{t + \tau, t} + \mathcal{G}_{t + \tau, t}\big)}_\diamond \leq 4(S + K)^2.
    \end{align}
    We note that $\mathcal{R}_{t + \tau, t}$ is separability preserving by construction. Furthermore, $\mathcal{G}_{t + \tau, t}$ is a sum of super-operators acting individually on the two qudits---consequently, $\mathcal{G}_{t + \tau, t}$ will be separability preserving as long as it is completely positive. To find a sufficient condition for complete positivity of $\mathcal{G}^{(i)}_{t + \tau, t}$, we will impose that its Choi state, $\Phi_{\mathcal{G}^{(i)}}$, is positive semi-definite. From Eq.~(\ref{eq:superoperator_G}), we obtain that
    \begin{align}
        \Phi_{\mathcal{G}^{(i)}} = \sum_{j, j' = 1}^d \bigg(K\sum_{k}  L_{i, k}\ket{j}\!\bra{j'}L_{i,k}^\dagger - \sum_{\alpha} S_\alpha O_{i, \alpha} \ket{j}\!\bra{j'}O_{i, \alpha}\bigg) \otimes \ket{j}\!\bra{j'}.
    \end{align}
    Now, suppose $\ket{\psi} = \sum_{j, j'}\psi_{j, j'}\ket{j, j'} \in \mathbb{C}^d\otimes \mathbb{C}^d $ is a two-qudit state and $\Psi = \sum_{j, j'}\psi_{j, j'}\ket{j}\!\bra{j'}$ is its corresponding matrix, then
    \begin{align}
        \bra{\psi}\Phi_{\mathcal{G}^{(i)}}\ket{\psi} &= K\sum_{k}\smallabs{\text{Tr}(L_{i, k}^\dagger \Psi)}^2 - \sum_{\alpha} S_\alpha \smallabs{\text{Tr}(O_{i, \alpha}\Psi)}^2 \nonumber \\
        &\numgeq{1} K \lambda_0\norm{\Psi}_F^2 - \sum_{\alpha} S_\alpha \norm{O_{i, \alpha}}_F^2 \norm{\Psi}_F^2 \nonumber \\
        &= (K \lambda_0 - S) \norm{\ket{\psi}}^2,
    \end{align}
    where in (1) we have used the fact that $\text{Tr}(A^\dagger B)^2 \leq \text{Tr}(A^\dagger A) \text{Tr}(B^\dagger B)$ and also the condition $\sum_k \smallabs{\text{Tr}(L_{i, k}^\dagger \Psi)}^2 \geq \lambda_0 \norm{\Psi}_F^2$ from the lemma statement. Therefore, if $K \geq S/\lambda_0$, which is implied by the condition $\kappa(t) \geq \sum_{\alpha}s_\alpha(t) / \lambda_0$, then $\mathcal{G}^{(i)}$ are completely positive. This in turn implies that the super-operator $\mathcal{M}_{t + \tau, t} = \mathcal{R}_{t + \tau,t} + \mathcal{G}_{t + \tau, t}$ is both completely positive and separability preserving, which proves the lemma.
\end{proof}

Similar to the case of the fermionic and bosonic models, this lemma can be combined with first-order Trotterization in the many-body setting to show high-noise separability for a broad class of noise models. In particular, we could consider noisy dynamics described by the master equation
\begin{align}
\frac{d}{dt}\rho(t) = -i[H, \rho(t)] + \kappa \sum_{i, k} \mathcal{D}_{L_{i, k}}
\end{align}
with $L_{i, k}$ satisfying the conditions in Lemma \ref{lemma:separability_spins} and
\begin{align}
    H(t) = \sum_{i} h_i(t) + \sum_{i < j}\sum_{\alpha} s_{\alpha}^{i, j}(t)( O_{i, \alpha} \otimes O_{j, \alpha}),
\end{align}
where $O_{i, \alpha}$ would be a Hermitian operator acting on the $i^\text{th}$ qudit chosen to be normalized such that $\norm{O_{i, \alpha}}_F = 1$. Introducing the ``inter-site interaction-strength" parameter $J$ as the smallest number satisfying
\begin{align}
\sum_{j > i} \smallabs{s_\alpha^{i, j}(t)} + \sum_{j < i} \smallabs{s^{j, i}_{\alpha}(t)} \leq J \text{ for all }i, t \geq 0,
\end{align}
we can then establish using Lemma \ref{lemma:separability_spins} that if $\kappa \geq J / \lambda_0$, then an initial separable state of the spins always evolves into a separable state.

\section{Counter-examples}\label{supplemental:counterexamples}
In this section, we consider the question of whether a counterpart of Theorem 1 can be obtained for the bosonic model, and if a counterpart for Theorem 2 can be established for the fermionic model.
\subsection{High-noise regime for the bosonic model is not convex-Gaussian at all times}

We will provide evidence that there is no counterpart of Theorem 1 for bosonic systems. That is, even with noise rates larger than the non-Gaussianity $(\kappa \gg U)$, one can still obtain states which are not convex-Gaussian. In section \ref{subsubsection:high_dephasing noise}, we provide numerical evidence that, with a single bosonic mode and dephasing noise greater than the non-Gaussianity $(\kappa_3 \geq U)$, states with negative Wigner function can be reached, which automatically implies lack of convex Gaussianity. In section \ref{subsubsection:high_incoherent_particle_regime}, we provide a stronger argument in the absence of dephasing noise: for noise models containing only incoherent particle loss and gain $(\kappa_3  =0)$, one can perform high-fidelity arbitrary gates even if the non-Gaussianity is much smaller than the noise rate, $U \ll \kappa$, provided that the Gaussian couplings $J,\Omega$ can be made sufficiently large, enabling the implementation of gates with effective error rates below the fault tolerance threshold \cite{noh2020_fault_tolerant_bosons,aharonov1999faulttolerantquantum,matsuura2024fault_tolerant_bosons}. As a consequence, not only is the state not guaranteed to remain convex-Gaussian, but the classical simulation of local observables is provably $\mathrm{BQP}$-hard.

\subsubsection{High dephasing noise regime}\label{subsubsection:high_dephasing noise}
Here, we provide a simple example with a single bosonic mode, where dephasing noise, no-matter how high, is unable to make the state convex Gaussian. Our analysis is centered on the Wigner function, which, for a single bosonic mode state $\rho$, is defined as 
\begin{align}
W(x,p)=\frac{1}{\pi} \int_{-\infty}^{\infty} \bra{x-y}\rho \ket{x+y}e^{2ipy}dy.
\end{align}
\noindent Since  quantum states with positive Wigner functions can be efficiently simulated \cite{Mari_2012_wigner}, the negativity of the Wigner function is regarded as a necessary resource for quantum advantage. Furthermore, a negative Wigner function rules out convex-Gaussianity, since all pure Gaussian states have nonnegative Wigner functions \cite{1974_hudson_theorem,walschaers2021_nongaussian_states}, and as a consequence convex Gaussian states do as well.

Specifically, we consider an initial state $\rho(0)=\ket{\alpha}\bra{\alpha}$, where $\ket{\alpha}$ represents the single-mode coherent state 
$\ket{\alpha}=e^{\alpha(a^{\dagger}-a)}\ket{\text{vac}}$, with $\alpha$ a real number. Then, the state is evolved under the Hamiltonian $H=Un^2$ and dephasing noise of rate $\kappa_3 = \kappa$, which yields the master equation
\begin{align}\label{eq:evolution_wigner}
\frac{d}{dt}\rho(t)=\mathcal{L}\rho(t)=-iU[n^2,\rho(t)]+\kappa \left(n\rho(t)n-\frac{1}{2} \{n^2,\rho(t)\}\right).
\end{align}
\begin{figure}
    \includegraphics[width=\textwidth]{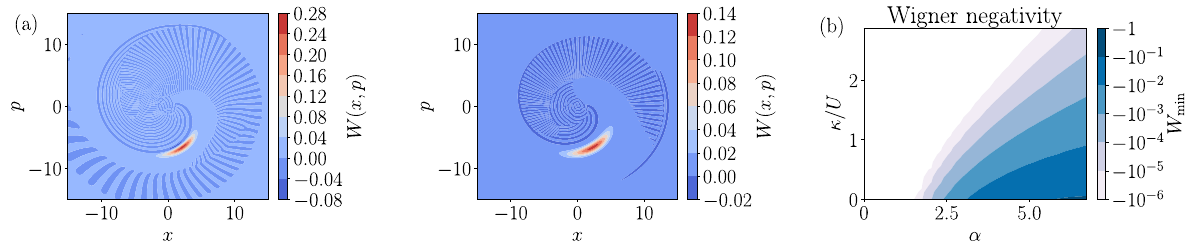}
    \caption{Representation of the Wigner function for the state $\rho(t)$ obtained by evolving Eq.~(\ref{eq:evolution_wigner}) with initial state $\rho(0)=\ket{\alpha}\bra{\alpha}$, with parameters $U=0.05$ and $t=0.5$. (a) The Wigner function $W(x,p)$ is represented in phase space for $\alpha=5$ and error rate $\kappa=0.1U$ (left) and $\kappa=2U$ (right). (b) Representation of the minimum value of the Wigner function $W_{\min}=\min_{x,p}W(x,p)$. One can clearly observe that, for a fixed value of $\kappa/U$, the Wigner function becomes more negative by increasing $\alpha$, which suggest that the Gaussian resources can effectively boost the non-Gaussianity of the system. Hence, even for large values of $\kappa$, a negative Wigner state might be reached, by increasing $\alpha$.
    }\label{fig:Wigner}
\end{figure}

\noindent In this setting, we numerically compute the minimum value of the Wigner function, $W_{\min}=\min_{x,p}W(x,p)$, and represent it in Fig. \ref{fig:Wigner}. One can appreciate that, even when $\kappa \geq U$, using a sufficiently large $\alpha$ results in a state with a negative Wigner function. Consequently, we do not expect an analogue of Theorem 1 to hold for bosons: even for a high dephasing noise rate, with a sufficiently large Gaussian displacement, an initially Gaussian state can evolve into Wigner negative (not convex-Gaussian) states at short times. The time-scale at which the state becomes Wigner negative is determined by both the value of $U$, as well as the displacement. For the one-mode problem considered above, Fig.~\ref{fig:Wigner_timescale} shows the relative wigner-negativity $-W_\text{min}/W_\text{max} = - \min_{x,p}W(x, p) / \max_{x,p} W(x, p)$ as a function of time. We find that while the state eventually becomes Wigner non-negative, even in the regime $\kappa_3 \gg U$, there is an intermediate temporal region that depends on $U$ and $\alpha$ where the state is Wigner negative. For a fixed $\alpha$, the time $t^*$ at which the state is maximally Wigner negative scales as $1 / U$, consistent with the fact that $U$ determines the strength of the process generating the non-Gaussianity or Wigner negativity in the dynamics.
\begin{figure}
    \includegraphics[width=\textwidth]{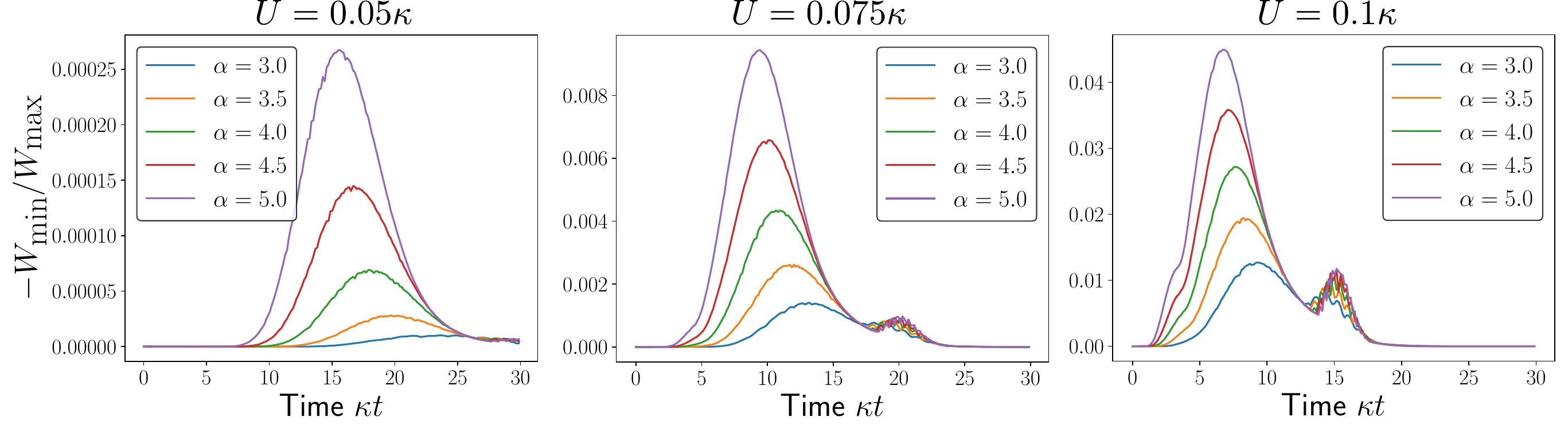}
    \caption{The relative Wigner negativity, quantified by $-W_\text{min}/W_\text{max} = - \min_{x,p}W(x, p) / \max_{x,p} W(x, p)$, for a single bosonic mode evolving under the Hamiltonian $H = Un^2$ as well as dephasing noise at rate $\kappa$. The initial state of the bosonic mode is the coherent state $\ket{\alpha}$. As expected, the relative Wigner negativity decreases as $U$ decreases, and at a fixed $\alpha$, the maximum Wigner negativity is attained at time-scales of $1/U$.
    }\label{fig:Wigner_timescale}
\end{figure}

We remark that this does not necessarily imply simulation hardness: the question of whether there is a threshold error $\kappa_{\text{th}}(U)$ depending on $U$ but not on $J,\Omega$ above which the classical simulation becomes tractable remains open.

\subsubsection{High incoherent particle loss and gain regime}\label{subsubsection:high_incoherent_particle_regime}
Here we analyze the complexity of classically simulating the bosonic system in the absence of dephasing noise ($\kappa_3=0$) and find that the problem does not become easy above a noise threshold depending exclusively on the non-Gaussian interaction strength, thus showing that no counterpart of Theorem 1 can exist for bosons in the absence of dephasing noise. Specifically, we show that, even when the non-Gaussian interaction strength is much smaller than the noise strength, $U \ll \kappa$, one can perform arbitrarily fast gates. Due to the threshold theorem, this allows for the implementation of fault-tolerant schemes \cite{noh2020_fault_tolerant_bosons,matsuura2024fault_tolerant_bosons}. To show this, it is enough to consider systems with only one bosonic mode per site ($L=1$) and only onsite non-Gaussianity ($U_C=0$).

We will start with a technical lemma: we will show that, for a system with $m$ modes evolving under the noise model in Eq.~(\ref{eq:lindbladian_noise_supplement}), the error induced by the noise can be upper bounded by $O(m \kappa t)$.

\begin{lemma}[Error bound between noisy and noiseless evolution]\label{lemma:error_bound}
Consider a bosonic system with $m$ modes evolving under the master equation
\begin{align*}
\frac{d}{dt}\rho(t)=\mathcal{L}\rho(t)=-i[H(t),\rho(t)]+\sum_{l=1}^2 \sum_{v=1}^m \kappa_l \mathcal{D}_{L^{(l)}_v},
\end{align*}
\noindent where $\mathcal{D}_L \rho = L \rho L^\dagger - \{L^\dagger L, \rho \} / 2$, $L_{v}^{(1)} = a_{v}, L^{(2)}_{v} = a_{v}^\dagger$, and $\kappa_1+ \kappa_2=1$. Assume that, for some integer $d$, $\rho(0)$ lies in the subspace $\mathcal{H}_{\leq d}$ of the Hilbert space spanned by the first $d+1$ levels  $\{\ket{0},\dots, \ket{d}\}$, and that the Hamiltonian $H(t)$ contains no couplings between the state $\ket{d}$ and any state $\ket{k}$ with $k>d$. Then, denoting by $\mathcal{U}(\cdot)=\mathcal{T}\exp \left(-i\int_0^t [H(s),\cdot]ds\right)$ the time-evolution in the noiseless case ($\kappa=0$), the error induced by the dissipation can be bounded as
\begin{align*}
\left\|\mathcal{T}\exp \left(\int_0^t \mathcal{L}(s)ds\right)\rho(0)-\mathcal{U}\rho(0)\right\|_1 \leq 2m\kappa t(d+1).
\end{align*}
\end{lemma}

\begin{proof}

Let us consider the following effective Hamiltonian:
\begin{align}\label{eq:effective_H}
H_{\mathrm{eff}}=H-i\frac{1}{2}\sum_{v=1}^m\left(\kappa_1 a_v^{\dagger}a_v+\kappa_2 a_va_v^{\dagger}\right).
\end{align}
Then, the time evolution may be written as 
\begin{align}\label{eq:expression_rho(t)}
\rho(t)=\mathcal{T}\exp \left(\int_0^t \mathcal{L}(s)ds\right)\rho(0)=\underbrace{\mathcal{T}\exp\left(-i\int_0^t [H_{\mathrm{eff}}(s),\cdot]ds\right)\rho(0)}_{\sigma}+\mathcal{N}(\rho(0))=\sigma+\mathcal{N}(\rho(0)),
\end{align}
\noindent where $\sigma$ is the (unnormalized) state obtained by evolving under the effective Hamiltonian, and $\mathcal{N}(\rho)$ is a completely positive channel that is not trace preserving. Naturally, $\mathrm{tr}(\sigma)+\mathrm{tr}(\mathcal{N}(\rho))=1$. The state $\sigma$ can be understood as the output when no errors occur, while $\mathcal{N}(\rho)$ captures the output with one or more errors.

For the $v^{\mathrm{th}}$ bosonic mode, we define the projector that truncates to at most $d$ particles as $\Pi_{v,\leq d}=\sum_{j=0}^d \ket{j}\bra{j}$. Then, $\Pi_{\leq d} = \otimes_v \Pi_{v,\leq d}$ is the projector onto $\mathcal{H}_{\leq d}$. Note that, since $H$ does not contain couplings to higher levels, neither does $H_{\mathrm{eff}}$. As a consequence, the dynamics of $\sigma$ are constrained to the first $d+1$ levels, and can be truncated.

Let us denote the truncated Hamiltonians by $\tilde{H}=\Pi_d H \Pi_d$,$\tilde{H}_{\mathrm{dis}}=\Pi_d H_{\mathrm{dis}} \Pi_d$. Using the definition of $\tilde{H}_{\mathrm{dis}}$ in Eq.~(\ref{eq:effective_H}), the operator norm of the truncated effective Hamiltonian $\tilde{H}_{\mathrm{dis}}$ can then be bounded as
\begin{align}\label{eq:norm_Hdis}
\|\tilde{H}_{\mathrm{dis}}\| \leq \frac{m}{2} \left(\kappa_1 d+\kappa_2 (d+1)\right) \leq \frac{m}{2}\left(\kappa_1+\kappa_2\right)d \leq \frac{m\kappa}{2}(d+1).
\end{align}
Let us now write $\sigma$ in a more convenient form as $\sigma=\lim_{N \rightarrow \infty}(O_N\rho(0)O_N^{\dagger})$, where
\begin{align}\label{eq:expression_O_N}
O_N=\prod_{k=1}^N \left(e^{-i\tilde{H}(kt/N)t/N}e^{-i\tilde{H}_{\mathrm{dis}}t/N}\right).
\end{align}
\noindent This expression can be derived, for example, from standard Trotterization techniques. We would now like to bound $\mathrm{tr}(\sigma)$, which can be understood as the probability of no errors occurring during the computation. Naturally, the unitary parts of the evolution in Eq.~(\ref{eq:expression_O_N}) are trace preserving, and we only need to bound the imaginary time evolution induced by the Hamiltonian $\tilde{H}_{\mathrm{dis}}$. Note also that, using Eq.~(\ref{eq:norm_Hdis}), the minimum singular value in each step can be bounded as 
\begin{align}\label{eq:min_value_exp}
\sigma_{\min}(e^{-\tilde{H}_{\mathrm{dis}}t/N}) \geq \exp\left(-\frac{m\kappa}{2} \frac{t}{N} \left(\kappa_1(d+1)+\kappa_2(d+1)\right)\right) \geq \exp\left(-\frac{m\kappa}{2}\frac{t}{N}(d+1)\right).
\end{align}
\noindent As a consequence, using Eq.~(\ref{eq:expression_O_N}) and Eq.~(\ref{eq:min_value_exp}), it can be easily checked that
\begin{align}
\sigma_{\min}(O_N^{\dagger}O_N) \geq \left[\sigma_{\min}(e^{-\tilde{H}_{\mathrm{dis}}t/N})\right]^{2N} \geq \exp\left[-m \kappa t(d+1)\right].
\end{align}
\noindent This allows us to bound the trace as 
\begin{align}
\mathrm{tr}(\sigma) = \lim_{N \rightarrow \infty} \mathrm{tr}(O_N^{\dagger}O_N \rho(0)) \geq \lim_{N \rightarrow \infty} \sigma_{\min}(O_N^{\dagger}O_N) \mathrm{tr}(\rho(0)) \geq \exp\left[-m \kappa t(d+1)\right].
\end{align}
\noindent As a consequence, since the total evolution of the system must be trace preserving, it immediately follows that $\mathrm{tr}(\mathcal{N}(\rho(0)) \leq 1-e^{-m \kappa td^2}$.

Now, let us bound the distance between $\sigma$ and the state obtained under ideal (noiseless) evolution, $\|\sigma-\mathcal{U}\rho(0)\|_1$. We use the fact that 
\begin{align}\label{eq:derivative_Hdis}
\frac{\partial}{\partial \kappa} e^{-\tilde{H}_{\mathrm{dis}}\frac{t}{N}}=-\frac{t}{N}\frac{\tilde{H}_{\mathrm{dis}}}{\kappa}e^{-\tilde{H}_{\mathrm{dis}}\frac{t}{N}} \quad \mathrm{and} \quad \|e^{-\tilde{H}_{\mathrm{dis}}\frac{t}{N}}\|, \|e^{-i\tilde{H}(\kappa t/N)\frac{t}{N}}\| \leq 1.
\end{align}
Using Eq.~(\ref{eq:norm_Hdis}), one can bound 
\begin{align}\label{eq:derivative_Hdis2}
\left\|\frac{\partial}{\partial \kappa} e^{-\tilde{H}_{\mathrm{dis}}\frac{t}{N}}\right\|=\left\|\frac{t}{N}\frac{\tilde{H}_{\mathrm{dis}}}{\kappa}e^{-\tilde{H}_{\mathrm{dis}}\frac{t}{N}} \right\| \leq \frac{t}{N \kappa} \|\tilde{H}_{\mathrm{dis}}\| \leq \frac{mt(d+1)}{2N}.
\end{align}
\noindent Furthermore, using the definition of $O_N$ in Eq.~(\ref{eq:expression_O_N}), the norm bound in Eq.~(\ref{eq:derivative_Hdis2}), and the fact that $\|e^{-\tilde{H}_{\mathrm{dis}}\frac{t}{N}}\|, \|e^{-i\tilde{H}(\kappa t/N)\frac{t}{N}}\| \leq 1$, one can bound
\begin{align}
\left\|\frac{\partial}{\partial \kappa}\sigma\right\| \leq2\left\|\frac{\partial}{\partial \kappa}O_N\right\| \leq \frac{2t}{\kappa} \smallnorm{\tilde{H}_{\mathrm{dis}}} \leq mt(d+1).
\end{align}
\noindent This directly yields a bound on the desired distance:
\begin{align}\label{eq:derivative_sigma}
\|\sigma-\mathcal{U}\rho(0)\|_1 = \left\| \int_0^\kappa \frac{\partial}{\partial \kappa} \sigma d\kappa \right\|_1 \leq \int_0^k \left\| \frac{\partial}{\partial \kappa} \sigma \right\|_1 d\kappa \leq m\kappa t(d+1).
\end{align}
Finally, this, together with Eqs.~(\ref{eq:expression_rho(t)}, \ref{eq:derivative_sigma}), implies that
\begin{align}
\left\|\mathcal{T}\exp \left(\int_0^t \mathcal{L}(s)ds\right)\rho(0)-\mathcal{U}\rho(0)\right\|_1 \leq \|\rho(t)-\sigma\|_1 + \|\mathcal{N}(\rho(0))\|_1 \leq m\kappa t(d+1)+1-e^{-m \kappa t (d+1)} \leq 2m\kappa t(d+1),
\end{align}
\noindent which proves the lemma.

\end{proof}

We will now show how one can use a single bosonic mode Hamiltonian to apply arbitrary single qubit gates with high fidelity, even when the nonlinearity is much smaller than the noise strength. We note that a similar result is already shown in Refs.~\cite{liangjiang2023universalcontrolbosonicsystems,eickbusch2022fast_gates}. To do this, we will consider a single bosonic mode with $\kappa \gg U$, where both the error rate and non-Gaussianity $U$ are fixed, and study the effective error rate in the asymptotic limit of large Gaussian strength.

\begin{lemma}[Single-qubit gates, from Ref.~\cite{liangjiang2023universalcontrolbosonicsystems}]\label{lemma:single_qubit}
Consider a single bosonic mode under the master equation
\begin{align*}
\frac{d}{dt}\rho=\mathcal{L}\rho=-i[H(t),\rho]+\sum_{l=1}^2\kappa_l \mathcal{D}_{{L}^{(l)}} \rho,
\end{align*}
\noindent with Hamiltonian
\begin{align*}
H(t)=U(t)a^{\dagger 2}a^2 + \left(\Lambda_1(t)a^{\dagger}+\Lambda_2(t)a^{\dagger 2}+\mathrm{h.c.}\right)+\Delta(t)a^{\dagger}a,
\end{align*}
\noindent where $\mathcal{D}_L \rho = L \rho L^\dagger - \{L^\dagger L, \rho \} / 2$, $L^{(1)} = a, L^{(2)} = a^\dagger$ and $L^{(3)} = a^\dagger a = n$, and $\kappa_1 + \kappa_2 = \kappa$, and the strength of the Gaussian terms is bounded by $P$, $|\Lambda_1(t)|,|\Lambda_2(t)|,|\Delta(t)| \leq P$. Then, for any single-qubit quantum unitary operation $\mathcal{U}$ (i.e.~$\mathcal{U}(\cdot)=U(\cdot)U^{\dagger}$ for some single-qubit gate $U$), the Lindbladian $\mathcal{L}(t)$ can approximate $\mathcal{U}$, $\|(\mathcal{T} e^{\int_0^t\mathcal{L}(s)ds}-\mathcal{U})\rho_0\|_{1} \leq \tilde{O}(\kappa (U^2P)^{-1/3})$, in time $t=O((U^2P)^{-1/3})$, with $\rho_0$ a single-qubit state.
\end{lemma}
\begin{proof}

We will show that, by tuning the parameters in the Hamiltonian $H(t)$, one can generate $T$, $S$, and $\sqrt{X}$ gates, which is sufficient for arbitrary single-qubit rotations. 
For implementing a $T$ gate or an $S$ gate, one simply has to set $\Omega(t)=0$ and $\Delta(t)=P$. This yields the Hamiltonian $H=U (a^{\dagger 2}a^2)/2 + P a^{\dagger}a$. Since there are no couplings between the states $\ket{0},\ket{1}$ and the rest, we can restrict ourselves to the subspace spanned by $\{\ket{0},\ket{1}\}$. Denoting the projector onto this subspace $\Pi_1=\ket{0}\bra{0}+\ket{1}\bra{1}$, the projection of $H_{\alpha}$ on the blockaded subspace yields $\Pi_1H \Pi_1=P\ket{1}\bra{1}$. Then, evolving under the Hamiltonian for time $t=3\pi/(2P)$ yields an $S$ gate, while evolving  for time $t=3\pi/(4P)$ yields a $T$ gate. Therefore, applying the error bound in Lemma \ref{lemma:error_bound}, $T$ gates and $S$ gates can be implemented with precision $O(\kappa/P)$ in time $t=O(1/P).$

Now, let us consider the problem of applying a $\sqrt{X}$ gate. This can be done by using the construction from Refs.~\cite{liangjiang2023universalcontrolbosonicsystems,lingenfelter2021_fock_state_generation}, by going to a displaced frame. We will first show that, considering the Hamiltonian in a displaced frame, one can implement a fast $\sqrt{X}$ gate. Then, we will show that one can go to the displaced frame by applying fast pulses at the beginning and end of the computation, hence enabling the application of a high-fidelity fast $\sqrt{X}$ gate in the laboratory frame, even in the presence of errors.

First, let us consider the Hamiltonian in a frame displaced by $\alpha(t)$, $a \rightarrow a+\alpha(t)$. We write the Hamiltonian in the displaced frame as $H_{\alpha(t)}$ and the noise as $\mathcal{D}_{L^{(l)};\alpha(t)}$. Note that the noise in the displaced frame can be written as 
\begin{equation}\label{eq:displaced_noise}
\sum_{l=1}^2\kappa_l \mathcal{D}_{L^{(l)},\alpha(t)}(\cdot)=\sum_{l=1}^2\kappa_l \mathcal{D}_{L^{(l)}}+\frac{i}{2} (\kappa_1 - \kappa_2)\left[i(\alpha(t)a^{\dagger}-\alpha^*(t)a),(\cdot)\right].
\end{equation}
Furthermore, the displaced Hamiltonian $H_{\alpha(t)}$ may be written as
\begin{align}
H_{\alpha(t)}=U(t)a^{\dagger 2}a^2+\tilde{\Delta}(t)a^{\dagger}a+(\tilde{\Lambda}_1(t)a^{\dagger}+\tilde{\Lambda}_2(t)a^{\dagger 2}+\tilde{\Lambda}_3(t)a^{\dagger 2}a+\mathrm{h.c.}),
\end{align}
\noindent where 
\begin{align}\label{eq:displacement_eqs}
    &\tilde{\Delta}(t)=\Delta(t)+4U(t)|\alpha(t)|^2,  \nonumber \\
    &\tilde{\Lambda}_2(t)=\Lambda_2(t)+2U(t) \alpha(t)^2, \nonumber \\
    &\tilde{\Lambda}_1(t)=\Lambda_1(t)+ \alpha
    \Delta(t)+2\alpha(t)^*\Lambda_2(t)+2U(t)|\alpha(t)|^2\alpha(t)-\frac{1}{2}i  \alpha(t)(\kappa_1 - \kappa_2), \nonumber \\
    &\tilde{\Lambda}_3(t)=2U(t)\alpha(t).  
\end{align}
\noindent where the noise term from Eq~(\ref{eq:displaced_noise}) has already been absorbed into the Hamiltonian. The master equation in the displaced frame is then 
\begin{align}
\frac{d}{dt}\rho_{\alpha(t)}=\mathcal{L}_{\alpha(t)}\rho_{\alpha(t)}=-i[H_{\alpha(t)},\rho_{\alpha(t)}]+ \sum_{l=1}^2\kappa_l \mathcal{D}_{L^{(l)}}.
\end{align}
\noindent By suitably choosing the parameters so that $\tilde{\Delta}(t)=\tilde{\Lambda}_1(t)=\tilde{\Lambda}_2(t)=0$, the Hamiltonian becomes $H_{\alpha(t)}=2U(t)\alpha(t) a^{\dagger}(n-1)+ \mathrm{h.c.}$. 

Crucially, one can notice that the Hamiltonian $H_{\alpha(t)}$ is blockaded, since it does not contain couplings to state $\ket{2}$. Therefore, in the noiseless case, the dynamics will be restricted to the qubit subspace spanned by $\{\ket{0},\ket{1}\}$. Denoting the projector onto this subspace by $\Pi_1=\ket{0}\bra{0}+\ket{1}\bra{1}$, the projection of $H_{\alpha(t)}$ onto the blockaded subspace yields $\Pi_1 H_{\alpha(t)}\Pi_1=-2U \alpha(t) X$. Let us pick a constant $\alpha(t)=\alpha_F$. It is then clear that evolving under the Hamiltonian $H_{\alpha_F}$ for a time $t=\pi/(8U \alpha_F)$ produces a $\sqrt{X}$ gate.

We will pick the displacement to be $\alpha_F=\Theta((P/U)^{1/3})$, since it is the largest displacement that can simultaneously fulfill Eq.~(\ref{eq:displacement_eqs}) and the restriction that $|\Lambda_1(t)|,|\Lambda_2(t)|,|\Delta(t)| \leq P$.

Naturally, one is interested in performing operations in the laboratory frame, which means that at the beginning ($t=0$) and end ($t=t_F$) of the computation the displacement is $\alpha(0)=\alpha(t_F)=0$. This can be achieved by simply applying a displacement term in the beginning and end of the computation. Hence, the computation can be performed in three steps. First, a displacement term is applied for a time $t_{\mathrm{dis}}$ to go from $\alpha(0)=0$ to $\alpha(t_{\mathrm{dis}})=\alpha_F$. Then, the gate is performed in the frame displaced by $\alpha_F$, which takes time $t_{\mathrm{gate}}=\pi/(8U \alpha_F)$. Finally, the displacement is taken to $0$ again, which takes time $t_{\mathrm{dis}}$. Therefore, the total computation time for a $\sqrt{X}$ gate is $t_{\sqrt{X}}=2t_{\mathrm{dis}}+t_{\mathrm{gate}}$.

In order to achieve the desired displacement $\alpha_F$, one can apply the Hamiltonian
\begin{align}
H(t)=i(P-\alpha_F \kappa/2)(a^{\dagger}-a)+\frac{i }{2}(P-\alpha_F)t(\kappa_1 - \kappa_2)(a^{\dagger}-a),
\end{align}
\noindent where the first term takes the system to the frame displaced by $\alpha(t)=(P-\alpha_F \kappa/2)$, and the second term corrects the contributions of the noise. The choice of parameters ensures that $|\Lambda_1| \leq P$ at all times. Specifically, in the displaced frame, the system evolves under the master equation
\begin{align}
\frac{d}{dt}\rho_{\alpha(t)}= \sum_{l=1}^3\kappa_l \mathcal{D}_{L^{(l)}} \quad \mathrm {for} \quad t \leq t_{\mathrm{dis}},
\end{align}
\noindent with $\alpha(0)=0$ and $\alpha(t_{\mathrm{dis}})=\alpha_F$, and  $t_{\mathrm{dis}}=\alpha_F/(P-\alpha_F)=O(\alpha_F/P)=O((P^2U)^{-1/3})$. Let us denote the total evolution time by $t_{\sqrt{X}}=2t_{\mathrm{dis}}+t_{\mathrm{gate}}$. Note that $t_{\mathrm{gate}}=O((U\alpha_F)^{-1})=O((PU^2)^{-1/3})$, while $t_{\mathrm{dis}}=O(\alpha_F/P)=O((P^2U)^{-1/3})$. In the large $P$ limit, it is clear that $t_{\mathrm{dis}} \ll t_{\mathrm{gate}}$, and the total time scales as $t_{\sqrt{X}}=t_{\mathrm{gate}}+2t_{\mathrm{dis}}=O(t_{\mathrm{gate}})=O((PU^2)^{-1/3})$. From the Solovay-Kitaev theorem, it follows that any single-qubit rotation can be approximated to precision $\varepsilon$ in time $t=O(t_{\sqrt{X}}\log^c(1/\varepsilon))$ for some constant $c<2$. 
We can now bound the total contribution of the error: straightforward application of Lemma \ref{lemma:error_bound} shows that the error after time $t=\tilde{O}(t_{\sqrt{X}})=\tilde{O}((U^2P)^{-1/3})$ is
\begin{align}
\|(\mathcal{T}e^{\int_0^t \mathcal{L}(s)ds}-\mathcal{U})\rho_0\|_{1} \leq\tilde{O} \left(\frac{\kappa}{(U^2P)^{1/3}}\right),
\end{align}
\noindent where $\tilde{O}$ hides polylogarithmic factors. This proves the lemma.
\end{proof}

So far we have shown that one can make arbitrary single-qubit gates with high fidelity even if the noise is much larger than the non-Gaussianity, $\kappa \gg U$, as long as one can increment the strength of the Gaussian terms, $P \gg \kappa^3/U^2$. We will now show how to implement entangling gates, which is enough to obtain a universal gate-set.

\begin{lemma} [Two-qubit gates]\label{lemma:boson_2q_gates}
Consider two bosonic modes evolving under the master equation 
\begin{align}
\frac{d}{dt}\rho(t)=-i[H(t),\rho(t)]+\sum_{l=1}^2 \sum_{v=1}^2  \kappa_l \mathcal{D}_{L^{(l)}_v},
\end{align}
where $H(t)$ is the Hamiltonian $H(t)=H_1(t)+H_2(t)+ig(t)[a_1a_2^{\dagger}-a_1^{\dagger}a_2]$, with
\begin{align}
H_i(t)=U_i(t)a_i^{\dagger 2}a_i^2+\left(\Lambda_{i,1}(t)a_i^{\dagger} +\Lambda_{i,2}(t)a_i^{\dagger 2} + \mathrm{h.c.}\right)+\Delta_i(t)a_i^{\dagger}a_i,
\end{align}
\noindent and $\mathcal{D}_L \rho = L \rho L^\dagger - \{L^\dagger L, \rho \} / 2$, $L_{v}^{(1)} = a_{v}, L^{(2)}_{v} = a_{v}^\dagger$, and $\kappa_1 + \kappa_2 = \kappa$. 
Assume that the strength of the Gaussian terms is bounded by  $P$ ($|\Lambda_{i,1}(t)|,|\Lambda_{i,2}(t)|,|\Delta_i(t)|,|g(t)| \leq P$).

Then, for any two-qubit quantum unitary operation $\mathcal{U}$ (i.e.~$\mathcal{U}(\cdot)=U(\cdot)U^{\dagger}$ for some single-qubit gate $U$), the Lindbladian $\mathcal{L}(t)$ can implement a time evolution that approximates $\mathcal{U}$, $\| (\mathcal{T} e^{\int_0^t\mathcal{L}(s)ds}-\mathcal{U})\rho_0\|_{1} \leq \tilde{O}(\kappa (U^2P)^{-1/3})$, for a time $t= O((U^2P)^{-1/3})$, with $\rho_0$ a two-qubit state.

\end{lemma}
\begin{proof}
In Lemma \ref{lemma:single_qubit} it is shown how to use $H(t)$ to generate arbitrary single-qubit gates on either of the modes with arbitrarily high fidelity. Hence, it is only necessary to show how to apply an entangling two-qubit gate in order to have a universal gate-set. 

To do this, let us consider the collective modes $b_1=(a_1+a_2)/\sqrt{2}$ and $b_2=(a_1-a_2)/\sqrt{2}$. We will denote by $\ket{j,k}_{a_1,a_2}=(j!\, k!)^{-1/2}(a_1^{\dagger})^j(a_2^{\dagger})^k \ket{0,0}$ the Fock states in the original basis, and $\ket{j,k}_{n_1,n_2}=(j!\, k!)^{-1/2}(b_1^{\dagger})^j(b_2^{\dagger})^k \ket{0,0}$. One obtains that
\begin{align}\label{eq:mixing_states}
& \ket{0,0}_{a_1,a_2}=\ket{0,0}_{b_1,b_2}, \nonumber \\
&\ket{0,1}_{a_1,a_2}=\frac{1}{\sqrt{2}}\left(\ket{1,0}_{b_1,b_2}-\ket{0,1}_{b_1,b_2}\right), \nonumber \\
&\ket{1,0}_{a_1,a_2}=\frac{1}{\sqrt{2}}\left(\ket{1,0}_{b_1,b_2}+\ket{0,1}_{b_1,b_2}\right), \nonumber\\
&\ket{1,1}_{a_1,a_2}=\frac{1}{2}\left(\ket{2,0}_{b_1,b_2}-\ket{0,2}_{b_1,b_2}\right). 
\end{align}
\noindent Let us now denote by $U$ the single-mode unitary that maps $U\ket{0}_{b_1}=\ket{0}_{b_1}$, $U\ket{1}_{b_1}=i\ket{1}_{b_1}$, $U\ket{2}_{b_1}=\ket{2}_{b_1}$. Using Eq.~(\ref{eq:mixing_states}), it can be seen that $U$ will act in the $a_1,a_2$ basis as 
\begin{align}
& U\ket{0,0}_{a_1,a_2}=\ket{0,0}_{a_1,a_2}, \nonumber \\
&U\ket{0,1}_{a_1,a_2}=\frac{1}{2}\left((1+i)\ket{0,1}_{a_1,a_2}-(1-i)\ket{1,0}_{a_1,a_2}\right), \nonumber \\
&U\ket{1,0}_{a_1,a_2}=\frac{1}{2}\left(-(1-i)\ket{0,1}_{a_1,a_2}+(1+i)\ket{1,0}_{a_1,a_2}\right),  \nonumber \\
&U\ket{1,1}_{a_1,a_2}=\ket{1,1}_{a_1,a_2}. 
\end{align}

Hence, $U$ is clearly an entangling gate between modes $a_1$ and $a_2$. Furthermore, $U$ can be implemented in a fast manner using the same technique as in Lemma \ref{lemma:single_qubit}. Let us detail the procedure. First, one can evolve the system under the term $P(a_1^{\dagger}a_2+\mathrm{h.c})$, which induces the mixing of the modes $a_1 \rightarrow b_1$ and $a_2 \rightarrow b_2$ in time $t=O(1/P)$. The Hamiltonian in the new basis, $H_b(t)$, may be written us
\begin{align}\label{eq:Hamiltonian_hb}
H_b(t)=U_1b_1^{\dagger 2}b_1^2 + (\Lambda_{1,1}b_1^{\dagger}+\Lambda_{1,2}b_1^{\dagger 2}+\mathrm{h.c.})+\Delta_1b_1^{\dagger}b_1,
\end{align}
\noindent where we have chosen $\Delta_2=\Lambda_{2,1}=\Lambda_{2,2}=U_2=0$. One can note that the technique in Lemma \ref{lemma:single_qubit} can be readily applied to the Hamiltonian $H_b$ in Eq.~(\ref{eq:Hamiltonian_hb}). That is, one can go to a frame in which $b_1$ is displaced by $\alpha(t)$, $b_1 \rightarrow b_1+\alpha(t)$. As shown in the proof of Lemma \ref{lemma:single_qubit}, a suitable choice of the parameters leads to the displaced Hamiltonian 
$H_{b,\alpha(t)}=U_1(t)(b_1^{\dagger})^2b_1^2+2U_1 [\alpha(t) b_1^{\dagger}(b_1^{\dagger}b_1-2)+\mathrm{h.c.}]$. Note that this Hamiltonian contains no couplings between $\ket{2}_{b_1}$ and $\ket{3}_{b_1}$, and hence the subspace spanned by $\{\ket{0}_{b_1},\ket{1}_{b_1},\ket{2}_{b_1}\}$ is blockaded. Furthermore, one can rewrite
\begin{align}
H_{b,\alpha(t)}=U_1(t)(b_1^{\dagger})^2b_1^2+2U_1(t)\mathrm{Re}(\alpha(t))H_{\alpha,R}+2U_1(t)\mathrm{Im}(\alpha(t))H_{\alpha,I},
\end{align}
\noindent with $H_{\alpha(t),R}=b_1^{\dagger}(b_1^{\dagger}b_1-2)+(b_1^{\dagger}b_1-2)b_1$ and $H_{\alpha(t),I}=i(b_1^{\dagger}(b_1^{\dagger}b_1-2)-(b_1^{\dagger}b_1-2)b_1)$. One can compute the commutator $[H_{\alpha(t),R},H_{\alpha(t),I}]=2i(3 b_1^{\dagger 2}b_1^2-6b_1^{\dagger}b_1+4)$, which is diagonal, and can clearly implement the gate $U$, which consists only of a phase rotation of the state $\ket{1}_{b_1}$. In fact, the analysis in Ref.~\cite{liangjiang2023universalcontrolbosonicsystems} shows that one can generate arbitrary unitaries in the subspace spanned by $\{\ket{0}_{b_1},\ket{1}_{b_1},\ket{2}_{b_1}\}$. 

Following the analysis in Lemma \ref{lemma:single_qubit}, the gate $U$ can then be implemented in time $t_U=O((U^2P)^{-1/3})$. This is also the dominant source of error, since applying the displacement takes time $t_{\mathrm{dis}}=O((P^2U)^{-1/3})$, and going to the collective mode $b_1$ takes time $t_b=O(1/P)$, and therefore $t_U \gg t_{\mathrm{dis}},t_b$. Hence, the total error will be $O(\kappa t_U)=O(\kappa(U^2P)^{1/3})$.

Therefore, we have shown how to implement an entangling gate. Together with the implementation of arbitrary single-qubit gates and the Solovay-Kitaev theorem (which introduces and additional polylogarithmic factor), this proves that any 2-qubit gate $\mathcal{U}$ can be implemented by evolving the Lindbladian $\mathcal{L}(t)$ up to precision
\begin{align}
\|(\mathcal{T}e^{\int_0^t \mathcal{L}(s)ds}-\mathcal{U})\rho_0\|_{1} \leq \tilde{O} \left(\frac{\kappa}{(U^2P)^{1/3}}\right),
\end{align}
\noindent where $\rho_0$ is a two-qubit state, and the error bound follows directly from Lemma \ref{lemma:error_bound}.

\end{proof}

So far, we have shown how the bosonic Hamiltonian can be used to generate high-fidelity universal gates. Specifically, we have seen that 2-qubit gates can be implemented in time $t=\tilde{O}((PU^2)^{-1/3})$, where $U$ is the non-Gaussian strength, and $P$ refers to the maximum absolute value allowed for the Gaussian terms of the Hamiltonian. Let us now consider a system with $nL$ bosonic modes as described in section \ref{subsection:model} and study the asymptotic scaling with $n$. In this case, $|J_{i,j}^{\alpha,\alpha^{\prime}}(t)|,|\Omega_i^{\alpha}(t)| \leq P$. Let us assume that $J,\Omega=\Theta(P)$ (this is the case, for example, for geometrically local Hamiltonians). Then, provided that the Gaussian couplings $J,\Omega=O(1)$ can be arbitrarily large constants, it follows from Lemma \ref{lemma:boson_2q_gates} that any gate can be implemented to an arbitrarily small (but independent of $n$) precision.  Since this allows for the implementation of gates with an effective gate error below that of the threshold theorem \cite{noh2020_fault_tolerant_bosons,matsuura2024fault_tolerant_bosons,aharonov1999faulttolerantquantum}, this implies that fault-tolerant circuits can be implemented. We remark that, in addition to high-fidelity gates, fault-tolerant constructions usually require the ability to implement $\mathrm{RESTART}$ operations to provide fresh qubits \cite{aharonov1999faulttolerantquantum}. For spin systems in the presence of non-unital noise, cooling algorithms \cite{boykin2002_cooling1,schulman1999_cooling2,Alhambra2019heatbathalgorithmic} in conjunction with the noise channel can be leveraged to implement such an operation \cite{benor2013quantumrefrigerator,Trivedi2022_transitions,shtanko2024complexitylocalquantumcircuits}. In our case, a similar scheme would be needed; however, we leave a careful analysis of the construction for future work.

\subsection{High-noise regime of the fermionic model is not separable at all times}
For the bosonic model, Theorem 2 establishes that, in the presence of a sufficiently incoherent high particle loss or incoherent particle gain, the state of the bosonic model is separable at all times. In this subsection, we show that such a result cannot be true for the fermionic model. We show this for two notions of separability for fermions \cite{maricarmen2007_entanglement_fermions,Moriya_2006_entanglement_fermions}---the first notion  holds for all observables, and the second weaker notion holds for parity-conserving, or even, observables. 

Throughout this section, it will be enough for us to consider separability in the bi-partite setting. We will consider $m$ fermionic modes which are divided into two subgroups of modes, $A$ with modes $\{1, 2, \dots, m_A\}$ and $B$ with modes $\{m_A + 1, m_A + 2, \dots, m\}$. Recall that an operator on the fermionic Hilbert space is an element of the algebra generated by $\{c_v^1, c_v^2\}_{v \in \{1, 2 \dots m\}}$ or alternatively by $\{a_v, a_v^\dagger\}_{v\in\{1, 2 \dots m\}}$. An operator acting on sub-system $A$ will be an element of the algebra generated by $\{a_v, a_v^\dagger\}_{v \in \{1, 2 \dots m_A\}}$ and, similarly, an operator acting on the sub-system $B$ will be an element of the algebra generated by $\{a_v, a_v^\dagger\}_{v \in \{m_A + 1, m_A + 2\dots m_B\}}$. The parity operator of a set $S \subseteq \{1, 2 \dots m\}$ of fermionic modes is $P_A = \exp({i\pi \sum_{i \in S}a_i^\dagger a_i})$. Operators that conserve the parity operator are called \emph{even operators}. Physically relevant states of the fermionic modes are restricted to be even operators---note, however, that a physical operator that is even on all the fermionic modes is not necessarily even on a subset of these fermionic modes.

\begin{definition}\label{def:ferm_sep_all_obs}
    A state $\rho$ of $m$ fermionic modes will be called a \textbf{product state with respect to all observables} on the bi-partition $A | B$ if there exist states $\rho_A$ for the modes in $A$ and $\rho_B$ for the modes in $B$ such that, for all operators $O_A$ supported on $A$ and $O_B$ supported on $B$,
    \begin{align*}
    \textnormal{Tr}(O_A O_B \rho) = \textnormal{Tr}(O_A \rho_A) \textnormal{Tr}(O_B \rho_B).
    \end{align*}
    A state $\rho$ of the $m$ fermionic modes will be called a \textbf{separable state with respect to all observables} on the bi-partition $A | B$ if it can be expressed as a convex-combination of such product states.
\end{definition}
\noindent We remark that it was shown in Ref.~\cite{Moriya_2006_entanglement_fermions} that if $\rho$ is an even operator, which is also a product state as per definition \ref{def:ferm_sep_all_obs}, then $\rho_A$ and $\rho_B$ are also both even operators.

\begin{definition}\label{def:ferm_sep_even_obs}
    A state $\rho$ of $m$ fermionic modes will be called a \textbf{product state with respect to even observables} on the bi-partition $A | B$ if there exist states $\rho_A$ for the modes in $A$ and $\rho_B$ forthe  modes in $B$ such that, for all even operators $O_A^{(+)}$ supported on $A$ and $O_B^{(+)}$ supported on $B$,
    \begin{align*}
    \textnormal{Tr}(O_A^{(+)} O_B^{(+)} \rho) = \textnormal{Tr}(O_A^{(+)} \rho_A) \textnormal{Tr}(O_B^{(+)} \rho_B).
    \end{align*}
    A state $\rho$ of the $m$ fermionic modes will be called a \textbf{separable state with respect to even observables} on the bi-partition $A | B$ if it can be expressed as a convex-combination of such product states.
\end{definition}
\noindent As discussed in Ref. \cite{maricarmen2007_entanglement_fermions} , while separability with respect to all observables implies separability with respect to even observables, the converse is not necessarily true. This can be seen explicitly in a simple 2-mode example---consider the state $\ket{\psi} = (a_1^\dagger + a_2^\dagger) \ket{\text{vac}}/ \sqrt{2}$. This state is not separable as per definition \ref{def:ferm_sep_all_obs}---to see this, one can use the 2-mode separability criteria from Refs.~\cite{Moriya_2006_entanglement_fermions,maricarmen2007_entanglement_fermions} , which we also provide in Lemma \ref{lemma:sep_criteria_2_mode}. However, this state is separable as per definition \ref{def:ferm_sep_even_obs}---to see this, we note that, for any even observable $O_1^{(+)}$ on the first fermionic mode and $O_2^{(+)}$ on the second fermionic mode,
\begin{align}
\bra{\psi}O_1^{(+)}O_2^{(+)}\ket{\psi} = \frac{1}{2} \big(\bra{\phi_1}O_1^{(+)}\ket{\phi_1} \bra{\phi_2}O_2^{(+)}\ket{\phi_2} + \bra{\theta_1}O_1^{(+)}\ket{\theta_1} \bra{\theta_2}O_2^{(+)}\ket{\theta_2}\big),
\end{align}
where $\ket{\phi_1} = a_1^\dagger\ket{\text{vac}}, \ket{\phi_2} = \ket{\text{vac}}$, $\ket{\theta_1} = \ket{\text{vac}}$ and $\ket{\theta_2} = a_2^\dagger \ket{\text{vac}}$.

\subsubsection{Non-separability for any observable}
Here, we provide a simple 2-mode example which shows that, unlike the bosonic model, no matter how high the rate of particle loss and incoherent particle gain is in the fermionic model, the dynamics of the fermionic model is not separability preserving with respect to all observables (definition \ref{def:ferm_sep_all_obs}). To establish this result, we first review the 2-mode seperability criteria from Refs.~\cite{maricarmen2007_entanglement_fermions,Moriya_2006_entanglement_fermions}.
\begin{lemma}[2-mode separability, Refs.~\cite{maricarmen2007_entanglement_fermions,Moriya_2006_entanglement_fermions}]\label{lemma:sep_criteria_2_mode}
A 2-mode density matrix $\rho$ is separable with respect to all observables (definition \ref{def:ferm_sep_all_obs}) if and only if it is diagonal in the computational basis. 
\end{lemma}

\begin{proposition} \label{prop:non_sep_fermion_all_obs}
    Consider a system with $m = 2$ fermionic modes, with the sub-system $A$ with mode 1 and sub-system $B$ with mode 2 whose density matrix $\rho(t)$ satisfies the Lindblad master equation 
    \begin{align*}
    \frac{d}{dt}\rho(t) = -i[J(a_1^\dagger a_2 + a_2^\dagger a_1), \rho(t)] + \sum_{j = 1}^4 \big(\kappa_1\mathcal{D}_{a_j}+\kappa_2\mathcal{D}_{a_j^\dagger}\big) \rho(t), 
    \end{align*}
    Then, $\forall \kappa_1, \kappa_2 > 0$, $\exists \rho(0)$ which is separable with respect to all observables (Definition \ref{def:ferm_sep_even_obs}) such that $\rho(t)$ cannot be separable for all $t \geq 0$.
\end{proposition}
\begin{proof}
    Throughout this proof, ``separability" refers to separability with respect to all observables (Definition \ref{def:ferm_sep_all_obs}). Consider the initial state $\rho(0) = a_1^\dagger \ket{\text{vac}}\!\bra{\text{vac}}a_1$. Note that $\rho(0)$ is trivially separable---we now establish that, for a small time $t$,   
    $\rho(t) = e^{\mathcal{L}t}(\rho(0))$ is not seperable to $O(t^2)$ for any $\kappa$, which is enough to contradict separability of $\rho(t)$ at all times $t$. Now,
    \begin{align}
    \rho(t) &= \rho(0) + t\mathcal{L}\rho(0) + O(\varepsilon^2),
    \end{align}
    which, in the computational basis (i.e.~$\ket{0, 0} = \ket{\text{vac}}, \ket{1, 0} = a_1^\dagger \ket{\text{vac}}, \ket{0, 1} = a_2^\dagger \ket{\text{vac}}, \ket{1, 1} = a_1^\dagger a_2^\dagger \ket{\text{vac}}$), satisfies $\abs{\bra{1, 0}\rho(\varepsilon)\ket{0, 1}} = Jt + O(t^2)$.
    From the separability criteria in Lemma \ref{lemma:sep_criteria_2_mode}, it then follows that there cannot exist a separable state $\sigma(t)$ such that $\norm{\rho(t) - \sigma(t)}_1 \leq O(t^2)$, no matter what the rate $\kappa$ is. 
\end{proof}

\noindent While the analysis above indicates that there would be times when the state $\rho(t)$ would be entangled, we can also analyze the time-scale at which this entanglement is developed. In Fig.~\ref{fig:fermionic_allobs_timescale}, we simulate the two-fermionic-mode model from Proposition \ref{prop:non_sep_fermion_all_obs} for $\kappa_1 =\kappa_2 =\kappa$ and numerically compute the entanglement measure $E(\rho) = \sum_{b, b': b \neq b'} \abs{\rho_{b, b'}}$ (i.e., the 1-norm of a vector formed with the off-diagonal elements of $\rho$) as a function of $t$. Note that from Lemma \ref{lemma:sep_criteria_2_mode}, this measure quantified how non-separable the two-mode state is when considering all observables. We observe that this measure becomes $0$ at long times, however it is largest at time $t^* \sim 1/\kappa$ --- thus consistent with the analysis of Proposition \ref{prop:non_sep_fermion_all_obs}, even at large $\kappa$, the fermionic state becomes entangled at short times and, unlike the bosonic model, does not exhibit a threshold behavior of transitioning to an always separable state at sufficiently large $\kappa$.

\begin{figure}
    \includegraphics[width=0.75\textwidth]{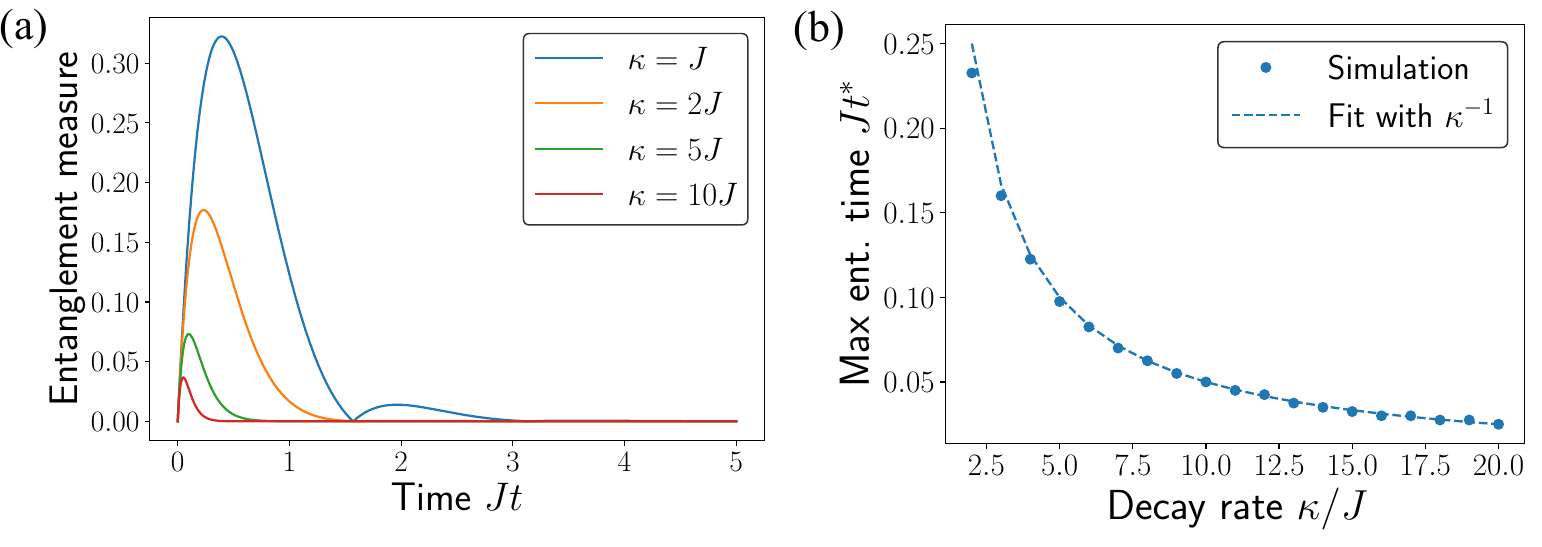}
    \caption{Numerical simulation of the two-fermionic-mode model considered in Proposition \ref{prop:non_sep_fermion_all_obs} with both particle loss rate $\kappa_1$ and particle gain rate $\kappa_2$ set to $\kappa$ and $\rho(0) = a_1^\dagger \ket{\text{vac}}\!\bra{\text{vac}}a_1$. (a) Time evolution of the entanglement measure computed by adding all off-diagonal elements of the two-mode density matrix of the model. As per Lemma \ref{lemma:sep_criteria_2_mode}, this quantifies non-separability with respect to all observables for 2-mode fermionic systems. (b) The time $t^*$ at which the entanglement measure computed in (a) is maximum as a function of $\kappa$. We see that entanglement is developed at time-scales $t^* \sim 1 / \kappa$ no matter how large $\kappa$ is.
    }\label{fig:fermionic_allobs_timescale}
\end{figure}

\subsubsection{Non-separability for even observables}
\noindent We will use the following lemma, which reduces the problem of checking the non-separability of a 4-mode fermionic state with respect to even observables to an effective problem with 2 qubits. 
\begin{lemma}\label{lemma:sep_suff}
    Suppose $\rho$ is a state of $m = 4$ fermionic modes, with sub-system $A$ with modes 1 and 2 and sub-system $B$ with modes $3$ and $4$, and suppose the 2-qubit state $\sigma = (Q_{A, e} Q_{B, o})\rho (Q_{A, e}^\dagger Q_{B, o}^\dagger) $, where
    \begin{align*}
    Q_{A, e} = \ket{0_A}\!\bra{\textnormal{vac}} + \ket{1_A}\!\bra{\textnormal{vac}}a_1 a_2, Q_{B,o} = \ket{0_B}\!\bra{\textnormal{vac}}a_3 + \ket{1_B}\!\bra{\textnormal{vac}}a_4,
    \end{align*}
    is entangled, then $\rho$ is not separable with respect to even observables.
\end{lemma}
\begin{proof}
    This lemma follows by contradiction---let us assume that $\rho$ is separable with respect to even observables. Now, for any two operators $O_A, O_B \in \mathbb{C}^{2\times 2}$, consider the even observables
    \begin{align}\label{eq:O_A_O_B}
        O_A^{(+)} = Q_{A, e}^\dagger O_A Q_{A, e} \text{ and }O_B^{(+)} = Q_{B, o}^\dagger O_B Q_{B, o}.
    \end{align}
    Note that $O_A^{(+)}$ acts on the fermionic modes in $A$ and $O_B^{(+)}$ acts on the fermionic modes in $B$. We note also that
    \begin{align}
    \textnormal{Tr}(O_A^{(+)}O_B^{(+)}\rho) = \textnormal{Tr}(O_A O_B \sigma).
    \end{align}
    Now, since $\rho$ is separable with respect to even observables by assumption, it follows that $\exists \rho_{A, x}, \rho_{B, x}$ and a probability measure $\mu$ such that
    \begin{align}
    \text{Tr}(O_A^{(+)}O_B^{(+)}\rho) = \int \text{Tr}(O_A^{(+)}\rho_{A, x}) \text{Tr}(O_B^{(+)}\rho_{B, x}) d\mu (x),
    \end{align}
    and consequently, using Eq.~(\ref{eq:O_A_O_B}), we find that
    \begin{align}
    \text{Tr}(O_A O_B \sigma) = \int \textnormal{Tr}(O_A \sigma_{A, x}) \textnormal{Tr}(O_B \sigma_{B, x}) d\mu (x),
    \end{align}
    where $\sigma_{A, x} = Q_{A, e} \rho_{A, x} Q_{A, e}^\dagger$ and $\sigma_{B, x} = Q_{B, o} \rho_{B, x} Q_{B, o}^\dagger$. This would imply that $\sigma$ is separable and therefore, by contradiction, we conclude that $\rho$ cannot be separable even with respect to even observables.
\end{proof}

\begin{proposition}\label{prop:non_sep_fermion_even_obs}
    Consider a system with $m = 4$ fermionic modes, with the sub-system $A$ with modes 1 and 2 and sub-system $B$ with modes 3 and 4, whose density matrix $\rho(t)$ satisfies the Lindblad master equation 
    \begin{align*}
    \frac{d}{dt}\rho(t) = -i[J(a_2 a_3 + a_3^\dagger a_2^\dagger), \rho(t)] +  \sum_{j = 1}^4 \big(\kappa_1 \mathcal{D}_{a_j}+\kappa_2 \mathcal{D}_{a_j^\dagger}\big) \rho(t).
    \end{align*}
    Then, $\forall \kappa_1, \kappa_2 > 0$, $\exists \rho(0)$ which is separable with respect to even observables (Definition \ref{def:ferm_sep_even_obs}) such that $\rho(t)$ cannot be separable for all $t \geq 0$.
\end{proposition}
\begin{proof}
    Throughout this proof, ``separability" refers to separability with respect to even observables (Definition \ref{def:ferm_sep_even_obs}). We choose $\rho(0) = \ket{\psi(0)}\!\bra{\psi(0)}$, where
    \begin{align}
    \ket{\psi(0)} = \frac{1}{\sqrt{2}}\big(a_1^\dagger\ket{\text{vac}} + a_4^\dagger \ket{\text{vac}}\big).
    \end{align}
    It can be noted that $\rho(0)$ is separable with respect to even observables since, for any even observables $O_A^{(+)}$ on $A$ and $O_B^{(+)}$ on $B$,
    \begin{align}
    \textnormal{Tr}(\rho(0) O_A^{(+)}O_B^{(+)}) = \frac{1}{2}\bra{\psi_A^{(1)}} O_A^{(+)} \ket{\psi_A^{(1)}}\bra{\psi_B^{(1)}} O_B^{(+)}\ket{\psi_B^{(1)}} + \frac{1}{2}\bra{\psi_A^{(2)}} O_A^{(+)}\ket{\psi_A^{(2)}} \bra{\psi_B^{(2)}} O_B^{(+)} \ket{\psi_B^{(2)}},
    \end{align}
    where $\ket{\psi_A^{(1)}} = a_1^\dagger \ket{\text{vac}}, \ket{\psi_B^{(1)}} = \ket{\text{vac}}, \ket{\psi_A^{(2)}} = \ket{\text{vac}} $, and $\ket{\psi_B^{(2)}} = a_4^\dagger \ket{\text{vac}}$. Again, to show that $\rho(t)$ is not separable for all $t > 0$, it is enough to show that there isn't a separable state $\sigma(t)$ such that $\norm{\rho(t) - \sigma(t)}_1 \leq O(t^2)$ as $t \to 0$. To show this, we consider a first-order expansion of $\rho(t)$:
    \begin{align}
    \rho(t) &= \rho(0) + t\mathcal{L}\rho(0) + O(t^2) \nonumber \\
    &=\ket{\psi(t)}\!\bra{\psi(t)} +  t \sum_{j = 1}^4 \big(\kappa_1\mathcal{D}_{a_j}(\ket{\psi(0)}\!\bra{\psi(0)}) + \kappa_2\mathcal{D}_{a_j^\dagger}(\ket{\psi(0)}\!\bra{\psi(0)})\big) + O(t^2),
    \end{align}
    where $\ket{\psi(t)} = (a_1^\dagger + a_4^\dagger - it a^\dagger_3 a^\dagger_2 a_1^\dagger -i t a_3^\dagger a_2^\dagger a_4^\dagger)\ket{\text{vac}}/\sqrt{2}$. We can now compute the state $\sigma(t) = (Q_{A, e}Q_{B, o})\rho(t) (Q_{B, o}^\dagger Q_{A, e}^\dagger) $ defined in Lemma \ref{lemma:sep_suff}, which effectively amounts to projecting $\rho(t)$ on the subspace spanned by $\{a_3^\dagger\ket{\text{vac}}, a_4^\dagger \ket{\text{vac}}, a_3^\dagger a_2^\dagger a_1^\dagger \ket{\text{vac}}, a_4^\dagger a_2^\dagger a_1^\dagger \ket{\text{vac}}\}$ and identifying $a_3^\dagger\ket{\text{vac}} \to \ket{0_A, 0_B}, a_4^\dagger \ket{\text{vac}} \to \ket{0_A, 1_B}, a_3^\dagger a_2^\dagger a_1^\dagger \ket{\text{vac}} \to \ket{1_A, 0_B}, a_4^\dagger a_2^\dagger a_1^\dagger \ket{\text{vac}} \to \ket{1_A, 1_B}$:
    \begin{align}
        \sigma = ((1 - (\kappa_1 + \kappa_2) t)\ket{0_A, 1_B} - i t \ket{1_A, 0_B})((1 - (\kappa_1+\kappa_2) t)\bra{0_A, 1_B} + it \bra{1_A, 0_B}) + O(t^2),
    \end{align}
It is easy to see that $\forall \kappa_1, \kappa_2 > 0$, $\sigma(t)$ (as a 2-qubit state), does not admit an $O(t^2)$ separable approximation for sufficiently small $t$. Consequently, from Lemma \ref{lemma:sep_suff}, we find that $\rho(t)$ (as a 4-mode fermionic state) does not admit an $O(t^2)$ separable approximation, thus proving the lemma.
\end{proof}
\noindent To analyze the timescales at which the state $\rho(t)$ becomes non-separable relative to even observables, in Fig.~\ref{fig:fermionic_even_obs_timescale}, we numerically simulate the four-fermionic-mode model from Proposition \ref{prop:non_sep_fermion_even_obs} with $\kappa_1=\kappa_2=\kappa$ and compute $\rho(t)$. To quantify the extent to which $\rho(t)$ is non-separable, we first construct the effective two-qubit state $\sigma(t) = Q_{A, e} Q_{B, o} \rho(t) Q_{B, o}^\dagger Q_{A, e}^\dagger$ from $\rho(t)$ defined in Lemma \ref{lemma:sep_suff} and then compute the minimum eigenvalue of its partial transpose. The negative of this minimum eigenvalue is the entanglement measure shown in in Fig.~\ref{fig:fermionic_even_obs_timescale}(a). We find that, while at long times $\sigma$ has a non-negative partial transpose and is thus separable, at short times $\sigma$ is entangled irrespective of how large $\kappa$ is. Furthermore, the minimum eigenvalue of the partial transpose of $\sigma$ is attained at $t^* \sim 1/ \kappa$ [Fig.~\ref{fig:fermionic_even_obs_timescale}(b)], which sets the time-scale at which non-separability with respect to even observables in this system is developed.

\begin{figure}
    \includegraphics[width=0.7\textwidth]{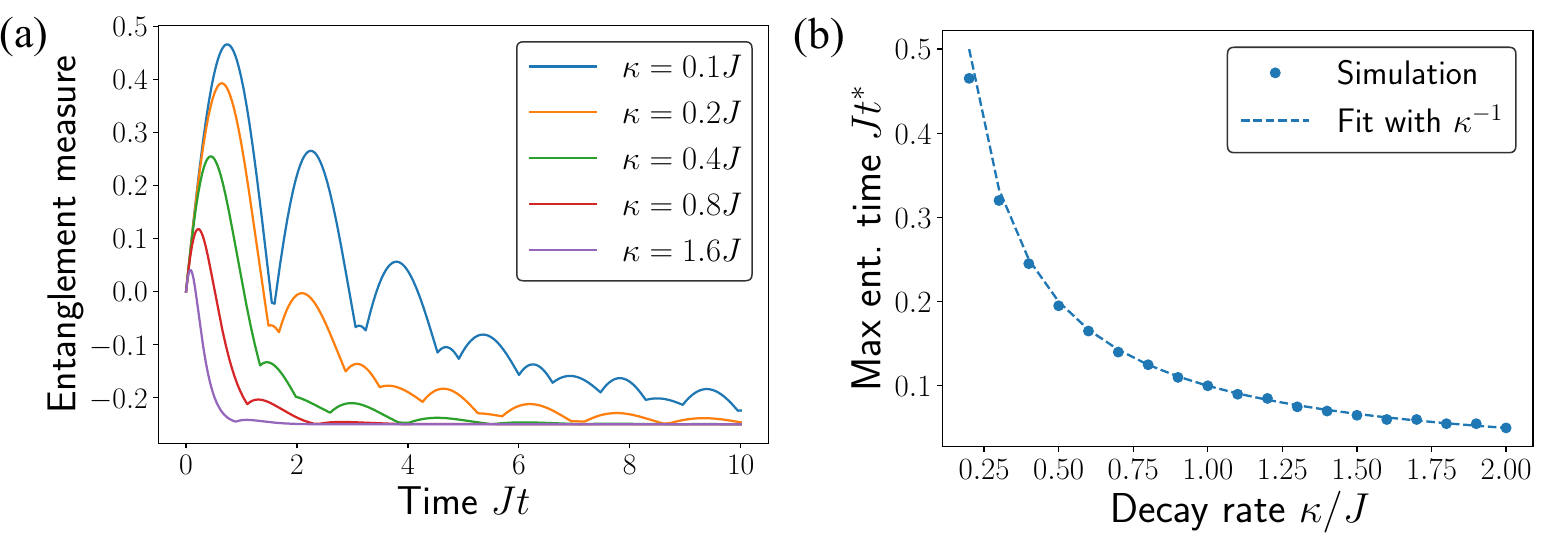}
    \caption{Numerical simulation of the four-mode fermionic model considered in Proposition \ref{prop:non_sep_fermion_even_obs} with both particle loss rate $\kappa_1$ and particle gain rate $\kappa_2$ set to $\kappa$ and $\ket{\psi(0)} = (a_1^\dagger + a_4^\dagger) \ket{\text{vac}}$. (a) Time evolution of the entanglement measure computed by first computing the two-qubit state $\sigma(t)$ corresponding to the 4-mode fermionic state $\rho(t)$ from Lemma \ref{lemma:sep_suff} and then computing (the negative) of the minimum eigenvalue of its partial transpose. As per Lemma \ref{lemma:sep_suff}, this quantifies non-separability with respect to all observables for 2-mode fermionic systems. (b) The time $t^*$ at which the entanglement measure computed in (a) is maximum as a function of $\kappa$ --- we see that entanglement is developed at time-scales $t^* \sim 1 / \kappa$ no matter how large $\kappa$ is.}\label{fig:fermionic_even_obs_timescale}
\end{figure}

\bibliographystyle{apsrev4-1}
\bibliography{references.bib}